\documentclass[numbers,authoryear, 11pt]{article}
\usepackage[T1]{fontenc}
\usepackage[utf8]{inputenc}
\usepackage[resetlabels]{multibib}
\newcites{oa}{References}
\usepackage[round]{natbib}
\usepackage{fancyhdr}
\usepackage{palatino}
\usepackage{setspace}
\usepackage[english]{babel}
\usepackage{a4wide}
\usepackage[pdftex]{graphicx}
\usepackage{pdfpages}
\usepackage[top=3cm, bottom=3cm, left=3cm, right=3cm]{geometry}
\usepackage{comment}
\usepackage{changepage}
\usepackage{graphicx}
\usepackage{paralist}
\usepackage{enumitem,kantlipsum}
\usepackage{listings}
\usepackage{nth}
\usepackage[english, ruled,vlined,resetcount]{algorithm2e}

\usepackage{algorithm2e}
\SetKwComment{Comment}{/* }{ */}
\SetAlgoCaptionSeparator{}
\usepackage{amsmath, amsthm, mathtools, amsfonts,mathrsfs, upgreek, mathabx, wasysym}
\usepackage{adjustbox}
\usepackage{setspace}
\usepackage{here}
\usepackage{rotating}
\usepackage{multirow}
\usepackage{booktabs}
\usepackage{lettrine}
\usepackage{dsfont}
\usepackage{textcomp}
\usepackage[tikz]{bclogo}
\usepackage{fancybox}
\usepackage{pifont}
\usepackage{color, xcolor}
\usepackage{bbm}
\usepackage{multicol}
\usepackage{footnote}
\usepackage{supertabular, longtable, caption, afterpage, tabularx, threeparttable}
\usepackage{colortbl}
\usepackage{dcolumn}
\newcolumntype{d}[1]{D..{#1}} 
\usepackage{siunitx}
\usepackage{lscape}

\usepackage[colorlinks=true,urlcolor= blue,citecolor=blue]{hyperref}
\usepackage{titling}
\usepackage{xpatch}
\usepackage{titlesec}
\titleformat{\section}{\large\bfseries}{\thesection}{1em}{}
\titleformat{\paragraph}{\normalfont\normalsize\bfseries}{\theparagraph}{0.7em}{}
\titlespacing*{\paragraph}{0pt}{1ex plus 1ex minus .2ex}{1pt}
\usepackage{ragged2e}
\makeatletter
\newtheorem{theorem}{Theorem}[section]
\newtheorem{proposition}{Proposition}[section]
\newtheorem{lemma}{Lemma}[section]

\newcommand{\neutralize}[1]{\expandafter\let\csname c@#1\endcsname\count@}
\newtheorem{assumption}{Assumption}[section]

\newtheorem{corollary}{Corollary}[section]
\newtheorem{example}{Example}[section]

\let\endtitlepage\relax
\newenvironment{mytitlepage}%
  {\begin{titlepage}\def\@thanks{}}%
  {\end{titlepage}}
\xpatchcmd\titlepage{\setcounter{page}\@ne}{}{}{}
\xpatchcmd\endtitlepage{\setcounter{page}\@ne}{}{}{}
\newcommand{\ostar}{\mathbin{\mathpalette\make@circled\star}}
\newcommand{\make@circled}[2]{%
  \ooalign{$\m@th#1\smallbigcirc{#1}$\cr\hidewidth$\m@th#1#2$\hidewidth\cr}%
}
\newcommand{\smallbigcirc}[1]{%
  \vcenter{\hbox{\scalebox{0.77778}{$\m@th#1\bigcirc$}}}%
}
\DeclareMathOperator{\plim}{plim}

\DeclareMathOperator{\diag}{diag}
\DeclareMathOperator{\argmax}{argmax}
\definecolor{colari}{rgb}{0.7, 0, 0.7}
\definecolor{coland}{rgb}{0, 0.7, 0.4}

\makeatletter
\usepackage{authblk}
\title{Inference for Two-Stage Extremum Estimators\footnote[1]{For comments and suggestions, we are grateful to Arnaud Dufays, Ulrich Hounyo, Mathieu Marcoux, Antoine Djogbenou, Frank Windmeijer, Xiaohong Chen, Jean-Marie Dufour, Jad Beyhum, Prosper Dovonon, Désiré Kédagni, Pamela Giustinelli and Florian Pelgrin. We also thank the participants of the EDHEX Business School econometric seminar, the CIREQ econometric seminar, the 58th Annual Meetings of the CEA, and the 2024 conference of IAAE.\\
Replication codes for the results from this research are available at \url{https://github.com/ahoundetoungan/InferenceTSE}.}}
\author[a]{Aristide Houndetoungan\footnote[2]{Corresponding author \--
Email addresses: \href{mailto:aristide.houndetoungan@cyu.fr}{aristide.houndetoungan@cyu.fr} (A. Houndetoungan),\\ \href{mailto:abdoulhaki.maoude@concordia.ca}{abdoulhaki.maoude@concordia.ca} (A. H. Maoude)}}
\author[b]{Abdoul Haki Maoude}
\affil[a]{\normalsize\emph{CY Cergy Paris University}\vspace{-2pt}}
\affil[b]{\normalsize\emph{Concordia University}\vspace{-2pt}}
\date{\normalsize October 2024}

\begin{document}
\setlength{\abovedisplayskip}{4pt}
\setlength{\belowdisplayskip}{4pt}
\begin{mytitlepage}
\maketitle

\vspace{-0.8cm}
\begin{abstract}
\noindent 
{\linespread{1.2}\selectfont
We present a simulation-based inference approach for two-stage estimators, focusing on extremum estimators in the second stage. We accommodate a broad range of first-stage estimators, including extremum estimators, high-dimensional estimators, and other types of estimators such as 
Bayesian estimators. The key contribution of our approach lies in its ability to estimate the asymptotic distribution of two-stage estimators, even when the distributions of both the first- and second-stage estimators are non-normal and when the second-stage estimator's bias, scaled by the square root of the sample size, does not vanish asymptotically.  This enables reliable inference in situations where standard methods fail. Additionally, we propose a debiased estimator, based on the mean of the estimated distribution function, which exhibits improved finite sample properties. Unlike resampling methods, our approach avoids the need for multiple calculations of the two-stage estimator. We illustrate the effectiveness of our method in an empirical application on peer effects in adolescent fast-food consumption, where we address the issue of biased instrumental variable estimates resulting from many weak instruments.

\vspace{0.5cm}
\noindent \textbf{Keywords}: Hypothesis Testing, Two-stage Estimators,  Semiparametric and Nonparametric Methods, Simulation Methods, High-Dimensional Asymptotics  

\vspace{0.2cm}
\noindent \textbf{JEL Classification}: C12, C13, C14, C15, C55.
}
\end{abstract}
\end{mytitlepage}

\newpage

\section{Introduction}
Two-stage estimation approaches are widely used to address challenges such as endogeneity, selection bias, non-identification, missing data, and high dimensionality \cite[e.g.,][]{hirano2003efficient, jofre2003estimation, newey2003instrumental, 
 freyberger2022identification, chernozhukov2022locally, ichimura2022influence, boucher2020estimating}. 
These methods involve estimating a function or parameter in the first stage (FS), followed by incorporating the estimator into a second-stage (SS) model to estimate another parameter, denoted as $\boldsymbol{\theta}_0$. The estimator of $\boldsymbol{\theta}_0$ is referred to as a \textit{two-stage} (TS) or \textit{plug-in} estimator. For inference, researchers often require the asymptotic distribution of $\sqrt{n}(\boldsymbol{\hat{\theta}}_n - \boldsymbol{\theta}_0)$, where $\boldsymbol{\hat{\theta}}_n$ is the TS estimator and $n$ is the sample size. However, deriving this distribution can be complex due to the sampling error introduced in the FS.

Three important issues may arise regarding the asymptotic distribution of $\sqrt{n}(\boldsymbol{\hat \theta}_n - \boldsymbol{\theta}_0)$: 
\begin{enumerate}[topsep=0em, itemsep=-0.4em]
    \item Estimating the asymptotic variance can be challenging due to the FS sampling error.
    \item The asymptotic mean may not be zero if the FS sampling error is substantial (slow convergence rate in the FS), leading to a TS estimator with a regularization bias.
    \item The asymptotic distribution may not be normal if the asymptotic distribution of the FS estimator is not normal.
\end{enumerate} 
While the literature has often focused on the first issue \citep[e.g.,][]{newey1994large, ackerberg2012practical, chen2015sieve}, the latter two have not received much attention, despite their potential to arise in many contexts. Standard inference methods typically impose conditions for the asymptotic distribution of $\sqrt{n}(\boldsymbol{\hat \theta}_n - \boldsymbol{\theta}_0)$ to be normal with a mean of zero. If these conditions are not met, such methods may become inappropriate and lead to incorrect conclusions in hypothesis tests.

In this paper, we propose a simulation-based inference approach for TS estimators, focusing on extremum estimators in the SS. We accommodate a broad range of FS estimators, including extremum estimators, high-dimensional estimators, and other types of estimators (e.g., 
Bayesian estimators). The novelty of our approach lies in its ability to estimate the asymptotic cumulative distribution function (CDF) of $\sqrt{n}(\boldsymbol{\hat \theta}_n - \boldsymbol{\theta}_0)$, even when the distribution is non-normal and the asymptotic mean is different from zero. By estimating this mean, we demean $\sqrt{n}(\boldsymbol{\hat \theta}_n - \boldsymbol{\theta}_0)$ and introduce a debiased estimator that demonstrates improved finite sample properties. Our debiased estimator is straightforward and easily applicable to complex models. Unlike resampling methods, it avoids the need for repeated calculations of the TS estimator. 

Regularization bias can arise in various situations, particularly when the FS model is high-dimensional or when the number of observations in the FS is small relative to $n$. The slow convergence rate of the FS estimator can significantly distort the SS estimator \citep[][]{chernozhukov2017double, belloni2014high, belloni2017program, cattaneo2019two}. Moreover, when dealing with complex models, researchers may compute a Bayesian estimator, such as a posterior mean, which is then used in the SS to obtain a standard extremum estimator \citep[e.g.,][]{breza2020using, lubold2023identifying}. As noted by \cite{zellner1984bayesian}, Bayesian estimators may not be asymptotically normally distributed, potentially leading to a non-normal asymptotic distribution in the SS. Non-normal asymptotic distributions can also occur in the FS with certain frequentist estimators that are robust to outliers, such as the trimmed mean and U-statistics \citep{stigler1973asymptotic, ma2011asymptotic}. Our approach covers a broad range of models, including these contexts where standard methods are no longer applicable.

Our approach consists first of examining the \textit{conditional} asymptotic distribution of \break$\sqrt{n}(\boldsymbol{\hat{\theta}}_n - \boldsymbol{\theta}_0)$, given any realization of the FS estimator. Since we are disregarding the FS sampling error at this stage, the conditional distribution resembles that of a single-step estimator. Consequently, under certain assumptions, it will be normal with some mean and variance that depend on the FS estimator. Next, we characterize the \textit{unconditional} asymptotic CDF of $\sqrt{n}(\boldsymbol{\hat{\theta}}_n - \boldsymbol{\theta}_0)$ by using both the asymptotic normality at the SS, conditional on the FS estimator, and the asymptotic distribution of the FS estimator. We demonstrate that this CDF may not follow a normal distribution and may have a mean that is not necessarily zero.

We then combine simulations from an estimator of the distribution of the FS estimator with the asymptotic normality in the second stage, conditional on the FS estimator, to simulate the asymptotic distribution of $\sqrt{n}(\boldsymbol{\hat{\theta}}_n - \boldsymbol{\theta}_0)$. This simulation allows us to construct a sample for $\sqrt{n}(\boldsymbol{\hat{\theta}}_n - \boldsymbol{\theta}_0)$, which we use to approximate its CDF and obtain confidence intervals (CIs) for $\boldsymbol{\theta}_0$. Additionally, the average of the constructed sample estimates the asymptotic mean of $\sqrt{n}(\boldsymbol{\hat{\theta}}_n - \boldsymbol{\theta}_0)$, which helps reduce the bias of $\boldsymbol{\hat{\theta}}_n$.

Simulating from the estimated distribution of the FS estimator is a key requirement of our method and is feasible for a broad range of models. When conducting inference in the FS using certain frequentist approaches, the asymptotic distribution is known (generally normal), which allows us to draw from a normal distribution. In cases where a Bayesian estimator is used in the FS, the posterior distribution can serve as an estimator, and samples can be obtained via Gibbs sampling or the Metropolis-Hastings algorithm \citep[][]{casella1992explaining, chib1995understanding}.\footnote{An estimator of the distribution of the FS estimator can also be derived using resampling methods when standard approaches are not easily applicable. In general, our approach avoids the need for multiple computations of the estimators, except when resampling is required to obtain the asymptotic distribution in the FS.}

We evaluate the performance of our method through an extensive simulation study, covering various models, including an instrumental variable (IV) model with many weak instruments \citep{cattaneo2019two}, a Poisson model with unobserved variables, and a multivariate GARCH model \citep{gonccalves2022bootstrapping}, with a number of returns increasing in $n$. Our simulation results confirm that our method performs well on these models.

Furthermore, to demonstrate the effectiveness of our method, we revisit the empirical analysis by \cite{fortin2015peer} on peer effects in adolescent fast-food consumption. The estimation of endogenous peer effects often relies on IV methods, where instruments are constructed based on the characteristics of friends' friends \citep{bramoulle2009identification}. We address the issue of biased estimates that can arise when instruments are weak. By using the characteristics of both close and distant friends, we expand the IV set to include many weak instruments. We correct the finite sample bias of the IV estimator and provide valid inference. Our findings suggest that a one-point increase in the average friend's fast-food consumption frequency leads to a 0.23 increase in one's fast-food consumption frequency.

This paper contributes to the extensive and growing literature on inference methods for sequential estimators. Most existing studies impose conditions for the asymptotic distribution of $\sqrt{n}(\boldsymbol{\hat \theta}_n - \boldsymbol{\theta}_0)$ to be normal with a zero mean \citep{newey1984method, andrews1994asymptotics, hotz1993conditional, murphy2002estimation, boucher2020estimating, houndetoungan2023ident}. Examples include cases where both the first- and second-stage estimators are finite-dimensional extremum estimators converging at the same rate, and scenarios where the SS estimator is asymptotically invariant to infinitesimal variations in the FS estimator. We contribute to this literature by proposing a new method that relies on more flexible conditions. Additionally, we do not impose a specific class of estimators in the FS, making our approach more general. Importantly, our method yields results comparable to those obtained from classical asymptotic inference methods when applicable.

Even though the asymptotic distribution of the plug-in estimator is normal, accurately calculating its variance can be complex. Various methods have been proposed in the literature to estimate the asymptotic variance \citep[e.g., see][]{newey1984method, newey1994asymptotic, ackerberg2012practical, chen2015sieve}, and resampling methods have also been employed for this purpose \citep[e.g.,][]{gonccalves2005bootstrap, kline2012score, armstrong2014fast, honore2017poor, gonccalves2022bootstrapping}. We discuss scenarios where our simulation method can be used to estimate the asymptotic variance from the estimated CDF of $\sqrt{n}(\boldsymbol{\hat \theta}_n - \boldsymbol{\theta}_0)$. Moreover, our method not only estimates the asymptotic variance but also the asymptotic mean, thereby reducing the bias of $\boldsymbol{\hat \theta}_n$.

Our framework is also related to the literature on extremum estimators in the presence of high-dimensional nuisance parameters \citep[e.g.,][]{ai2007estimation, belloni2014high, belloni2017program, farrell2015robust, chernozhukov2015valid, mikusheva2022inference}. Although $\boldsymbol{\hat \theta}_n$ can still be consistent in this context, the limiting distribution of $\sqrt{n}(\boldsymbol{\hat \theta}_n - \boldsymbol{\theta}_0)$ may not have a zero mean. Resampling methods are often used to approximate the regularization bias of $\boldsymbol{\hat \theta}_n$, but the choice of method depends on the studied model. For instance, in IV approaches with many weak instruments, the standard bootstrap method fails to infer $\boldsymbol{\theta}_0$, while the Jackknife method is appropriate \citep[see][]{cattaneo2019two}. The standard bootstrap also fails when the FS involves variable selection methods such as Lasso \citep{chatterjee2011bootstrapping}. We contribute to this literature because our bias reduction technique can be applied to a broad range of models. As long as simulating from the distribution of the FS estimator is possible, the asymptotic mean of $\sqrt{n}(\boldsymbol{\hat \theta}_n - \boldsymbol{\theta}_0)$ can be estimated to reduce the bias of $\boldsymbol{\hat \theta}_n$.

The regularization bias of extremum estimators in the presence of high-dimensional nuisance parameters can also be addressed using the double or debiased machine learning (DML) approach \citep{chernozhukov2017double, chernozhukov2018double, chernozhukov2022locally}. This method combines Neyman-orthogonal moments (or scores) with cross-fitting to produce an estimator whose bias converges to zero faster than $1/\sqrt{n}$. \cite{chernozhukov2017double} present this approach in the context of several models, including partially linear regression models and interactive regression models. However, obtaining orthogonal moments can be challenging in some cases, particularly for nonlinear models or those with weakly dependent data. Our approach differs from the DML method in that it does not require modifying moment functions or obtaining orthogonal scores. Instead, we retain the standard estimation procedure and apply bias reduction post-estimation. We illustrate the bias-reducing performance of our method using a multivariate GARCH model, where obtaining an orthogonal score can be challenging. 

\paragraph{Plan of the Paper}
The remainder of the paper is organized as follows. In Section \ref{sec:condMestim}, we present our framework. Section \ref{sec:overview} provides an overview of our approach using a leading example. In Section \ref{sec:asymp}, we present our main results. Section \ref{sec:simu} provides a simulation study to assess the finite sample performance of our approach. In Section \ref{sec:appli}, we present an empirical analysis with peer effects.  Section \ref{sec:conc} concludes the paper.

\paragraph{Notation} 
The symbols  $\mathbb{E}$ and $\mathbb{V}$ denote expectation and variance, respectively. $\lVert . \rVert$ is the $\ell_2$-norm. $\partial_{x}$ is the derivative with respect to $x$. The symbol $\mathbbm{1}\{\cdot\}$ is the indicator function. If $\boldsymbol a = (a_1,~ \dots, ~ a_d)^{\prime}$, $\boldsymbol b=(b_1,~ \dots, ~ b_d)^{\prime}$ are vectors in $\mathbb{R}^d$, then $\boldsymbol a \preceq \boldsymbol b$ means that $a_k \leq b_k$ for all $k$. $\boldsymbol{I}_{d}$ is the $d$-dimensional identity matrix. We use $\lim$ and $\plim$ to denote the standard limit and the limit in probability as $n$ grows to infinity, respectively. For a positive definite matrix $\mathbf{M}$, we use $\mathbf{M}^{1/2}$ to denote its Cholesky decomposition and $\mathbf{M}^{-1/2}$ to denote the Cholesky decomposition of its inverse.

\section{Framework} \label{sec:condMestim}
This section introduces the class of plug-in estimators that are studied in this paper. For expositional ease, we consider the case of M-estimators in the second stage (SS). However, our findings can be extended to any extremum estimator since inference methods for extremum estimators are similar, regardless of whether it is an M-estimator, GMM estimator, or MD estimator \cite[see][]{amemiya1985advanced}.
Moreover, due to this similarity, we interchangeably use the terms M-estimator and extremum estimator in the paper. 

In the SS, we assume that the practitioner maximizes an objective function given by 
\begin{equation}\textstyle Q_n(\boldsymbol{\theta}, ~\mathbf{y}_n,~ \mathbf{X}_n, ~\mathbf{\hat B}_n) = \dfrac{1}{n}\sum_{i = 1}^nq(\boldsymbol{\theta},~ y_{i}, ~\boldsymbol{x}_{i}, ~\boldsymbol{\hat\beta}_{n,i}),\label{eq:objective}\end{equation} 
where $\mathbf{y}_n = (y_{1},~\dots,~y_{n})^{\prime}$, $\mathbf{X}_n = (\boldsymbol{x}_{1},~\dots,~\boldsymbol{x}_{n})^{\prime}$, $\mathbf{\hat B}_n = (\boldsymbol{\hat\beta}_{n,1},~\dots,~\boldsymbol{\hat\beta}_{n,n})^{\prime}$, and $q$ is a known function. In Equation \eqref{eq:objective}, $y_{i}$ and $\boldsymbol{x}_{i}$ are observed variables for the $i$-th unit in the sample (e.g., $y_{i}$ is a dependent variable and $\boldsymbol{x}_{i}$ are explanatory variables). $\boldsymbol{\hat\beta}_{n,1},~\dots,~\boldsymbol{\hat\beta}_{n,n}$ are estimators from some first stage (FS) regression (e.g., prediction of some variable in a preliminary regression). The estimator $\boldsymbol{\hat\beta}_{n,i}$ may be a scalar or finite-dimensional vector. The subscript $i$ may also refer to time in time-series models. We will refer to $\mathbf{\hat B}_n$ as the FS estimator. Importantly, we do not require $\boldsymbol{\hat\beta}_{n,i}$ to originate from an extremum estimation, or to have a particular asymptotic distribution (like the normal distribution). However, we assume that $\boldsymbol{\hat\beta}_{n,i}$ uniformly converges in probability to some $\boldsymbol{\beta}_{0,i}$, the true value of the parameter that it is designed to estimate (see Assumption \ref{ass:converge:beta} below).

Let $\boldsymbol{\hat{\theta}}_n$ be the estimator that maximizes the objective function \eqref{eq:objective}. $\boldsymbol{\hat{\theta}}_n$ is called \textit{two-stage} (TS) or \textit{plug-in} estimator. We denote by $\boldsymbol{\theta}_0$ the true value of the parameter $\boldsymbol{\theta}$; i.e., the value taken by $\boldsymbol{\theta}$ in the data-generating process (DGP). 

Special cases within our framework arise when $\boldsymbol{\beta}_{0,i} = f(\boldsymbol{z}_i, \boldsymbol{\gamma}_0)$ for some function $f$, where $\boldsymbol{z}_i$ is a control variable that may overlap components of $\boldsymbol{x}_i$, and $\boldsymbol{\gamma}_0$ is a parameter. In this case, we have $\boldsymbol{\hat\beta}_{n,i} = f(\boldsymbol{z}_i, \boldsymbol{\hat\gamma}_n)$, where $\boldsymbol{\hat\gamma}_n$ is an estimator of $\boldsymbol{\gamma}_0$. An example of this situation is the instrumental variable (IV) approach with $\boldsymbol{z}_i$ being the instrument and $\boldsymbol{\hat\beta}_{n,i}$ is the predicted value of the endogenous variable to be plugged into the second stage \citep{cattaneo2019two}. The function $f$ may also be constant; that is, $\boldsymbol{\beta}_{0,i} = \boldsymbol{\beta}_{0}$ for any $i$, where $\boldsymbol{\beta}_{0}$ is a finite-dimensional vector to be estimated in the first stage \citep[][]{murphy2002estimation}. In the case of semiparametric or nonparametric specification, we can have  $\boldsymbol{\beta}_{0,i} = f_n(\boldsymbol{z}_i, \boldsymbol{\gamma}_{0,n})$, where the specification of the function $f_n$ and the dimension of the parameter $\boldsymbol{\gamma}_{0,n}$ depends on the sample size $n$. Examples of estimators in this situation are power series, splines, and Fourier series approximations \citep[see][]{belloni2015some}.

Let $\mathcal{Y}\subseteq \mathbb{R}^{K_{y}}$, $\mathcal{X}\subseteq \mathbb{R}^{K_{x}}$, and $\mathcal{B} \subseteq \mathbb{R}^{K_{\beta}}$ be the supports of $y$, $\boldsymbol{x}$, and $\boldsymbol{\beta}_{0,i}$, respectively, where $K_{y}$, $K_{x}$, and $K_{\beta}$ are the corresponding dimensions. Let also $\boldsymbol{\Theta} \subset \mathbb{R}^{K_{\theta}}$ be the space of $\boldsymbol{\theta}_0$, where $K_{\theta}$ is the dimension of $\boldsymbol{\theta}_0$. 
We introduce the following assumptions.

\begin{assumption}[First-Stage]
    $\boldsymbol{\hat{\beta}}_{n,i}- \boldsymbol{\beta}_{0,i}$ converges in probability to zero, uniformly in $i$, in the sense that: $\displaystyle\max_i  \textstyle\lVert\boldsymbol{\hat{\beta}}_{n,i} - \boldsymbol{\beta}_{0,i} \rVert = o_p(1)$.\label{ass:converge:beta}
\end{assumption}

\begin{assumption}[Regularity Conditions]\label{ass:Qfonction}\hfill\\ \begin{inparaenum}
[(i)]
\item For all $\boldsymbol{\theta}$, $q(\boldsymbol{\theta}, ~y, ~\boldsymbol{x}, ~\boldsymbol{b})$ is a measurable function of $(y, ~\boldsymbol{x}^{\prime}, ~\boldsymbol{b}^{\prime})^{\prime}$ in the space $\mathcal{Y} \times \mathcal{X} \times \mathcal{B}$. \label{ass:Qfonction:measure} \\
\item For all $(y, ~\boldsymbol{x}^{\prime}, ~\boldsymbol{b}^{\prime})^{\prime} \in \mathcal{Y} \times \mathcal{X} \times \mathcal{B}$, $q(\boldsymbol{\theta}, ~y, ~\boldsymbol{x}, ~\boldsymbol{b})$ is twice continuously differentiable in $\boldsymbol{\theta}$. \label{ass:Qfonction:diff} \end{inparaenum}
\end{assumption}

\noindent Assumption \ref{ass:converge:beta} is a common requirement when dealing with FS estimators that may be infinite-dimensional \cite[see][]{chen2003estimation, ichimura2010characterization}. The condition will hold in many applications. For the case where $\boldsymbol{\beta}_{0,i} = f(\boldsymbol{z}_i, \boldsymbol{\gamma}_0)$, Assumption \ref{ass:converge:beta} requires $\boldsymbol{\hat\gamma}_n$ to be a consistent estimator and $f(\boldsymbol{z}_i, \boldsymbol{\gamma})$ to be continuously differentiable in $\boldsymbol{\gamma}$, with bounded derivative uniformly in $i$.\footnote{The result follows from the mean value theorem: $\boldsymbol{\hat{\beta}}_{n,i} - \boldsymbol{\beta}_{0,i} = \partial_{\boldsymbol{\gamma}^{\prime}}f(\boldsymbol{z}_i, \boldsymbol{\gamma}^+)(\boldsymbol{\hat\gamma}_n - \boldsymbol{\gamma}_0)$, for some $\boldsymbol{\gamma}^+$ that lies between $\boldsymbol{\hat\gamma}_n$ and $\boldsymbol{\gamma}_0$.} For nonparametric sieve estimators, lower-level conditions can be imposed to satisfy Assumption \ref{ass:converge:beta} \citep[see][]{belloni2015some}. Certain of these conditions are discussed by \cite{cattaneo2019two} in their online appendix. Assumption \ref{ass:converge:beta} also holds in the case where $\boldsymbol{\beta}_{0,i}$ represents fixed effects from another model \citep[e.g.,][]{dzemski2019empirical, yan2019statistical}.  
Assumption \ref{ass:Qfonction} sets regularity conditions on the objective function's behavior. These conditions are generally imposed for classical M-estimators and do not involve the FS estimator \citep[see][]{amemiya1985advanced}. 

Despite the risk of regularization bias, $\boldsymbol{\hat{\theta}}_n$ generally converges to $\boldsymbol{\theta}_0$ in probability, even when the FS estimator is high-dimensional \citep[][]{chen2003estimation, cattaneo2019two}. The regularization bias specifically affects $\sqrt{n}(\boldsymbol{\hat \theta}_n - \boldsymbol{\theta}_0)$, which may not have a zero mean asymptotically due to the $\sqrt{n}$ factor. We acknowledge the consistency of $\boldsymbol{\hat{\theta}}_n$ as a high-level assumption. 
\begin{assumption}[Consistency] \label{ass:consistent}
    $\boldsymbol{\hat \theta}_n$ is a consistent estimator of $\boldsymbol{\theta}_0$.
\end{assumption}
\noindent The proof of this consistency is context-dependent, and the required conditions may vary. In Online Appendix (OA) \ref{sm:cons}, we discuss primitive conditions for Assumption \ref{ass:consistent} by adapting Theorem 4.1.1 of \cite{amemiya1985advanced} to our framework.

\section{Overview of our Approach}\label{sec:overview}
Before presenting the theoretical framework underlying our approach and the formal results, this section provides an overview through an illustrative example involving a latent variable model. In this example, $\sqrt{n}(\boldsymbol{\hat \theta}_n - \boldsymbol{\theta}_0)$ asymptotically follows a normal distribution with a mean of zero; this result can be established using standard methods \citep[][]{newey1994large}. This section demonstrates how our method can be effectively applied to a simple case before extending it to more complex scenarios.

\begin{example}[Latent variable model\label{eg:latent}] We consider the following model: 
$$y_i = \theta_0 \beta_{0,i} + \varepsilon_i, \quad \beta_{0,i} = f(\boldsymbol{z}_i, \boldsymbol{\gamma}_0) = \boldsymbol{z}_i^{\prime}\boldsymbol{\gamma}_0, \quad d_i = \mathbbm{1}\{\beta_{0,i} > v_i\},  \quad v_i \sim \text{Uniform}(0, ~1),$$ 

\noindent where $\boldsymbol{\gamma}_0$ is an unknown parameter, $\theta_0$ is the parameter of interest, $\beta_{0,i}$ is an unobserved probability, and $\varepsilon_i$'s are independent and identically distributed (i.i.d) random errors with mean zero and variance $\sigma_{0,\varepsilon}^2$.  Assume that we observe an i.i.d. sample of $(y_i, ~d_i, ~ \boldsymbol{z}_i^{\prime})^{\prime}$, for $i = 1, ~\dots, ~ n$. We can use a two-stage (TS) approach to estimate $\theta_0$. In the first stage, we estimate $\boldsymbol{\gamma}_0$ by regressing $d_i$ on $\boldsymbol{z}_i$. Let $\boldsymbol{\hat \gamma}_n$ be the ordinary least squares (OLS) estimator of $\boldsymbol{\gamma}_0$. In the second stage, we estimate $\theta_0$ using the regression of $y_i$ on $\hat\beta_{n,i} = \boldsymbol{z}_i^{\prime}\boldsymbol{\hat\gamma}_n$.
 
The objective function to be maximized in the second stage is 
$\textstyle Q_n(\theta, ~\mathbf{y}_n, ~\mathbf{\hat{B}}_n) = -\frac{1}{n}\sum_{i = 1}^n (y_i -\theta\hat\beta_{n,i})^2,$
where $\mathbf{\hat{B}}_n = (\hat\beta_{n,1}, ~\dots,~ \hat\beta_{n,n})^{\prime}$. The first-order condition of this maximization is $\frac{2}{n}\sum_{i = 1}^n( y_i - \hat{\theta}_n\hat\beta_{n,i})\hat\beta_{n,i} = 0$, where $\hat{\theta}_n$ is the TS estimator of $\theta_0$. By the mean value theorem, this condition solves to $\sqrt{n}(\hat{\theta}_n - \theta_0) = \hat A_n^{-1} \dot q_n(\theta_0, ~\mathbf{y}_n, ~\mathbf{\hat{B}}_n)$, where 
$$
  \textstyle  \hat A_n = \frac{2}{n}\sum_{i = 1}^n\hat\beta_{n,i}^2  \quad \text{and} \quad \dot q_n(\theta_0, ~\mathbf{y}_n, ~\mathbf{\hat{B}}_n) = \frac{2}{\sqrt{n}}\sum_{i = 1}^n( y_i - \theta_0\hat\beta_{n,i})\hat\beta_{n,i}.
$$

\noindent We will refer to $\dot q_n(\theta_0, ~\mathbf{y}_n, ~\mathbf{\hat{B}}_n)$ as the influence function (IF). It is not always straightforward to apply the central limit theorem (CLT) to the IF as in the case of a single-step approach. This complexity arises because the terms $( y_i - \theta_0\hat\beta_{n,i})\hat\beta_{n,i}$ in the expression for the IF are dependent across $i$ through the FS estimator.

Assume that $A_0 = \plim \hat A_n$ exists.
For the sake of simplicity, we treat $\boldsymbol{z}_i$ as a nonstochastic variable. We define the conditional expectation and conditional variable of the IF, given $\mathbf{\hat{B}}_n$, as follows:
\begingroup
\allowdisplaybreaks
\begin{align}\label{eq:eg:EnVn}
\begin{split}
    \mathcal{E}_n &:=\textstyle \mathbb{E}(\dot q_n(\theta_0, ~\mathbf{y}_n, ~\mathbf{\hat{B}}_n)|\mathbf{\hat{B}}_n) = \frac{2}{\sqrt{n}}\theta_0\sum_{i = 1}^n( \beta_{0,i} - \hat\beta_{n,i})\hat\beta_{n,i},\\
    V_n &:=\textstyle \mathbb{V}(\dot q_n(\theta_0, ~\mathbf{y}_n, ~\mathbf{\hat{B}}_n)|\mathbf{\hat{B}}_n) = \frac{4}{n}\sigma^2_{0,\varepsilon}\sum_{i = 1}^n\hat\beta_{n,i}^2.  
\end{split}
\end{align}
\endgroup

We also define the standardized IF, conditional on $\mathbf{\hat B}_n$, by subtracting from the IF its expectation conditional on $\mathbf{\hat B}_n$ and then dividing the resulting difference by its conditional standard deviation.
$$u_n := \textstyle V_n^{-\frac{1}{2}}(\dot q_n(\theta_0, ~\mathbf{y}_n, ~\mathbf{\hat{B}}_n) - \mathcal{E}_n) = \sum_{i = 1}^n\hat a_{n,i} (y_i - \theta_0\beta_{0,i}), ~~ \text{where} ~~ \hat a_{n,i} = \sigma^{-1}_{0,\varepsilon}\hat\beta_{n,i}\big(\sum_{i = 1}^n\hat\beta_{n,i}^2\big)^{-\frac{1}{2}}.$$ 

\noindent The expectation of the standardized IF, $u_n$, is zero and the variance is one.
Conditional on $\mathbf{\hat{B}}_n$, the variables $a_{n,i}$'s are nonstochastic and $\sum_{i = 1}^n\hat a_{n,i} (y_i - \theta_0\beta_{0,i})$  is a sum of independent variables. Consequently, by a conditional central limit theorem (CLT), the conditional distribution of  $u_n$, given $\mathbf{\hat{B}}_n$, converges to $N(0, ~1)$, for almost all $\mathbf{\hat{B}}_n$.\footnote{See an example of conditional CLT in \cite{Rubshtein1996ACL}. Indeed, Lyapunov’s condition is verified if $\sum_{i = 1}^n \hat a_{n,i}^{2+\nu} \mathbb{E}(\lvert \varepsilon_i \rvert^{2 + \nu}) = o_p(1)$, for some $\nu > 0$. A similar condition is also required in the case where 
$\beta_{0,i}$ is known and $\theta_0$ is estimated using a single-step approach.} 

Given that $\sqrt{n}(\hat{\theta}_n - \theta_0) = \hat A_n^{-1} (V_n^{1/2} u_n - \mathcal{E}_n)$, we show that the unconditional asymptotic distribution of $\sqrt{n}(\hat{\theta}_n - \theta_0)$ can be approximated by the CDF of:
\begin{equation}
    \psi_{n} = A_0^{-1}(V_0^{1/2}\zeta + \mathcal{E}_{n}),\label{eq:psi}
\end{equation}
where $\zeta \sim N(0, ~1)$ and $V_0 = \plim V_n$. The term $V_0^{1/2}\zeta + \mathcal{E}_n$ represents the decomposition of the IF, with $V_0^{1/2}\zeta$ capturing the variance of the SS error term conditional on the FS estimator, and $\mathcal{E}_n$ accounting for the sampling error from the FS. This result is important because it allows for the simulation of $\psi_{n}$. For some large $\kappa$, we construct the sample: 
$$\{\hat\psi_{n,s} = \hat A_n^{-1}(\hat V_n^{1/2}\zeta_s + \hat{\mathcal{E}}_{n,s}), ~ s = 1, ~\dots,~ \kappa\},$$ 
where $\hat{V}_n = \frac{4}{n}\hat\sigma_{n,\varepsilon}^2\sum_{i = 1}^n\hat\beta_{n,i}^2$,  $\hat \sigma_{n,\varepsilon}^2$  is a consistent estimator of $\sigma^2_{0,\varepsilon}$, $\zeta_s$'s are i.i.d simulations from $N(0, ~1)$, and $\hat{\mathcal{E}}_{n,s}$'s are i.i.d simulations from an estimator of the distribution of $\mathcal{E}_{n}$.
Since the first stage is an OLS regression, the estimator of the distribution of $\boldsymbol{\hat\gamma}_n$ is a normal distribution with mean $\boldsymbol{\hat \gamma}_n$ and variance $\hat{\mathbb{V}}(\boldsymbol{\hat\gamma}_n) = (\sum_{i = 1}^n \boldsymbol{z}_i\boldsymbol{z}_i^{\prime})^{-1}(\sum_{i = 1}^n\hat{\nu}_i^2\boldsymbol{z}_i\boldsymbol{z}_i^{\prime}) (\sum_{i = 1}^n \boldsymbol{z}_i\boldsymbol{z}_i^{\prime})^{-1}$, where $\hat{\nu}_i = d_i - \boldsymbol{z}_i^{\prime}\boldsymbol{\hat \gamma}_n$. For $s=1,~\dots, ~\kappa$, let $\bar\beta_{n,i}^{(s)}  = \boldsymbol{z}_i^{\prime}\boldsymbol{\bar \gamma}_n^{(s)}$, where $\boldsymbol{\bar \gamma}_n^{(s)} \sim N(\boldsymbol{\hat \gamma}_n,~\hat{\mathbb{V}}(\boldsymbol{\hat\gamma}_n))$. Thus, $\hat{\mathcal{E}}_{n,s} = \frac{2\hat \theta_n}{\sqrt{n}}\sum_{i = 1}^n( \hat\beta_{n,i} - \bar\beta_{n,i}^{(s)})\bar\beta_{n,i}^{(s)}$.

The $2.5\%$ and $97.5\%$ quantiles of the sample $\{\hat{\theta}_n - \hat\psi_{n,s}/\sqrt{n}, ~ s = 1, ~\dots,~ \kappa\}$ are the bounds of the $95\%$ confidence interval (CI) of $\theta_0$. 

We also show that the asymptotic variance of $\sqrt{n}(\hat\theta_n - \theta_0)$ is given by $A_0^{-1}(V_0 + \lim \mathbb V(\mathcal E_n))A_0^{-1}$. A consistent estimator of this variance is: 
$$
    \hat A_n^{-1}(\hat V_n + \hat{\mathbb{V}}(\mathcal{E}_n))\hat A_n^{-1}, ~~ \text{where} ~~ \textstyle \hat{\mathbb{V}}(\mathcal{E}_n) = \dfrac{1}{\kappa - 1}\sum_{s = 1}^{\kappa}(\hat{\mathcal{E}}_{n,s} - \hat{\mathbb{E}}(\mathcal{E}_n))^2 ~~ \text{and} ~~ \hat{\mathbb{E}}(\mathcal{E}_n) = \dfrac{1}{\kappa}\sum_{s = 1}^{\kappa}\hat{\mathcal{E}}_{n,s}.$$

For more complex models, the asymptotic mean of $\sqrt{n}(\hat\theta_n - \theta_0)$ (which is the same for $\psi_n$) may not be zero. From \eqref{eq:psi}, this asymptotic mean is given by $A_0^{-1}\lim\mathbb{E}(\mathcal{E}_n)$. If  $\lim\mathbb{E}(\mathcal{E}_n)$ is not zero, the plug-in estimator can exhibit significant bias in finite samples. We address this issue by proposing the debias estimator $\theta_{n,\kappa}^{\ast} = \hat\theta_n - (\sqrt{n}\hat A_n)^{-1}\hat{\mathbb{E}}(\mathcal{E}_n)$. We demonstrate the limiting distribution of $\sqrt{n}(\theta_{n,\kappa}^{\ast} - \theta_0)$ has a zero mean.

A key ingredient of our approach lies in computing the conditional variance of the IF. 
One important simplification in this example is that $\varepsilon_i$'s are independent of the FS estimator. This is employed when computing $\mathcal{E}_n$ and $V_n$ in \eqref{eq:eg:EnVn}. In a more general context, computing the conditional moments of the IF may be challenging. We will later discuss this situation in Section \ref{sec:asymp:finite}.

\end{example}

\section{Inference for Conditional Extremum Estimators}\label{sec:asymp}
We present our main results in this section. Technical details of proofs can be found in Appendix \ref{Append:proof}. 
The first-order condition of the maximization of  \eqref{eq:objective} is $\frac{1}{n}\sum_{i = 1}^n\partial_{\boldsymbol{\theta}}q(\boldsymbol{\hat \theta}_n, ~y_{i}, ~\boldsymbol{x}_{i}, ~\boldsymbol{\hat{\beta}}_{n,i}) = 0$. By applying the mean value theorem to $\frac{1}{n}\sum_{i = 1}^n\partial_{\boldsymbol{\theta}}q(\boldsymbol{\theta}, ~y_{i}, ~\boldsymbol{x}_{i}, ~\boldsymbol{\hat{\beta}}_{n,i})$ with respect to (w.r.t) $\boldsymbol{\theta}$, we obtain: 
\begin{equation}
        \textstyle\Delta_n := \sqrt{n}(\boldsymbol{\hat{\theta}}_n - \boldsymbol{\theta}_0)  = \mathbf{A}_n^{-1}\dfrac{1}{\sqrt{n}}\sum_{i = 1}^n \boldsymbol{\dot q}_{n,i}(y_i,~\boldsymbol{\hat{\beta}}_{n,i}),\label{eq:sqrtntheta}
\end{equation}
where $\boldsymbol{\dot q}_{n,i}(y_i,~\boldsymbol{\hat{\beta}}_{n,i}) = \partial_{\boldsymbol{\theta}} q(\boldsymbol{\theta}_0, ~y_{i}, ~\boldsymbol{x}_{i}, ~\boldsymbol{\hat{\beta}}_{n,i})$ and $\mathbf{A}_n = -\frac{1}{n}\sum_{i = 1}^n \partial_{\boldsymbol{\theta}}\partial_{\boldsymbol{\theta}^\prime}q(\boldsymbol{\theta}^+_n, ~y_{i}, ~\boldsymbol{x}_{i}, ~\boldsymbol{\hat{\beta}}_{n,i})$, for some $\boldsymbol{\theta}^+_n$ that lies between $\boldsymbol{\hat \theta}_n$ and $\boldsymbol{\theta}_0$. As $\plim \boldsymbol{\hat \theta}_n = \boldsymbol{\theta}_0$, we also have $\plim \boldsymbol{\theta}^+_n = \boldsymbol{\theta}_0$. In large samples, $\mathbf{A}_n$ is assumed to be nonsingular (see Assumption \ref{ass:Anosing}  below). Let $\boldsymbol{\dot q}_{n}(\mathbf{y}_n,~\mathbf{\hat{B}}_n) = \frac{1}{\sqrt{n}} \sum_{i = 1}^n \boldsymbol{\dot q}_{n,i}(y_i,~\boldsymbol{\hat{\beta}}_{n,i})$. We will refer to $\boldsymbol{\dot q}_{n}(\mathbf{y}_n,~\mathbf{\hat{B}}_n)$ as the influence function (IF).  

In the case of a single-step estimator, the central limit theorem (CLT) implies (under regularity conditions) that the IF is asymptotically normally distributed with zero mean \cite[see][Theore 4.1.3]{amemiya1985advanced}. A crucial condition that is required by the CLT is that the dependence among the variables $\boldsymbol{\dot q}_{n,i}(y_i,~\boldsymbol{\hat{\beta}}_{n,i})$'s is "weak." Roughly speaking, if we define a certain order between the subscripts $i$'s (e.g., if  $i$ is time), the correlation between $\boldsymbol{\dot q}_{n,i}(y_i,~\boldsymbol{\hat{\beta}}_{n,i})$ and $\boldsymbol{\dot q}_{n,j}(y_j,~\boldsymbol{\hat{\beta}}_{n,j})$ must vanish at a certain rate as $\lvert i - j\rvert$ grows to infinity \citep[see][]{withers1981central, romano2000more, ekstrom2014general}. For TS estimators, the variables $\boldsymbol{\dot q}_{n,i}(y_i,~\boldsymbol{\hat{\beta}}_{n,i})$'s are dependent on each other because they all depend on the same FS estimator. Consequently, the weak dependence condition does not hold in general, even though $\boldsymbol{\hat{\beta}}_{n,i}$ uniformly converges in probability to $\boldsymbol{\beta}_{0, i}$. Without imposing additional conditions, there is no general CLT that guarantees asymptotic normality in this case. 

Our approach does not involve applying the CLT directly to the IF. Instead, we assume that the \textit{conditional} distribution of the IF, given $\mathbf{\hat{B}}_n$, is asymptotically normal. This condition is less restrictive in many cases because, conditional on $\mathbf{\hat{B}}_n$, the sampling error from the first stage is ignored, and the IF can be viewed as that of a single-step extremum estimator.

We introduce the following regularity assumptions.
\begin{assumption}[Influence Function]\label{ass:IF}\hfill\\
     \begin{inparaenum}[(i)] \item $\mu_{\nu}(\mathbf{\hat{B}}_n) = \mathbb{E}(\lVert \boldsymbol{\dot q}_{n}(\mathbf{y}_n,~\mathbf{\hat{B}}_n)\rVert^{\nu}|\mathbf{\hat{B}}_n)$ and $\mathbb{E}(\mu_{\nu}(\mathbf{\hat{B}}_n))$ exist for some $\nu>2$, where  $\mathbb{E}(\mu_{\nu}(\mathbf{\hat{B}}_n))$ is bounded. \label{ass:IF:moments} \\
     \item $\mathbf{V}_n := \mathbb{V}(\boldsymbol{\dot q}_{n}(\mathbf{y}_n,~\mathbf{\hat{B}}_n)|\mathbf{\hat{B}}_n)$ converges in probability to some nonstochastic quantity $\mathbf{V}_0$ and $\mathcal{E}_n:=\mathbb{E}(\boldsymbol{\dot q}_{n}(\mathbf{y}_n,~\mathbf{\hat{B}}_n)|\mathbf{\hat{B}}_n)$ converges in distribution to some random variable $\mathcal{E}_0$. \label{ass:IF:dist}
     \end{inparaenum}
\end{assumption}
\begin{assumption}[Hessian Matrix]\label{ass:Anosing}
    For any estimator $\boldsymbol{\theta}_n^+$ such that $\plim \boldsymbol{\theta}_n^+ = \boldsymbol{\theta}_0$, the Hessian of the objective function at $\boldsymbol{\theta}_n^+$, given by $\frac{1}{n}\sum_{i = 1}^n \partial_{\boldsymbol{\theta}}\partial_{\boldsymbol{\theta}^\prime}q(\boldsymbol{\theta}_n^+, ~y_{i}, ~\boldsymbol{x}_{i}, ~\boldsymbol{\hat{\beta}}_{n,i})$, converges in probability to a finite nonsingular matrix $\mathbf{A}_0 = \lim \mathbb{E}\big(\frac{1}{n}\sum_{i = 1}^n \partial_{\boldsymbol{\theta}}\partial_{\boldsymbol{\theta}^\prime}q(\boldsymbol{\theta}_0, ~y_{i}, ~\boldsymbol{x}_{i}, ~\boldsymbol{\beta}_{0,i})\big)$.
\end{assumption}

\noindent Condition (\ref{ass:IF:moments}) of Assumption \ref{ass:IF} imposes weak requirements on the existence of the conditional and unconditional moments of the IF.
Condition (\ref{ass:IF:dist}) will hold in general because $\mathbf{V}_n$ can be expressed as an average, whereas $\mathcal{E}_n$ is a sum of $n$ random variables scaled by $1/\sqrt{n}$. By the uniform Law of Large Numbers (LLN), if $\mathbf{V}_n$ is smooth in $\mathbf{\hat{B}}_n$, then it will converge in probability to a constant because the FS estimator is consistent (see Example \ref{eg:latent}). Additionally, in many cases, it is possible to express $\mathcal{E}_n$ as a function of $\mathbf{C}_n(\boldsymbol{\hat\gamma}_n - \boldsymbol{b}_n)$, for some sequences $\mathbf{C}_n$ and $\boldsymbol{b}_n$ and some estimator $\boldsymbol{\hat\gamma}_n$, such that $\mathbf{C}_n(\boldsymbol{\hat\gamma}_n - \boldsymbol{b}_n)$ has a limiting distribution.\footnote{In Example \ref{eg:latent}, $\textstyle\mathcal{E}_n$ can be approximated using a first-order Taylor expansion around $\boldsymbol{\gamma}_0$ as $\mathcal{E}_n \approx - (\frac{2\theta_0}{n}\sum_{i = 1}^n\boldsymbol{z}_i^{\prime}\boldsymbol{\gamma}_0\boldsymbol{z}_i^{\prime})\sqrt{n}(\boldsymbol{\hat\gamma}_n - \boldsymbol{\gamma}_0)$. Consequently, $\mathcal{E}_0$ is normally distributed with a zero mean.} Importantly, Condition (\ref{ass:IF:dist}) is flexible enough to encompass many models. Specifically, we allow for $\mathcal{E}_0$ not to be normally distributed and its mean may not be zero. 

Assumption \ref{ass:Anosing} ensures the consistency of the Hessian matrix. Under some weak regularity conditions, for instance, if $\partial_{\boldsymbol{\theta}}\partial_{\boldsymbol{\theta}^\prime}q(\boldsymbol{\theta}_0, ~y_{i}, ~\boldsymbol{x}_{i}, ~\boldsymbol{\beta}_{0,i})$ is ergodic stationary across $i$ with a finite variance, the LLN implies that $\plim\frac{1}{n}\sum_{i = 1}^n \partial_{\boldsymbol{\theta}}\partial_{\boldsymbol{\theta}^\prime}q(\boldsymbol{\theta}_0, ~y_{i}, ~\boldsymbol{x}_{i}, ~\boldsymbol{\beta}_{0,i}) = \mathbf{A}_0$. Assumption \ref{ass:Anosing} extends this convergence to the case where $\boldsymbol{\theta}_0$ and $\boldsymbol{\beta}_{0,i}$ are replaced with consistent estimators. It adapts Conditions (B) in Theorem 4.1.3 of \cite{amemiya1985advanced} to TS estimation approaches. We discuss primitive conditions for Assumption \ref{ass:Anosing} in OA \ref{sm:hessian}. These conditions require the Hessian at $\boldsymbol{\theta}_n^+$  to be smooth in $\boldsymbol{\theta}_n^+$ and $\boldsymbol{\hat{\beta}}_{n,i}$.

In the rest of this section, we first present some theoretical results on the asymptotic distribution of $\Delta_n$. Subsequently, we discuss the finite sample approximation of this distribution and introduce our debiased plug-in estimator.
 
\subsection{Asymptotic Distribution}\label{sec:asymp:theo}
We first examine the conditional distribution of the IF, given $\mathbf{\hat{B}}_n$. Indeed, treating $\mathbf{\hat{B}}_n$ as a predetermined sequence in $n$ allows us to approach the problem as in the case of single-step M-estimators. However, the IF may not have a zero mean, even asymptotically, thereby leading to a limiting distribution of $\Delta_n$ that is not centered at zero \citep[e.g., see][]{chernozhukov2018double}. Therefore, we define the \textit{standardized} IF as $\boldsymbol{u}_{n}(\mathbf{y}_n,~\mathbf{\hat{B}}_n) := \mathbf{V}_n^{-1/2}(\boldsymbol{\dot q}_{n}(\mathbf{y}_n,~\mathbf{\hat{B}}_n) - \mathcal{E}_n)$, which has a zero mean and variance $\boldsymbol{I}_{K_{\theta}}$. We impose the following assumption.

\begin{assumption}[Conditional Asymptotic Normality]
The conditional distribution of the standardized influence function $\boldsymbol{u}_{n}(\mathbf{y}_n,~\mathbf{\hat{B}}_n)$, given $\mathbf{\hat{B}}_n$,  converges to $N(0, ~\boldsymbol{I}_{K_{\theta}})$; in the sense that for all $\boldsymbol{t}\in\mathbb{R}^{K_{\theta}}$, we have $\plim \mathbb{P}(\boldsymbol{u}_{n}(\mathbf{y}_n,~\mathbf{\hat{B}}_n) \preceq \boldsymbol{t}|\mathbf{\hat{B}}_n) = \Phi(\boldsymbol{t})$, where $\Phi$ is the CDF of $N(0, ~\boldsymbol{I}_{K_{\theta}})$. \label{ass:CLT}
\end{assumption}

\noindent Conditional on $\mathbf{\hat{B}}_n$, the variable $\boldsymbol{u}_{n}(\mathbf{y}_n,~\mathbf{\hat{B}}_n)$ can be viewed as the standardized IF of a single-step estimator. Consequently, a conditional CLT may imply Assumption \ref{ass:CLT} \cite[see][Theorem 1]{Rubshtein1996ACL}.\footnote{A similar interpretation of Assumption \ref{ass:CLT} by \cite{kato2011note} is that ${\sup_{g\in LB}}\lvert \mathbb{E}[g(\boldsymbol{u}_{n}(\mathbf{y}_n,~\mathbf{\hat{B}}_n))|\mathbf{\hat{B}}_n] - {\displaystyle\varint}_{\mathbb{R}} g(t)d\Phi(\boldsymbol{t})\rvert = o_p(1)$, where $LB$ is the set of all functions on $\mathbb{R}$ with Lipschitz norm bounded by one. See also \cite{fligner1979use, van2000asymptotic} who used conditional asymptotic normality.}
For example, if $y_i$ is independent across $i$, conditional on $\mathbf{\hat{B}}_n$, then the variables $\boldsymbol{\dot q}_{n,i}(y_i,~\boldsymbol{\hat \beta}_{n,i})$'s would also be independent across $i$, conditional on $\mathbf{\hat{B}}_n$. Thus, the Lyapunov CLT or Lindeberg CLT, conditional on $\mathbf{\hat{B}}_n$, implies Assumption \ref{ass:CLT} (under similar conditions to that of single-step estimators). When $y_i$'s are dependent, we may use a more general CLT for dependent processes if the dependence is weak conditional on $\mathbf{\hat{B}}_n$. 

Assumption \ref{ass:CLT} allows us to separate the model error in the SS (conditional on the FS sampling error) from the sampling error in the FS. Importantly, since  $\boldsymbol{u}_{n}(\mathbf{y}_n,~\mathbf{\hat{B}}_n)$ is standardized, its first two conditional moments, given $\mathbf{\hat{B}}_n$, are independent of $\mathbf{\hat{B}}_n$. Consequently, even \textit{unconditionally}, $\boldsymbol{u}_{n}(\mathbf{y}_n, \mathbf{\hat{B}}_n)$ asymptotically follows a standard normal distribution. We claim and show this result in Lemma \ref{lem:CLT} in Appendix \ref{append:dist}. However, this result does not extend to the non-standardized IF, as its conditional moments depend on $\mathbf{\hat{B}}_n$. 

The following theorem establishes the asymptotic distribution of $\Delta_n$.

\begin{theorem}[Asymptotic Distribution]\label{theo:dist}
    Let
    $\boldsymbol{\psi}_n =  \mathbf{A}_0^{-1}\mathbf{V}_0^{1/2}\boldsymbol{\zeta} + \mathbf{A}_0^{-1}\mathcal{E}_n$, where $\boldsymbol{\zeta} \sim N(0, ~\boldsymbol{I}_{K_{\theta}})$. Let $F$ be the limiting distribution function of $\boldsymbol{\psi}_n$; that is, $ F(\boldsymbol{t}) = \lim\mathbb{P}(\boldsymbol{\psi}_n \preceq \boldsymbol{t})$ for all $\boldsymbol{t}\in\mathbb{R}^{K_{\theta}}$. Under Assumptions \ref{ass:converge:beta}--\ref{ass:CLT}, we have $\lim \mathbb{P}(\sqrt{n}(\boldsymbol{\hat{\theta}}_n - \boldsymbol{\theta}_0) \preceq \boldsymbol{t}) = F(\boldsymbol{t})$. 
\end{theorem}

The proof of Theorem \ref{theo:dist} is presented in Appendix \ref{append:dist}. Since $\Delta_n = \mathbf{A}_n^{-1}\mathbf{V}_n^{1/2}\boldsymbol{u}_{n}(\mathbf{y}_n,~\mathbf{\hat{B}}_n) + \mathbf{A}_n^{-1}\mathcal{E}_n$ and $\boldsymbol{u}_{n}(\mathbf{y}_n,~\mathbf{\hat{B}}_n)$ is asymptotically normally distributed, we show that $\boldsymbol{u}_{n}(\mathbf{y}_n,~\mathbf{\hat{B}}_n)$ can be substituted with $\boldsymbol{\zeta}$. It is noteworthy that this substitution is not trivial because $\boldsymbol{u}_{n}(\mathbf{y}_n,~\mathbf{\hat{B}}_n)$ and $\mathcal{E}_n$ may be correlated. However, we show that they are asymptotically independent; that is, both the asymptotic conditional distribution of $\boldsymbol{u}_{n}(\mathbf{y}_n,~\mathbf{\hat{B}}_n)$, given $\mathbf{\hat{B}}_n$ and the asymptotic unconditional distribution are the same (see Lemma \ref{lem:CLT}). In the expression of $\boldsymbol{\psi}_n$,  the sampling error from the FS is captured by the term $\mathbf{A}_0^{-1}\mathcal{E}_n$, whereas $\mathbf{A}_0^{-1}\mathbf{V}_0^{1/2}\boldsymbol{\zeta}$ captures variability in the SS.\footnote{In Appendix \ref{append:dist:p}, we extend Theorem \ref{theo:dist} to the uniform convergence. We show that $\sup_{\boldsymbol{t}\in\mathbb{R}^{K_{\theta}}}\lvert \mathbb{P}(\mathbf{V}_n^{-1/2}\mathbf{A}_0\Delta_n \preceq \boldsymbol{t}) -  G(\boldsymbol{t})\rvert = o_p(1)$, where $G$ is the limiting distribution function of $\boldsymbol{\zeta} + \mathbf{V}_n^{-1/2}\mathcal{E}_n$.}

Similarly, we can also decompose the asymptotic variance of $\Delta_n$. Let $\boldsymbol{\Sigma}_n = \mathbb{V}(\boldsymbol{\dot q}_{n}(\mathbf{y}_n,~\mathbf{\hat{B}}_n))$ be the variance of the IF and  $\boldsymbol{\Sigma}_0 = \lim \boldsymbol{\Sigma}_n$ be the limit of this variance. By Slutsky's theorem, the asymptotic variance of $\Delta_n $ is $\mathbb{V}(\textstyle\Delta_0):=\mathbf{A}_0^{-1}\boldsymbol{\Sigma}_0\mathbf{A}_0^{-1}$. This expression is similar to the asymptotic variance formula for single-step M-estimators. Yet, a notable difference here is that the sampling error from the FS estimator is incorporated into $\boldsymbol{\Sigma}_0$. By the law of iterated variances, we have $\boldsymbol{\Sigma}_n = \mathbb{E}(\mathbf{V}_n) +  \mathbb{V}(\mathcal{E}_n)$. The first term on the right-hand side (RHS) is the asymptotic variance of the IF conditional on the FS sampling error. The second term is the variance due to the FS estimation. Using this equation, we establish the following result.

\begin{theorem}[Asymptotic Variance] Under Assumptions \ref{ass:converge:beta}--\ref{ass:Anosing}, the variance of the asymptotic distribution of $\Delta_0$ is given by $\mathbb{V}(\Delta_0) = \mathbf{A}_0^{-1}(\mathbf{V}_0 + \mathbb{V}(\mathcal{E}_0))\mathbf{A}_0^{-1}$.\label{theo:variance}
\end{theorem}

Furthermore, Theorem \ref{theo:dist} implies that $\Delta_n$ is asymptotically normally distributed if $\mathcal{E}_0$ follows a normal distribution. However, $\Delta_n$ may exhibit a regularization bias that depends on the expectation of $\mathcal{E}_0$. This leads to the following result. 

\begin{corollary}[Asymptotic Normality]Under Assumptions  \ref{ass:converge:beta}--\ref{ass:CLT}, if $\mathcal{E}_0\sim N\big(\mathbb{E}(\mathcal{E}_0), ~\mathbb{V}(\mathcal{E}_0)\big)$, then $\sqrt{n}(\boldsymbol{\hat{\theta}}_n - \boldsymbol{\theta}_0)$ converges in distribution to $N\big(\mathbf{A}_0^{-1}\mathbb{E}(\mathcal{E}_0), ~\mathbf{A}_0^{-1}(\mathbf{V}_0 + \mathbb{V}(\mathcal{E}_0))\mathbf{A}_0^{-1}\big)$.\label{cor:normal} 
\end{corollary}

\noindent The regularization bias of $\Delta_n$ is given by  $\mathbf{A}_0^{-1}\mathbb{E}(\mathcal{E}_0)$. Corollary \ref{cor:normal} shares similarities with Theorem 1 of \cite{cattaneo2019two}. In the context of IV approaches with many instruments, they show that $\mathbf{V}_n^{-1/2}\mathbf{A}_0(\Delta_n - \mathbf{A}_0^{-1}\mathbb{E}(\mathcal{E}_n))$ is asymptotically normally distributed, where $\mathbf{A}_0^{-1}\mathbb{E}(\mathcal{E}_n)$ represents the bias of $\Delta_n$. Corollary \ref{cor:normal} generalizes this result to a broad class of models.

\subsection{Finite Sample Approximations}\label{sec:asymp:finite}
In this section, we discuss how to simulate the asymptotic distribution of $\boldsymbol{\psi}_n =  \mathbf{A}_0^{-1}\mathbf{V}_0^{1/2}\boldsymbol{\zeta} + \mathbf{A}_0^{-1}\mathcal{E}_n$ in finite samples. Since $\boldsymbol{\psi}_n$ and $\Delta_n$ share the same asymptotic distribution (Theorem \ref{theo:dist}), we can use the simulated distribution to infer $\boldsymbol{\theta}_0$.

\subsubsection{Simulating the Asymptotic Distribution}\label{sec:asymp:finite:simu}
We construct the sample $\mathcal{S}_{\kappa}=\{\boldsymbol{\hat \psi}_{n,s}: s = 1, ~\dots, ~ \kappa\}$ for some integer $\kappa \geq 1$, where $\boldsymbol{\hat \psi}_{n,1}$, \dots, $\boldsymbol{\hat \psi}_{n,\kappa}$ are independent variables with the same asymptotic distribution as $\boldsymbol{\psi}_n$. To obtain $\boldsymbol{\hat \psi}_{n,s}$, we replace the unknown nonstochastic quantities $\mathbf{A}_0$ and $\mathbf{V}_0$ in the expression of $\boldsymbol{\psi}_n$  with their estimators, and the random variables $\boldsymbol{\zeta}$ and $\mathcal{E}_n$ with independent draws from their (approximated) distributions. Specifically, $\boldsymbol{\hat \psi}_{n,s}$ is given by:
$$\boldsymbol{\hat \psi}_{n,s} =  \mathbf{\hat A}_n^{-1}\mathbf{\hat V}_n^{1/2}\boldsymbol{\zeta}_s + \mathbf{\hat A}_n^{-1}\hat{\mathcal{E}}_{n,s},$$
where $\mathbf{\hat A}_n$ and $\mathbf{\hat V}_n$ are respectively consistent estimators of $\mathbf A_0$ and $\mathbf{V}_0$, $\boldsymbol{\zeta}_1,~ \dots~\boldsymbol{\zeta}_{\kappa}$ are independent draws from $N(0, ~\boldsymbol{I}_{K_{\theta}})$, and $\hat{\mathcal{E}}_{n,1}$, \dots, $\hat{\mathcal{E}}_{n,\kappa}$ are independent draws from the (approximated) distribution of $\mathcal{E}_0$.  By Assumption \ref{ass:Anosing}, a consistent estimator of $\mathbf{A}_0$ is simply $\mathbf{\hat A}_n = -\frac{1}{n}\sum_{i = 1}^n \partial_{\boldsymbol{\theta}}\partial_{\boldsymbol{\theta}^\prime}q\big(\boldsymbol{\hat{\theta}}_n, ~y_{i}, ~\boldsymbol{x}_{i}, ~\boldsymbol{\hat{\beta}}_{n,i}\big)$. We will later discuss how to obtain $\mathbf{\hat V}_n$ and $\hat{\mathcal{E}}_{n,s}$.

The sample $\mathcal{S}_{\kappa}$ plays a crucial role. It can be used to construct confidence intervals for $\boldsymbol{\theta}_0$. Without loss of generality, assume that $\boldsymbol{\theta}_0$ is a scalar. Let $T_{\alpha}$ be the $\alpha$-quantile of the sample $\{\boldsymbol{\hat\theta}_n - \boldsymbol{\hat \psi}_{n,s}/\sqrt{n}: s = 1, ~\dots,~ \kappa\}$. Then, $[T_{\frac{\alpha}{2}}, ~ T_{1 - \frac{\alpha}{2}}]$ is a consistent estimator of the $(1 - \alpha)$ CI of $\boldsymbol{\theta}_0$, in the sense that $\lim_{\kappa\to\infty}\plim \mathbb{P}(\boldsymbol{\theta}_0 \in [T_{\frac{\alpha}{2}}, ~ T_{1 - \frac{\alpha}{2}}]) = 1 - \alpha$.\footnote{In practice, the integer $\kappa$ must be sufficiently large to ensure that the approximation error due to the number of simulations is negligible. Unlike resampling methods, increasing $\kappa$ does not introduce numerical issues, as we do not require many estimators of $\boldsymbol{\theta}_0$ from various samples.}

Furthermore, following Theorem \ref{theo:variance}, we can obtain a consistent estimator of the asymptotic variance of $\Delta_n$, by replacing the variance of $\mathcal{E}_0$ with the sample variance of $\hat{\mathcal{E}}_{n,s}$. This leads to the following estimator for the asymptotic variance of $\boldsymbol{\hat{\theta}}_n$:
\begin{equation}
    \mathbb{\hat V}_{n}(\boldsymbol{\hat{\theta}}_n) = \frac{\mathbf{\hat A}_n^{-1}\textstyle\boldsymbol{\hat \Sigma}_n^{\kappa}\mathbf{\hat A}_n^{-1}}{n},\label{eq:Vasy}
\end{equation}

where $\boldsymbol{\hat \Sigma}_n^{\kappa} = \hat{\mathbf{V}}_{n} + \frac{1}{\kappa - 1} \sum_{s = 1}^{\kappa} (\hat{\mathcal{E}}_{n,s} - \boldsymbol{\hat \Omega}_{n}^{\kappa}) (\hat{\mathcal{E}}_{n,s} - \boldsymbol{\hat \Omega}_{n}^{\kappa})^{\prime}$ and $\boldsymbol{\hat \Omega}_{n}^{\kappa} =  \frac{1}{\kappa}\sum_{s = 1}^{\kappa} \hat{\mathcal{E}}_{n,s}$. 

To obtain $\mathbf{\hat V}_n$ and $\hat{\mathcal{E}}_{n,s}$, we first need to compute $\mathbf{V}_n$ and $\mathcal{E}_{n}$, which are the conditional variance and conditional expectation of the IF, given $\mathbf{\hat{B}}_n$. In general, computing $\mathcal{E}_{n}$ is straightforward by replacing $\mathbf{y}_n$ in $\boldsymbol{\dot q}_{n}(\mathbf{y}_n,~\mathbf{\hat{B}}_n)$ with its specification (see Example \ref{eg:latent}). In this exercise, exogenous variables, such as $\mathbf{X}_n$, can be treated as nonstochastic, as is often done in practice. This simplification is innocuous and analogous to inferring $\boldsymbol{\theta}_0$ conditional on $\mathbf X_n$. Even if $\mathcal{E}_n$ does not have a closed-form expression, we can employ a large-sample approximation. For instance, $\mathcal{E}_n$ can be approximated by $(1/\sqrt{n})\sum_{i = 1}^n \boldsymbol{\dot q}_{n,i}(y_i,~\boldsymbol{\hat{\beta}}_{n,i})$.

Two scenarios may arise regarding the conditional variance $\mathbf{V}_n$. First, $\mathbf{\hat{B}}_n$ and $\mathbf{y}_n$ may be independent, conditional on the exogenous variables in the model. This occurs when the error terms in both stages are independent. An example is when the first stage involves predicting unobserved exogenous variables using auxiliary models that are independent of the second stage's error terms \cite[see][]{chernozhukov2018double, breza2020using, lubold2023identifying, boucher2020estimating}. In such cases, $\mathbf{V}_n$ can also be computed by substituting $\mathbf{y}_n$ with its specification. The calculations here are similar to those in a single-step approach. As for the $\mathcal{E}_{n}$, we can use a large-sample approximation if $\mathbf{V}_n$ does not have a closed form. For example, we can estimate $\mathbf{V}_n$ using a heteroskedasticity and autocorrelation consistent (HAC) covariance matrix \citep{andrews1991heteroskedasticity}.

The second scenario, which is more challenging, arises when the error terms in both stages are not independent. This renders intricate the dependence between $\boldsymbol{\dot q}_{n,i}(y_i,~\boldsymbol{\hat{\beta}}_{n,i})$'s conditional on $\mathbf{\hat{B}}_n$, making it difficult to estimate $\mathbf V_n$. This situation is typical in IV approaches and estimation methods that use $\mathbf{y}_n$ in the first stage \citep[][]{dufaysselective}.

One important context where we can address the problem includes linear IV methods. The SS of this approach consists of regressing $\mathbf y_n$ on $\mathbf{\hat X}_n$, where $\mathbf{\hat X}_n$ is the prediction of $\mathbf{X}_n$ using an OLS regression of $\mathbf{X}_n$ on an instrument set. To avoid the complication of treating the correlation between $\mathbf{\hat X}_n$ and the error term of the SS, we instead regress $\mathbf{\hat y}_n$ on $\mathbf{\hat X}_n$ in the SS, where $\mathbf{\hat y}_n$ is the prediction of $\mathbf y_n$ using an OLS regression of $\mathbf y_n$ on the same instrument set. The IF becomes $\boldsymbol{\dot q}_{n}(\mathbf{\hat y}_n,~\mathbf{\hat{B}}_n)$ and the FS estimator includes both $\mathbf{\hat y}_n$ and $\mathbf{\hat X}_n$. Conditional on $\mathbf{\hat y}_n$ and $\mathbf{\hat{B}}$, the IF is nonstochastic; i.e., $\mathbf{V}_n = 0$ and $\mathcal{E}_{n} = \boldsymbol{\dot q}_{n}(\mathbf{y}_n,~\mathbf{\hat{B}}_n)$. We conduct a simulation study with linear IV models (see DGP A and B) and related technical details are provided in OA \ref{sm:simu}.

Once we have $\mathbf{V}_n$ and $\mathcal{E}_n$, we can obtain $\mathbf{\hat V}_n$ and $\hat{\mathcal{E}}_{n,s}$. We recall that $\mathbf{\hat V}_n$ is an estimator of $\mathbf{V}_0 = \plim \mathbf{V}_n$, whereas $\hat{\mathcal{E}}_{n,s}$ are independent variables that share the asymptotic distribution of $\mathcal{E}_n$.  Since $\mathcal{E}_{n}$ and $\mathbf V_n$ are conditional moments, given $\mathbf{\hat B}_{n}$, they are functions of $\mathbf{\hat B}_{n}$. They can also depend on $\boldsymbol{\theta}_0$ and $\mathbf{B}_0 = (\boldsymbol{\beta}_{0,1},~\dots,~\boldsymbol{\beta}_{0,n})^{\prime}$ because the specification of $\mathbf y_n$ depends on $\boldsymbol{\theta}_0$ and $\mathbf B_0$. In the rest of this section, we thus use the notations $\mathbf V_n(\mathbf{\hat B}_{n}, ~\boldsymbol{\theta}_0, ~ \mathbf{B}_{0}) \equiv  \mathbf V_n$ and $\mathcal{E}_n(\mathbf{\hat B}_{n}, ~\boldsymbol{\theta}_0, ~ \mathbf{B}_{0}) \equiv \mathcal{E}_{n}$ to indicate that $\mathcal{E}_{n}$ and $\mathbf V_n$ are functions of $\mathbf{\hat B}_{n}$, $\boldsymbol{\theta}_0$, and $\mathbf{B}_0$. 

The estimator of $\mathbf V_0$ can be obtained by replacing $\boldsymbol{\theta}_0$ and $\mathbf{B}_{0}$ in $\mathbf V_n(\mathbf{\hat B}_{n}, ~\boldsymbol{\theta}_0, ~ \mathbf{B}_{0})$ with consistent estimators; i.e., $\mathbf{\hat V}_n = \mathbf V_n(\mathbf{\hat B}_{n}, ~\boldsymbol{\hat \theta}_n, ~ \mathbf{\hat B}_{n})$.
In contrast, we cannot obtain $\hat{\mathcal{E}}_{n,s}$ simply by replacing $\boldsymbol{\theta}_0$ and $\mathbf{B}_{0}$ in $\mathcal E_n(\mathbf{\hat B}_{n}, ~\boldsymbol{\theta}_0, ~ \mathbf{B}_{0})$ with consistent estimators. Indeed, what makes $\mathcal E_n(\mathbf{\hat B}_{n}, ~\boldsymbol{\theta}_0, ~ \mathbf{B}_{0})$ random is $\mathbf{\hat B}_{n}$. Consequently, an independent variable with the same asymptotic distribution as $\mathcal E_n(\mathbf{\hat B}_{n}, ~\boldsymbol{\theta}_0, ~ \mathbf{B}_{0})$ must be $\mathcal{E}_n(\mathbf{\bar B}_{n,s}, ~\boldsymbol{\theta}_0, ~ \mathbf{B}_{0})$, where $\mathbf{\bar B}_{n,s}$ and $\mathbf{\hat B}_{n}$ have the same asymptotic distribution. By replacing  $\boldsymbol{\theta}_0$ and $\mathbf{B}_{0}$ in $\mathcal{E}_n(\mathbf{\bar B}_{n,s}, ~\boldsymbol{\theta}_0, ~ \mathbf{B}_{0})$ with their estimators, we obtain $\hat{\mathcal{E}}_{n,s} = \mathcal{E}_n(\mathbf{\bar B}_{n,s}, ~\boldsymbol{\hat \theta}_n, ~ \mathbf{\hat B}_{n})$. In practice, we simulate $\mathbf{\bar B}_{n,s}$ from the estimator of the distribution of $\mathbf{\hat B}_{n}$.

Our approach requires the practitioner to possess a consistent estimator of the joint distribution of $\boldsymbol{\hat \beta}_{n,1}$, ~\dots,~ $\boldsymbol{\hat \beta}_{n,n}$. This estimator is obtained in the first stage for a broad class of models. For estimators of type $\boldsymbol{\hat\beta}_{n,i} = f(\boldsymbol{z}_i, \boldsymbol{\hat\gamma}_n)$, which encompass semiparametric methods, a simulation from an estimator of the distribution of $\mathbf{\hat{B}}_n$ is $\mathbf{\bar{B}}_{n,s} = (\boldsymbol{\bar\beta}_{n,1}^{(s)},~\dots,~\boldsymbol{\bar\beta}_{n,n}^{(s)})^{\prime}$, where  $\boldsymbol{\bar\beta}_{n,i}^{(s)} = f(\boldsymbol{z}_i, \boldsymbol{\bar\gamma}_{n}^{(s)})$ and $\boldsymbol{\bar\gamma}_{n}^{(s)}$ is simulated from an estimator of the asymptotic distribution of $\boldsymbol{\hat\gamma}_n$. From a frequentist perspective, the estimator of the distribution of $\boldsymbol{\hat \gamma}_n$ is typically derived through an asymptotic analysis (e.g., a normal distribution centered at $\boldsymbol{\hat\gamma}_n$ with some covariance matrix). In the Bayesian paradigm, we use the posterior distribution and obtain simulations by employing a Gibbs sampler or Metropolis-Hastings.

As pointed out above, $\mathbf V_n$ may not be tractable in some complex models when $\mathbf{\hat{B}}_n$ and the error term in the second stage are not independent. In these cases, we cannot construct the sample $\mathcal{S}_{\kappa}$. Instead, we apply Corollary \ref{cor:normal}; i.e., we infer $\Delta_n$ only when its asymptotic distribution is normal. However, since we can generally compute (or approximate) $\mathcal E_n$, then we can estimate the asymptotic mean of $\Delta_n$, thereby debasing $\boldsymbol{\hat \theta}_n$ (see DGP D in our simulation study). We present our debiasing technique in the next section.

\subsubsection{Bias Reduction}\label{sec:asymp:finite:bias}
Plug-in estimators may exhibit significant bias when the first-stage sampling error is substantial. This issue can arise when many covariates are included in the first-stage estimation (e.g., IV approach with many instruments) or when the number of observations in the first stage is low compared to $n$. Although $\boldsymbol{\hat\theta}_n$ can still be consistent in these cases, the limiting distribution of $\Delta_n$ may not have a zero mean \citep[e.g., see][]{belloni2014inference, belloni2017program, cattaneo2019two}.  In this section, we discuss how our method can be applied to address this issue.

Theorem \ref{theo:dist} implies that the asymptotic mean of $\Delta_n$ is $\mathbb{E}(\Delta_0)= \mathbf{A}_0^{-1}\mathbb{E}(\mathcal{E}_0)$. Our approach accommodates scenarios where $\mathbb{E}(\mathcal{E}_0)$ is not zero. Condition (\ref{ass:IF:dist}) of Assumption \ref{ass:IF} only states that $\mathcal{E}_n$ has a limiting distribution. The condition $\mathbb{E}(\mathcal{E}_0)\neq 0$ suggests that the plug-in estimator may exhibit significant finite sample bias. The good news is that we can estimate $\mathbf{A}_0$ and $\mathbb{E}(\mathcal{E}_0)$. Using these estimates, we propose a debiased estimator and establish its asymptotic distribution. We consider the estimator that is given by:
\begin{equation}
    \boldsymbol{\theta}_{n,\kappa}^{\ast} = \boldsymbol{\hat\theta}_n - \mathbf{\hat A}_n^{-1}\boldsymbol{\hat \Omega}_{n}^{\kappa}/\sqrt{n},\label{eq:thetast}
\end{equation}
where $\boldsymbol{\hat\Omega}_{n}^{\kappa} =  \frac{1}{\kappa}\sum_{s = 1}^{\kappa} \hat{\mathcal{E}}_{n,s}$ is an estimator of $\mathbb{E}(\mathcal{E}_0)$ as defined in Equation \eqref{eq:Vasy}.  The following theorem establishes the consistency of $\boldsymbol{\theta}_{n,\kappa}^{\ast}$ and its limiting distribution.

\begin{theorem}[Debiased Estimator]\label{theo:debias}
    Assume that Assumptions  \ref{ass:converge:beta}--\ref{ass:CLT} hold. Assume also that $\frac{1}{\kappa} \sum_{s = 1}^{\kappa} \hat{\mathcal{E}}_{n,s}$ converges in probability to $ \mathbb{E}(\mathcal{E}_{0})$ as $n$ and $\kappa$ grow to infinity. \\
    \begin{inparaenum}[(i)] \item  $\boldsymbol{\theta}_{n,\kappa}^{\ast}$ is a $\sqrt{n}$-consistent estimator of $\boldsymbol{\theta}_0$.\label{theo:debias:consistency}\\
     \item Let $\boldsymbol{\psi}_n^{\ast} =  \mathbf{A}_0^{-1}\mathbf{V}_n^{1/2}\boldsymbol{\zeta} + \mathbf{A}_0^{-1}(\mathcal{E}_n - \mathbb{E}(\mathcal{E}_{0}))$ and let $F^{\ast}$ be the limiting distribution function of $\boldsymbol{\psi}_n^{\ast}$; that is, $ F^{\ast}(\boldsymbol{t}) = \lim\mathbb{P}(\boldsymbol{\psi}_n^{\ast} \preceq \boldsymbol{t})$ for all $\boldsymbol{t}\in\mathbb{R}^{K_{\theta}}$, then $\lim_{\kappa \to \infty}\lim \mathbb{P}(\sqrt{n}(\boldsymbol{\theta}_{n,\kappa}^{\ast} - \boldsymbol{\theta}_0) \preceq \boldsymbol{t}) = F^{\ast}(\boldsymbol{t})$.\label{theo:debias:distribution}\\
     \item The limiting distribution of $\sqrt{n}(\boldsymbol{\theta}_{n,\kappa}^{\ast} - \boldsymbol{\theta}_0)$ has a zero mean and a variance given by $\mathbf{A}_0^{-1}(\mathbf{V}_0 + \mathbb{V}(\mathcal{E}_0))\mathbf{A}_0^{-1}$.\label{theo:debias:moments}
    \end{inparaenum}
\end{theorem}
\noindent As mentioned earlier, $\mathcal{E}_n(\mathbf{\bar B}_{n,s}, ~\boldsymbol{\theta}_0, ~ \mathbf{B}_{0})$ and $\mathcal{E}_n(\mathbf{\hat B}_{n}, ~\boldsymbol{\theta}_0, ~ \mathbf{B}_{0})$ have the same asymptotic distribution. By the LLN, $\frac{1}{\kappa}\sum_{s = 1} ^{\kappa}\mathcal{E}_n(\mathbf{\bar B}_{n,s}, ~\boldsymbol{\theta}_0, ~ \mathbf{B}_{0})$ converges in probability to $\mathbb{E}(\mathcal{E}_{0})$ as $n$ and $\kappa$ grow to infinity.\footnote{This requires that $\mathbb{E}(\lvert \mathcal{E}_n(\mathbf{\bar B}_{n,s}, ~\boldsymbol{\theta}_0, ~ \mathbf{B}_{0})\rvert^{\nu}) < \infty$ for some $\nu > 2$. Under this condition, the LLN implies that  $\frac{1}{\kappa}\sum_{s = 1} ^{\kappa}\mathcal{E}_n(\mathbf{\bar B}_{n,s}, ~\boldsymbol{\theta}_0, ~ \mathbf{B}_{0})$ converges in probability to $\mathbb{E}(\mathcal{E}_n(\mathbf{\bar B}_{n,s}, ~\boldsymbol{\theta}_0, ~ \mathbf{B}_{0}))$ as $\kappa \to \infty$. Since $\mathcal{E}_n(\mathbf{\bar B}_{n,s}, ~\boldsymbol{\theta}_0, ~ \mathbf{B}_{0})$ converges in distribution to $\mathcal{E}_0$ (Assumption \ref{ass:IF}, Condition (\ref{ass:IF:dist})), it follows that $\mathbb{E}(\mathcal{E}_n(\mathbf{\bar B}_{n,s}, ~\boldsymbol{\theta}_0, ~ \mathbf{B}_{0}))$ converges to $\mathbb{E}(\mathcal{E}_{0})$ as $n\to\infty$. This holds because $\mathbb{E}(\lvert \mathcal{E}_n(\mathbf{\bar B}_{n,s}, ~\boldsymbol{\theta}_0, ~ \mathbf{B}_{0})\rvert^{\nu}) < \infty$ \citep[see][Theorem 4.5.2]{chung2001course}.} 
Thus, by assuming that $\frac{1}{\kappa} \sum_{s = 1}^{\kappa} \hat{\mathcal{E}}_{n,s}$ converges in probability to $\mathbb{E}(\mathcal{E}_{0})$ as $n$ and $\kappa$ grow to infinity in Theorem \ref{theo:debias}, we indirectly impose that $\mathcal{E}_n(\mathbf{\bar B}_{n,s}, ~\boldsymbol{\theta}_0, ~ \mathbf{B}_{0})$ is smooth in $\boldsymbol{\theta}_0$ and $\mathbf{B}_{0}$, ensuring that both $\frac{1}{\kappa}\sum_{s = 1} ^{\kappa}\mathcal{E}_n(\mathbf{\bar B}_{n,s}, ~\boldsymbol{\theta}_0, ~ \mathbf{B}_{0})$ and $\hat{\mathcal{E}}_{n,s}$ to have the same limit. This condition is similar to Assumption \ref{ass:Anosing} when $\boldsymbol{\theta}_0$ and $\mathbf{B}_{0}$ in $\frac{1}{n}\sum_{i = 1}^n \partial_{\boldsymbol{\theta}}\partial_{\boldsymbol{\theta}^\prime}q(\boldsymbol{\theta}_0, ~y_{i}, ~\boldsymbol{x}_{i}, ~\boldsymbol{\beta}_{0,i})$ are replaced with consistent estimators (see OA \ref{sm:hessian}). 

As in the case of the standard plug-in estimator, we can construct the sample $\mathcal{S}^{\ast}_{\kappa} =\{\boldsymbol{\hat \psi}_{n,s}^{\ast}: s = 1, ~\dots, ~ \kappa\}$ analogous to $\mathcal{S}_{\kappa}$  to estimate the CDF $F^{\ast}$. The variables $\boldsymbol{\hat \psi}_{n,s}^{\ast}$ are given by $\boldsymbol{\hat\psi}_{n,s}^{\ast} = (\mathbf{\hat A}_n^{\ast})^{-1}(\mathbf{\hat V}_{n}^{\ast})^{1/2}\boldsymbol{\zeta}_{s} + (\mathbf{\hat A}_n^{\ast})^{-1}(\hat{\mathcal{E}}_{n,s}^{\ast} -\boldsymbol{\hat \Omega}_{n}^{\ast\kappa})$, where $\boldsymbol{\hat \Omega}_{n}^{\ast\kappa} =  \frac{1}{\kappa}\sum_{s = 1}^{\kappa} \hat{\mathcal{E}}_{n,s}^{\ast}$. The variables $\mathbf{\hat A}_n^{\ast}$, $\hat{\mathcal{E}}_{n,s}^{\ast}$, and $\mathbf{\hat V}_{n}^{\ast}$ are respectively defined as $\mathbf{\hat A}_n$, $\hat{\mathcal{E}}_{n,s}$, and $\mathbf{\hat V}_{n,}$ using $\boldsymbol{\theta}_{n,\kappa}^{\ast}$ and not $\boldsymbol{\hat\theta}_n$.

The asymptotic variance given in Statement (\ref{theo:debias:moments}) of Theorem \ref{theo:debias} is the same as that of Theorem \ref{theo:variance}. This implies that we can use the debiased estimator even though $\mathbf V_n$ cannot be estimated but the asymptotic distribution of $\sqrt{n}(\boldsymbol{\theta}_{n,\kappa}^{\ast} - \boldsymbol{\theta}_0)$ is normal (i.e., if $\mathcal{E}_0$ is normally distributed). In this case, it is not necessary to construct the sample $\mathcal{S}^{\ast}_{\kappa}$ since the asymptotic variance $\mathbf{A}_0^{-1}(\mathbf{V}_0 + \mathbb{V}(\mathcal{E}_0))\mathbf{A}_0^{-1}$ can be directly estimated, as is done in standard inference approaches \citep[e.g., see][]{newey1984method, newey1994asymptotic, ackerberg2012practical, chen2015sieve}.

One issue regarding the debiased estimator is that the estimate of the finite sample bias of $\boldsymbol{\hat\theta}_{n}$ may itself be biased if the estimator of the FS distribution is biased. This situation can arise in small samples when the FS parameters are bounded, and their estimators exhibit large variances. In such cases, draws from the FS distribution may hit their bounds, leading to a biased estimate of $\mathbb{E}(\mathcal{E}_{0})$. We illustrate this issue in our simulation study through a Copula-GARCH model. To mitigate the bias in the estimate of $\mathbb{E}(\mathcal{E}_{0})$, we replace the sample mean $\boldsymbol{\hat\Omega}_{n}^{\kappa}$ in Equation \eqref{eq:thetast} with the sample median of $\{\hat{\mathcal{E}}_{n,s}, ~s= 1, ~\dots, ~\kappa\}$. The median correction yields better performance by avoiding the problem of outliers in $\hat{\mathcal{E}}_{n,s}$. Theorem \ref{theo:debias} remains valid under the median correction if the distribution of $\mathcal{E}_{0}$ is symmetric.

\section{Monte Carlo Simulation}\label{sec:simu}
In this section, we conduct a Monte Carlo study to assess the finite sample performance of the proposed CDF estimator and the debiased estimator. 

\subsection{Data-Generating Processes (DGPs)}
We study four DGPs, denoted as DGP A--D. The sample size $n$ takes values in $\{250, ~500, ~1000, \break 2000\}$. DGP A is a treatment effect model with endogeneity. The model is defined as follows:
$$y_i = \theta_0d_i + \varepsilon_i, \quad d_i = \mathbbm{1}\{z_i > 0.5(\varepsilon_i + 1.2)\}, \quad z_i \sim \text{Uniform}[0, ~1], \quad \varepsilon_i \sim \text{Uniform}[-1, 1],$$

\noindent where $d_i$ is a treatment status indicator ($d_i = 1$ if $i$ is treated), $z_i$ is an instrument for the treatment, and $\theta_0 = 1$. The treatment is endogenous given that it is correlated with the error term $\varepsilon_i$. The practitioner observes an i.i.d sample of $(y_i, ~d_i, ~z_i)$. We estimate $\theta_0$ by the IV method. In the first stage, we predict $\mathbb{E}(y_i|z_i)$ using an OLS regression of $y_i$ on $z_i$. We also predict $\mathbb{E}(d_i|z_i)$ using an OLS regression of $d_i$ on $z_i$. In the second stage, we regress the prediction of  $\mathbb{E}(y_i|z_i)$ on the prediction of $\mathbb{E}(d_i|z_i)$.

DGP B is similar to DGP A with the difference that many regressors are involved in the first stage. The practitioner has access to $k_n = O(\sqrt{n})$ instruments, $z_{1,i}, ~\dots, ~z_{k_n,i}$.  We maintain the specification of $y_i$ from DGP A but change the treatment status as follows: 
$$d_i = \mathbbm{1}\{0.2 + z_{1,i} + z_{2,i} + z_{3,i} + z_{4,i} > 0.5(\varepsilon_i + 1.2)\}, \quad z_{1,i}, ~\dots, ~z_{k_n,i} \overset{i.i.d}{\sim} \text{Uniform}[0, ~0.2].$$
Only four instruments, $z_{1,i}, ~\dots, ~z_{4,i}$, are relevant for $d_i$, while the others are superfluous variables that are independent of $d_i$. Yet, we include all $k_n$ instruments in the FS regressions, creating a scenario with many weak instruments \citep{cattaneo2019two}. We consider the cases where $k_n = \big\lfloor 2\sqrt{n} \big\rceil$ and $k_n = \big\lfloor 4\sqrt{n} \big\rceil$, where $\lfloor . \rceil$ is the rounding to the nearest integer. 

DGP C is a Poisson model with a latent covariate that is defined as:
    $$y_i \sim \text{Poisson}(\exp(\theta_{0,1} + \theta_{0,2}p_i)), \quad p_i = \sin^2(\pi z_i), \quad z_i \sim \text{Uniform}[0, ~10], \quad d_i \sim \text{Bernoulli}(p_i),$$

\noindent where $p_i$ is an unobserved probability and $\boldsymbol{\theta}_0 = (\theta_{0,1}, ~\theta_{0,2})^{\prime} = (-0.8, ~2)^{\prime}$. We consider an i.i.d. sample comprising data points $(y_i, ~z_i, ~d_i)$. The practitioner observes the pairs $(y_i, ~z_i)$ for all $i$ but only observes $d_i$ for a representative subsample of size $n^{\ast} = \lfloor n^{\alpha_n} \rceil$. The parameter $\alpha_n$ takes values in $\{1, ~0.985, ~0.945, ~0.91\}$, in the same order as $n$, i.e., $\alpha_n = 1$ if $n = 250$, $\alpha_n = 0.985$ if $n = 500$, and so forth. In the FS, we estimate $p_i = \mathbb{E}(d_i|z_i)$ using a semiparametric regression of $d_i$ on $z_i$ in the subsample of size $n^{\ast}$ where $d_i$ is observed. We rely on a piecewise cubic spline regression \citep[see][]{hastie2017generalized}. The regression results can be used to estimate $p_i$ for all $i$ in the full sample, as we observe $z_i$ for all $i$. The second stage is a standard Poisson regression after replacing $p_i$ with its estimator. 

DGP D is a multivariate time-series model similar to the model used in the simulation study by \cite{gonccalves2022bootstrapping}. We consider $k_n$ returns $y_{1,i}$, ~\dots,~ $y_{k_n, i}$, where $i$ is time and $k_n \geq 2$. Each $y_{p,i}$, for $p = 2, ~\dots,~ k_n$, follows an AR(1)-GARCH(1, 1) defined as: 
$$y_{p,i} = \phi_{p,0} + \phi_{p,1}y_{p,i-1}+  \sigma_{p,i}\varepsilon_{p,i}, \quad \sigma_{p,i}^2 = \beta_{p,0}+\beta_{p,1}\sigma_{p,i-1}^2\varepsilon_{p,i-1}^2+\beta_{p,2}\sigma_{p,i-1}^2,$$
where $\phi_{p,0} = 0$, $\phi_{p,i-1} = 0.4$, $\beta_{p,0} = 0.05$, $\beta_{p,1} = 0.05$, $\beta_{p,2} = 0.9$, and $\varepsilon_{p,i}$ follows a standardized Student distribution of degree-of-freedom $\nu_p = 6$.  The number of returns, $k_n$, increases with the sample size and takes values in $\{2, ~3, ~5, ~8\}$, in the same order as $n$, i.e., $k_n = 2$ if $n = 250$, $k_n = 3$ if $n = 500$, and so forth. We account for the correlation between the returns using the Clayton copula \citep[see][]{nelsen2006introduction}. The joint density function of $y_i = (y_{1,i}, ~\dots,~ y_{p,i})^{\prime}$ conditional on $\mathcal{F}^{i-1}$ (information set at $i-1$) is given by $c_i(G_{1,i}(\boldsymbol{\beta}_{0,1}), ~\dots,~ G_{k_n,i}(\boldsymbol{\beta}_{0,k_n}), ~\theta_0)$, where $\boldsymbol{\beta}_{0,p} = (\phi_{p,0}, ~\phi_{p,1}, ~\beta_{p,0}, ~\beta_{p,1}, ~\beta_{p,2}, ~\nu_p)^{\prime}$, $G_{p,i}(\boldsymbol{\beta}_{0,p})$ is the CDF of $y_{p,i}$ conditional on $\mathcal{F}^{i-1}$, and $c_i$ is the PDF of $k_n$-dimensional Clayton copula of parameter $\theta_0 = 4$. The practitioner observes the sample $y_1,~\dots,~y_n$. We rely on a multiple-stage estimation strategy to estimate $\theta_0$. In the first $k_n$ stages, we separately estimate each $\boldsymbol{\beta}_{0,p}$  by applying an AR(1)-GARCH(1, 1) model to the sample $y_{p,1}, ~\dots, y_{p,n}$. In the last stage, we estimate $\theta_0$ by maximum likelihood (ML) after replacing $\boldsymbol{\beta}_{0,p}$ in the density function of $y_i$ with its estimator.

\subsection{Simulation Results}
This section presents the simulation results. We perform 10,000 simulations and set $\kappa$ to 1,000. We begin by examining the estimates of the CDF of $\Delta_n$ (see Figures \ref{fig:simuAB} and \ref{fig:simuCD}). We use the sample distribution of $\Delta_n$ resulting from the simulations as the benchmark for comparing our estimates.  The corresponding CDF is represented by the curve $F_0$.

For DGPs A, B, and C, the conditional variance $\mathbf V_n$ is tractable. We thus use the simulation method described in Section \ref{sec:asymp:finite:simu} to estimate the asymptotic CDF of $\Delta_n$. The average of the estimated CDFs (for the 10,000 simulations) is represented by the curve $\hat F_n$. The curve $\hat H_n$ displays the average CDF obtained using the assumption that the asymptotic distribution of $\Delta_n$ is normal with a zero mean. The variance of this normal distribution is estimated by our simulation method (Equation \eqref{eq:Vasy}). In contrast, for DGP D, the conditional variance $\mathbf V_n$ is not tractable. We thus rely on Corollary \ref{cor:normal} and directly estimate the asymptotic variance of $\Delta_n$ using a standard inference method for sequential extremum estimators \citep[see][Section 6.3]{newey1994large}. Both curves $\hat F_n$ and $\hat H_n$ are normal CDFs. The mean of the normal distribution for the curve $\hat F_n$ is estimated by our simulation method (see Section \ref{sec:asymp:finite:bias}), while $\hat H_n$ corresponds to a normal CDF with a zero mean. For each estimated CDF, we present in parentheses the $L_1$-Wasserstein distance to the sample CDF $F_0$.\footnote{The $L_1$-Wasserstein distance between a CDF $\hat R_n$ and $F_0$ is $\lVert \hat R_n - F_0 \rVert_w = {\displaystyle\varint}_{\mathbb R}\lvert \hat R_n(t) - F_0(t) \rvert dt$.} 

For DGP A, both the $\hat H_n$ and $\hat F_n$ approximations yield strong performance, with the $\hat H_n$ providing a better fit of the actual CDF according to the Wasserstein distance. The result is not surprising since the first- and second-stage estimators are finite-dimensional estimators of type M. Thus $\Delta_n$ is asymptotically normally distributed with a zero mean \citep[see][]{murphy2002estimation}. Importantly, even when the asymptotic normality is verified, accurately computing the variance of a plug-in estimator can be intricate. Equation \eqref{eq:Vasy}, which is used for $\hat H_n$, accurately approximates this variance. 

Including many superfluous variables in the first stage of an IV approach can lead to biased estimates \citep{cattaneo2019two}. We can observe this result with DGP B given that the true distribution of $\Delta_n$ is not centered at zero. Yet, our inference method captures this bias, whereas $\hat H_n$, which corresponds to standard inference approaches, does not. The bias of our estimated CDF is larger when $k_n = \big\lfloor 4\sqrt{n} \big\rceil$, but vanishes as $n$ grows.\footnote{We construct 95\% CIs for $\boldsymbol{\theta}_0$ using  $\hat F_n$ and $\hat H_n$, and evaluate their coverage rates. The results, which are presented in OA \ref{sm:crate}, show $\hat H_n$ leads to over-rejection of the null hypothesis that $\boldsymbol{\theta}_0$ is equal to its actual value.}

In the case of DGP C, the size of the FS sample does not grow at the same rate as $n$.  This makes most classical inference approaches inapplicable. Our simulation approach performs well for both $\sqrt{n}(\hat{\theta}_{n,1} - \theta_{0,1})$ and $\sqrt{n}(\hat{\theta}_{n,2} - \theta_{0,2})$, where $\hat{\theta}_{n,1}$ and $\hat{\theta}_{n,2}$ represent the respective estimators for $\theta_{0,1}$ and $\theta_{0,2}$. Because of the slow convergence rate in the first stage, the CDFs of $\sqrt{n}(\hat{\theta}_{n,1} - \theta_{0,1})$ and $\sqrt{n}(\hat{\theta}_{n,2} - \theta_{0,2})$ are not centered at zero and our approach captures this feature. Conversely, the normal approximations perform poorly. 

Simulation results for DGP D show that our approach performs well but is somewhat less accurate in small samples. This discrepancy arises because inference for GARCH models requires a large number of observations. Moreover, as the number of returns increases with $n$, $\sqrt{n}(\log(\hat{\theta}_n) - \log(\theta_0))$ exhibits bias. Our approach captures this bias, whereas the normal approximations perform poorly.

\begin{figure}[!htbp]
    \begin{center}
        DGP A
    \end{center}
    \hspace{-1.9cm}\includegraphics[scale = 0.75]{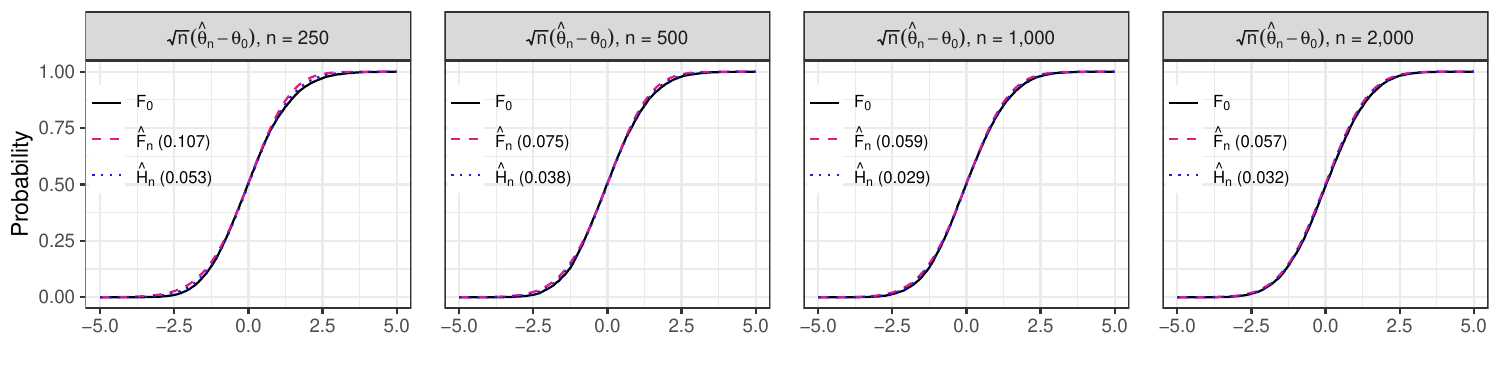}

    \vspace{-0.5cm}
    \begin{center}
        DGP B
    \end{center}
    \hspace{-1.9cm}\includegraphics[scale = 0.75]{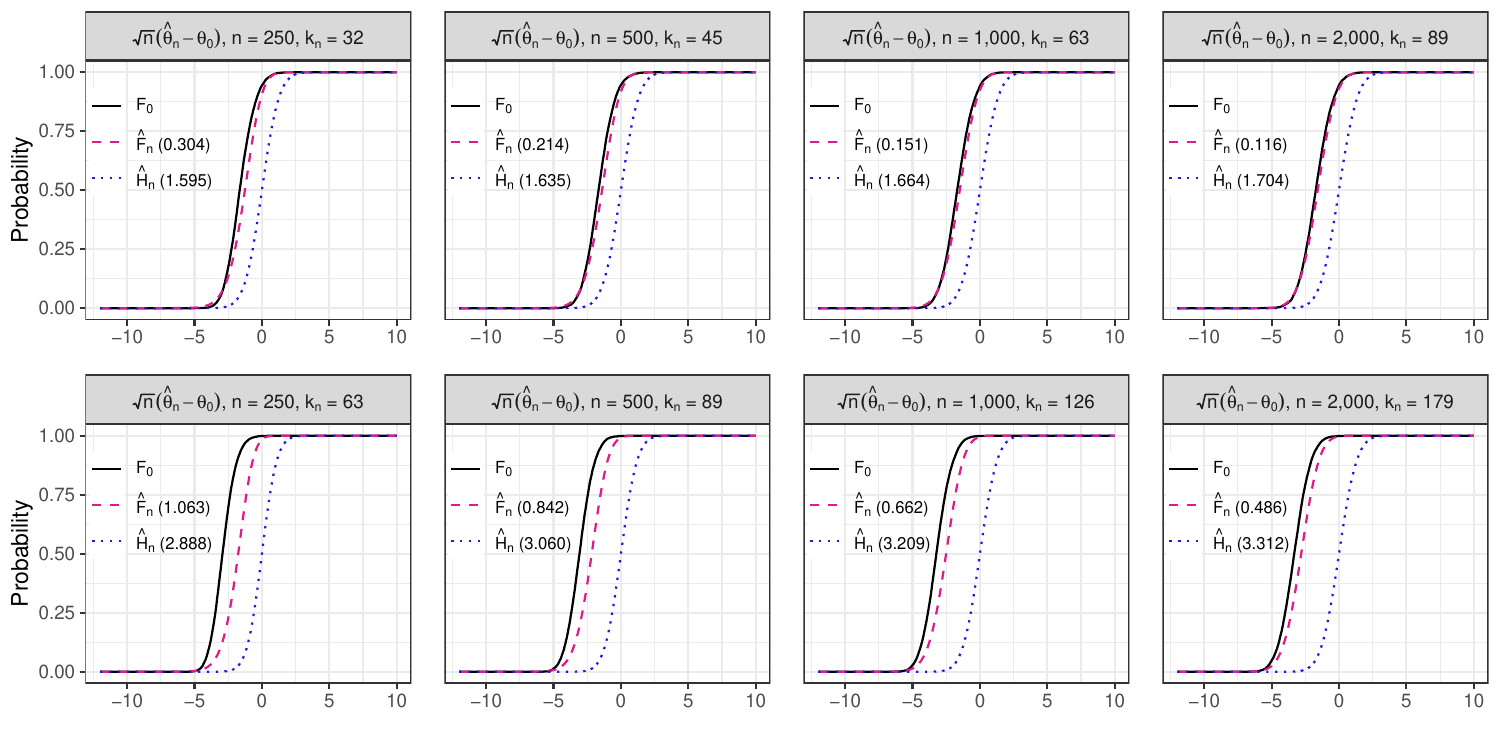}
    \caption{Monte Carlo Simulations: Estimates of asymptotic CDFs (DGPs A and B)}
    \label{fig:simuAB}

    \vspace{-0.2cm}
    \justify{\footnotesize This figure illustrates the average estimates of the CDF of $\sqrt{n}(\hat\theta_n - \theta_0)$ for DGPs A and B. $F_0$ is the true sample CDF. $\hat H_n$ represents the estimate based on asymptotic normality with a zero mean. $\hat F_n$ corresponds to the CDF estimate obtained through our simulation approach. The $L_1$-Wasserstein distance between each estimated CDF and $F_0$ is enclosed in parentheses.}
\end{figure}

\begin{figure}[!htbp]
    \begin{center}
        DGP C
    \end{center}
    \hspace{-1.9cm}\includegraphics[scale = 0.75]{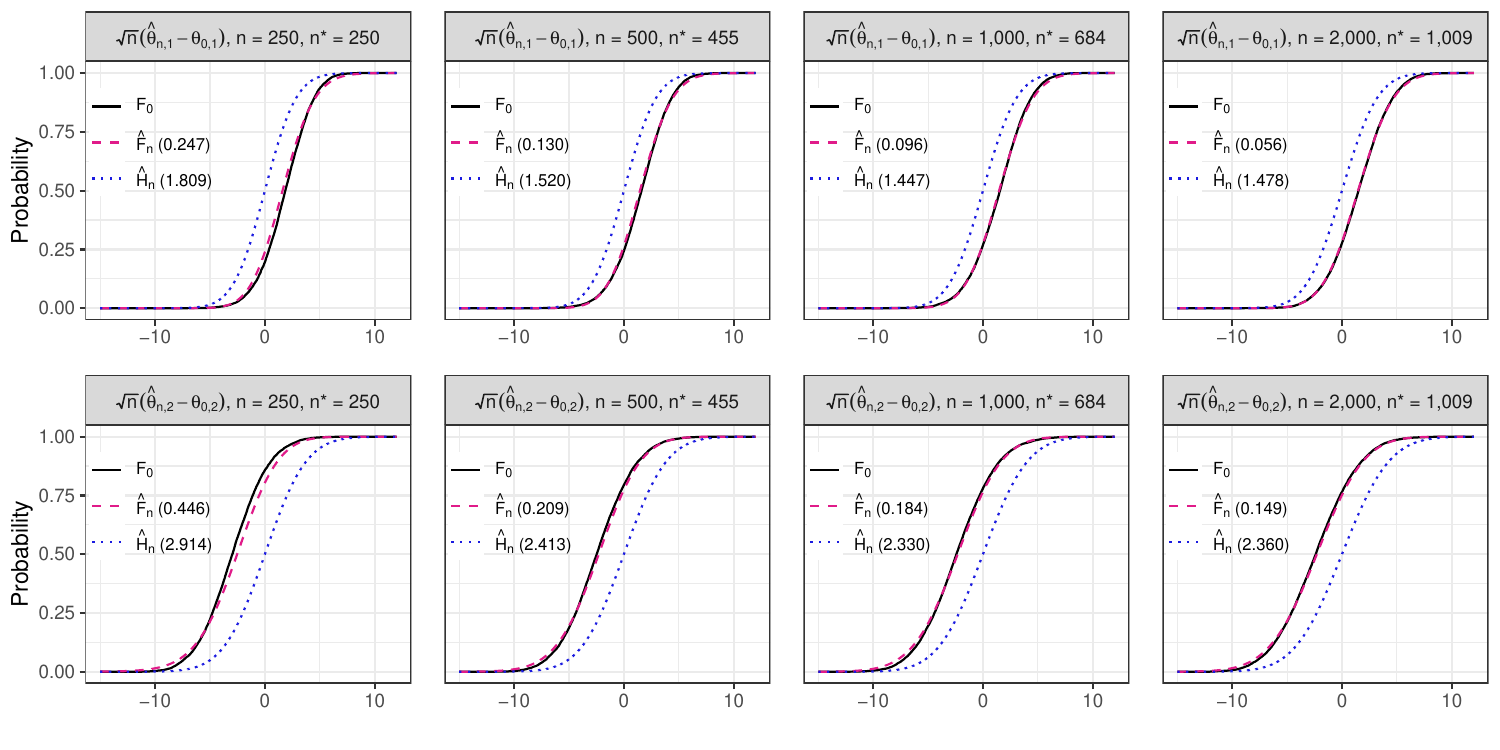}

    \vspace{-0.5cm}
    \begin{center}
      DGP D
    \end{center}
    \hspace{-1.9cm}\includegraphics[scale = 0.75]{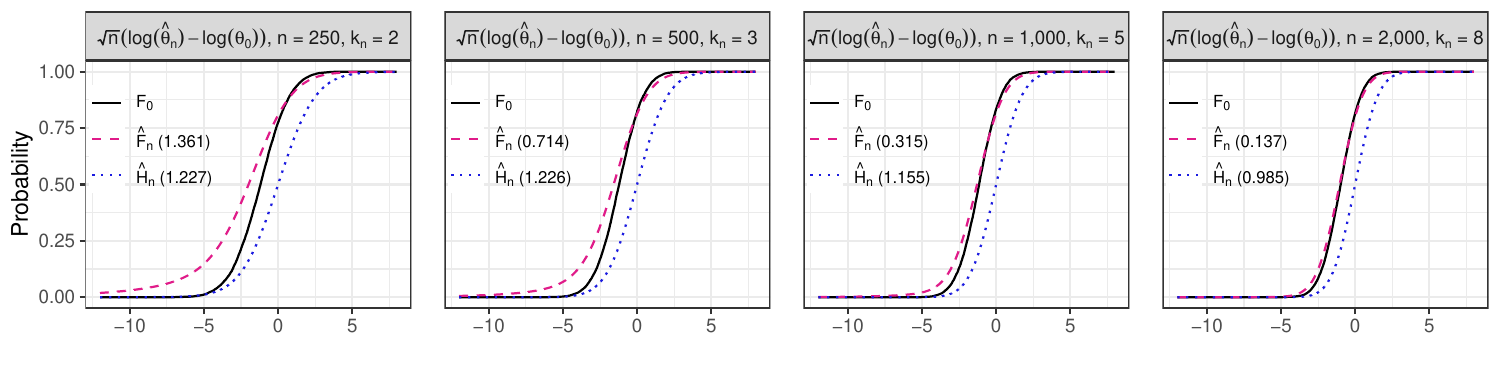}
    \caption{Monte Carlo Simulations: Estimates of asymptotic CDFs (DGPs C and D)}
    \label{fig:simuCD}

    \vspace{-0.2cm}
    \justify{\footnotesize This figure illustrates the average estimates of the CDF of $\sqrt{n}(\hat\theta_n - \theta_0)$ for DGPs C and D. $F_0$ is the true sample CDF. $\hat H_n$ represents the estimate based on asymptotic normality with a zero mean. $\hat F_n$ corresponds to the CDF estimate obtained through our simulation approach. The $L_1$-Wasserstein distance between each estimated CDF and $F_0$ is enclosed in parentheses.}
\end{figure}

We now turn to the finite sample performance of the debiased estimator (see Table \ref{tab:simubias}). For DGP A, where the classical plug-in estimator does not exhibit finite sample bias (because both stages result in a finite-dimensional extremum estimator), the debiased estimator closely aligns with the classical estimator. For DGP B, our debiased estimator significantly reduces the bias of the classical estimator. For example, with $2\sqrt{n}$ instrumental variables involved in the first stage, the bias of the classical plug-in estimator is substantial at $-0.101$ for $n = 250$. In contrast, the debiased estimator reduces this bias to $-0.017$. Not surprisingly, the bias reduction performs less well with $4\sqrt{n}$ instrumental variables but decreases the bias by more than 70\%.

The results are similar for DGP C. The biases of the estimates for $\theta_{0,1}$ and $\theta_{0,2}$ are respectively $0.114$ and $-0.184$ when $n = 250$ and $0.033$ and $-0.053$ when $n = 2000$. In contrast, these biases are negligible for the debiased estimator.

The bias correction performs less well for DGP D.  When $n = 2,000$, the debiased estimator's bias is ten times lower in absolute value (0.008 for the debiased estimator and $-$0.086 for the classical estimator). However, the correction is less effective in smaller samples. This result aligns with the discussion in Section \ref{sec:asymp:finite:bias}, which notes that the estimate of finite sample bias may be substantially biased when inequality constraints are imposed on the parameters in the first stage, especially if their estimates exhibit high variance. Constraints on GARCH model parameters include $\phi_{0,1}^2 < 1$ and $0< \beta_{p,1} + \beta_{p,3} < 1$. If the distribution from which we simulate in the first stage has high variance, the constraints may tighten for many draws, leading to an incorrect bias approximation. For instance, in the second stage, the standard deviation of the debiased estimator is 2.516, approximately 7 times higher than that of the classical estimator. 
In such situations, median correction is preferable. For $n = 250$, the debiased estimator using the median correction shows a lower bias of 0.088, which is 3 times smaller than the bias of the classical estimator.\footnote{In OA \ref{sm:simuR}, we present estimates of the asymptotic CDF of $\Delta_{n,\kappa}^{\ast}:=\sqrt{n}(\boldsymbol{\theta}_{n,\kappa}^{\ast} - \theta_0)$. Unlike the case of the classical plug-in estimator, the true sample CDFs are asymptotically centered at zero.}  

\renewcommand{\arraystretch}{1}
\begin{table}[!ht]
\centering
\caption{Parameter Estimates Summary}
\label{tab:simubias}
\begin{adjustbox}{max width=\linewidth}
\begin{threeparttable}
\small
\begin{tabular}{ld{3}d{3}d{3}d{3}d{3}cd{3}d{3}d{3}d{3}d{3}}
\toprule
                       & \multicolumn{5}{c}{Classical plug-in estimator}  &  & \multicolumn{5}{c}{Debiased plug-in estimator} \\
                   & \multicolumn{1}{c}{Mean}  & \multicolumn{1}{c}{Bias}   & \multicolumn{1}{c}{Sd}    & \multicolumn{1}{c}{RMSE}  & \multicolumn{1}{c}{MAE}   &  & \multicolumn{1}{c}{Mean}    & \multicolumn{1}{c}{Bias}     & \multicolumn{1}{c}{Sd}      & \multicolumn{1}{c}{RMSE}   & \multicolumn{1}{c}{MAE}    \\ \midrule
\multicolumn{6}{l}{DGP A: $\theta_0 = 1$}                                                       &  &         &         &         &         &        \\
n = 250          & 1.002          & 0.002          & 0.076        & 0.076        & 0.060        &  & 1.007    & 0.007    & 0.077  & 0.077  & 0.061  \\
n = 500          & 1.001          & 0.001          & 0.053        & 0.053        & 0.042        &  & 1.003    & 0.003    & 0.053  & 0.053  & 0.042  \\
n = 1000         & 1.001          & 0.001          & 0.037        & 0.037        & 0.030        &  & 1.002    & 0.002    & 0.037  & 0.037  & 0.030  \\
n = 2000         & 1.001          & 0.001          & 0.026        & 0.026        & 0.021        &  & 1.001    & 0.001    & 0.026  & 0.026  & 0.021  \\[0.2cm]
\multicolumn{6}{l}{DGP B: $k_n = \big\lfloor 2\sqrt{n} \big\rceil$, $\theta_0 = 1$}             &  &         &         &         &         &        \\
n = 250          & 0.899          & -0.101         & 0.061        & 0.118        & 0.104        &  & 0.983    & -0.017   & 0.076  & 0.077  & 0.062  \\
n = 500          & 0.927          & -0.073         & 0.045        & 0.086        & 0.076        &  & 0.992    & -0.008   & 0.053  & 0.054  & 0.043  \\
n = 1000         & 0.947          & -0.053         & 0.033        & 0.062        & 0.054        &  & 0.996    & -0.004   & 0.037  & 0.037  & 0.030  \\
n = 2000         & 0.962          & -0.038         & 0.024        & 0.045        & 0.039        &  & 0.998    & -0.002   & 0.026  & 0.026  & 0.021  \\[0.2cm]
\multicolumn{6}{l}{DGP B: $k_n = \big\lfloor 4\sqrt{n} \big\rceil$, $\theta_0 = 1$}             &  &         &         &         &         &        \\
n = 250          & 0.817          & -0.183         & 0.053        & 0.190        & 0.183        &  & 0.933    & -0.067   & 0.072  & 0.098  & 0.083  \\
n = 500          & 0.863          & -0.137         & 0.041        & 0.143        & 0.137        &  & 0.962    & -0.038   & 0.052  & 0.064  & 0.053  \\
n = 1000         & 0.899          & -0.101         & 0.031        & 0.106        & 0.101        &  & 0.979    & -0.021   & 0.037  & 0.043  & 0.035  \\
n = 2000         & 0.926          & -0.074         & 0.022        & 0.077        & 0.074        &  & 0.989    & -0.011   & 0.026  & 0.028  & 0.023  \\[0.2cm]
\multicolumn{6}{l}{DGP C: $\theta_{0,1} = -0.8$}                                                &  &         &         &         &         &        \\
n = 250          & -0.686         & 0.114          & 0.137        & 0.179        & 0.147        &  & -0.791   & 0.009    & 0.158  & 0.158  & 0.126  \\
n = 500          & -0.732         & 0.068          & 0.103        & 0.124        & 0.101        &  & -0.800   & 0.000    & 0.141  & 0.141  & 0.091  \\
n = 1000         & -0.754         & 0.046          & 0.076        & 0.089        & 0.072        &  & -0.801   & -0.001   & 0.082  & 0.082  & 0.065  \\
n = 2000         & -0.767         & 0.033          & 0.056        & 0.065        & 0.053        &  & -0.800   & 0.000    & 0.059  & 0.059  & 0.047  \\[0.2cm]
\multicolumn{6}{l}{DGP C: $\theta_{0,2} = 2$}                                                   &  &         &         &         &         &        \\
n = 250          & 1.816          & -0.184         & 0.174        & 0.253        & 0.212        &  & 1.983    & -0.017   & 0.210  & 0.211  & 0.168  \\
n = 500          & 1.892          & -0.108         & 0.131        & 0.170        & 0.140        &  & 1.999    & -0.001   & 0.188  & 0.188  & 0.120  \\
n = 1000         & 1.926          & -0.074         & 0.098        & 0.122        & 0.100        &  & 2.001    & 0.001    & 0.108  & 0.108  & 0.085  \\
n = 2000         & 1.947          & -0.053         & 0.073        & 0.090        & 0.073        &  & 2.000    & 0.000    & 0.079  & 0.079  & 0.063  \\[0.2cm]
\multicolumn{6}{l}{DGP D: $\theta_0 = 4$}                                                       &  &         &         &         &         &        \\
n = 250          & 3.724          & -0.276         & 0.384        & 0.473        & 0.389        &  & 4.506    & 0.506    & 2.516  & 2.566  & 0.735  \\
n = 500          & 3.794          & -0.206         & 0.226        & 0.306        & 0.253        &  & 4.140    & 0.140    & 0.937  & 0.947  & 0.298  \\
n = 1000         & 3.859          & -0.141         & 0.146        & 0.202        & 0.168        &  & 4.033    & 0.033    & 0.200  & 0.202  & 0.144  \\
n = 2000         & 3.914          & -0.086         & 0.096        & 0.129        & 0.106        &  & 4.008    & 0.008    & 0.105  & 0.106  & 0.083  \\[0.2cm]
\multicolumn{12}{l}{DGP D: $\theta_0 = 4$ (with a median correction for the debiased estimator)}        \\
n = 250          & 3.724          & -0.276         & 0.384        & 0.473        & 0.389        &  & 4.088    & 0.088    & 0.616  & 0.622  & 0.414  \\
n = 500          & 3.794          & -0.206         & 0.226        & 0.306        & 0.253        &  & 4.024    & 0.024    & 0.304  & 0.305  & 0.224  \\
n = 1000         & 3.859          & -0.141         & 0.146        & 0.202        & 0.168        &  & 4.003    & 0.003    & 0.168  & 0.168  & 0.131  \\
n = 2000         & 3.914          & -0.086         & 0.096        & 0.129        & 0.106        &  & 4.002    & 0.002    & 0.102  & 0.102  & 0.082   \\\bottomrule
\end{tabular}
\begin{tablenotes}[para,flushleft]
\footnotesize
This table displays the summary of the parameter estimates, including their mean, bias, standard deviations (sd), root mean square error (RMSE), and mean absolute error (MAE).
\end{tablenotes}
\end{threeparttable}
\end{adjustbox}
\end{table}
\renewcommand{\arraystretch}{1}

\section{Peer Effects in Adolescent Fast-Food Consumption Habits}\label{sec:appli} 
In this section, we revisit the empirical analysis conducted by \cite{fortin2015peer} on peer effects in adolescent fast-food consumption habits. Given the potential externalities that are associated with fast-food consumption and its link to overweight issues among adolescents, there may be justification for introducing a consumption tax on fast food. The optimal tax level hinges on the social multiplier of eating habits, emphasizing the need for an accurate measure of peer effects. We improve the popular IV estimator of peer effects by expanding the set of instruments, including \textit{many weak instruments}. We reduce the bias of the estimate using our approach.

\subsection{Estimation of Linear-in-Means Peer Effect Models}
This section presents the model used in this application. We consider a set of $R$ schools, where the number of students in the $r$-th school is denoted by $n_r$. Students within the same school interact. The network in the $r$-th school is represented by an adjacency matrix $\mathbf{G}_r = [g_{r,ij}]_{\substack{i = 1, ~\dots,~ n_r \\j = 1, ~\dots,~ n_r}}$, where $g_{r,ij} = 1$ if student $j$ is a friend of student $i$ and $g_{r,ij} = 0$ otherwise. We restrict friendships to the same school; students from different schools cannot be friends. Moreover, self-friendships are not allowed, in the sense that $g_{r,ii} = 0$ for all $r$ and $i$.
We consider the following linear-in-means peer effect model:
\begin{equation}
    \textstyle y_{r,i} = \alpha_{0,r} + \theta_{0,1}\sum_{j = 1}^{n_r}\frac{g_{r,ij}}{n_{r,i}} y_{r,j} + \boldsymbol{x}_{r,i}^{\prime}\boldsymbol{\theta}_{0,2} + \sum_{j = 1}^{n_r}\frac{g_{r,ij}}{n_{r,i}} \boldsymbol{x}_{r,j}^{\prime} \boldsymbol{\theta}_{0,3}  + \varepsilon_{r,i}, \label{eq:app:net}
\end{equation}
where $y_{r,i}$ is the weekly fast-food consumption frequency of student $i$ (reported frequency in days of fast-food restaurant visits in the past week), $\boldsymbol{x}_{r,i}$ is a vector of student $i$'s observable characteristics, $\varepsilon_{r,i}$ is an error term assumed to be independent of $\boldsymbol{x}_{r,i}$ and $\mathbf{G}_r$, and $n_{r,i} = \sum_{j = 1}^{n_r}g_{r,ij}$ is the number of friends of student $i$. The parameter $\theta_{0,1}$ captures peer effects, which measure the influence of an increase in the average friend's fast-food consumption frequency on one's fast-food consumption frequency.\footnote{The uniqueness of equilibrium in this model requires $\lvert \theta_{0,1} \rvert < 1$. That is, students do not increase their consumption frequency greater than the increase in
their average friends’ consumption frequency \citep[see][]{bramoulle2009identification}.} The parameter $\boldsymbol{\theta}_{0,2}$ reflects the effect of student characteristics, whereas $\boldsymbol{\theta}_{0,3}$ captures contextual effects, i.e., the influence of the average observable characteristics among friends. The parameter $\alpha_{0,r}$ accounts for unobserved effects of school characteristics, such as school location, regional taxes, and pricing policies.

The average friend's fast-food consumption frequency, which is measured by $\bar y_{r,i} = \frac{\sum_{j = 1}^{n_r}g_{r,ij} y_{r,j}}{n_{r,i}}$, is endogenous. Consequently, the classical OLS estimates of the parameters in \eqref{eq:app:net} are likely to be inconsistent.
Fortunately, one does not need to seek instruments for $\bar y_{r,i}$ elsewhere; they can be generated from the model. For the sake of clarity, we rewrite Equation \eqref{eq:app:net} in a matrix form for school $r$. Let $\mathbf{y}_r = (y_{r,1}, ~\dots, ~y_{r,n_s})^{\prime}$, $\mathbf{X}_r = (\boldsymbol{x}_{r, 1}, ~\dots, ~ \boldsymbol{x}_{r, n_r})^{\prime}$, and $\boldsymbol{\varepsilon}_r = (\varepsilon_{r,1}, ~\dots, ~\varepsilon_{r,n_s})^{\prime}$.\footnote{The notations $\mathbf{y}_r$ and $\mathbf{X}_r$ are only used in this section and must not be confused with $\mathbf{y}_n$ and $\mathbf{X}_n$ used elsewhere.} Let also the row-normalized adjacency matrix $\mathbf{\tilde G}_r = [\tilde g_{r,ij}]_{\substack{i = 1, ~\dots,~ n_r \\j = 1, ~\dots,~ n_r}}$, where $\tilde g_{r,ij} = 1/\sum_{j = 1}^{n_r} g_{r,ij}$ if $j$ is an $i$'s friend and $\tilde g_{r,ij} = 0$ otherwise. The linear-in-means peer effect model at the school level is:
\begin{equation}
   \mathbf{y}_r = \alpha_{0,r}\mathbf{1}_{n_r}  + \theta_{0,1}\mathbf{\tilde G}_r\mathbf{y}_r + \mathbf{X}_r\boldsymbol{\theta}_{0,2} + \mathbf{\tilde G}_r \mathbf{X}_r \boldsymbol{\theta}_{0,3}  + \boldsymbol{\varepsilon}_r, \label{eq:app:netmat}
\end{equation}

\noindent where $\mathbf{1}_{n_r}$ is an $n_r$-dimensional vector of ones. By premultiplying the terms of Equation \eqref{eq:app:netmat} by $\mathbf{\tilde G}_r$, we can observe that $\mathbf{\tilde G}_r^2 \mathbf{X}_r$ is correlated with the endogenous variable $\mathbf{\tilde G}_r\mathbf{y}_r$ if $\boldsymbol{\theta}_{0,3}\ne 0$. Since $\mathbf{\tilde G}_r^2 \mathbf{X}_r$ is not an explanatory variable in Equation \eqref{eq:app:netmat}, it can thus serve as an excluded instrument \citep[see][]{kelejian1998generalized, bramoulle2009identification}. This instrument is interpreted as the average friends of their average friends of the characteristics $\boldsymbol{x}_{r,i}$.

However, the instrument $\mathbf{\tilde G}_r^2 \mathbf{X}_r$ might suffer from weakness if $\boldsymbol{\theta}_{0,3} \approx 0$. To address this concern, we can show from Equation \eqref{eq:app:netmat} that:
\begin{equation}
\mathbb{E}(\mathbf{\tilde G}_r\mathbf{y}_r|\mathbf{X}_r, \mathbf{\tilde G}_r) = 
\mathbf{\tilde G}_r(\mathbf{I}_{n_r} - \theta_{0,1}\mathbf{\tilde G}_r)^{-1}(\alpha_{0,r}\mathbf{1}_{n_r} + \mathbf{X}_r\boldsymbol{\theta}_{0,2} + \mathbf{\tilde G}_r \mathbf{X}_r \boldsymbol{\theta}_{0,3}).\label{eq:app:EGy}
\end{equation}
Therefore, $\mathbb{E}(\mathbf{\tilde G}_r\mathbf{y}_r|\mathbf{X}_r, \mathbf{\tilde G}_r)$ can be used as an instrument. This approach is optimal since $\mathbb{E}(\mathbf{\tilde G}_r\mathbf{y}_r|\mathbf{X}_r, \mathbf{\tilde G}_r)$ fully captures exogenous variations of the endogenous variable $\mathbf{\tilde G}_r\mathbf{y}_r$.  In practice, employing this instrument entails a TS IV approach. A first IV method with $\mathbf{\tilde G}_r^2 \mathbf{X}_r$ as an instrument is used to estimate $\boldsymbol{\theta}_0 = (\alpha_{0,r}, ~\theta_{0,1}, ~\boldsymbol{\theta}_{0,2}^{\prime}, ~\boldsymbol{\theta}_{0,3}^{\prime})^{\prime}$. This estimate can also be employed to approximate $\mathbb{E}(\mathbf{\tilde G}_r\mathbf{y}_r|\mathbf{X}_r, \mathbf{\tilde G}_r)$ by replacing $\boldsymbol{\theta}_0$ in Equation \eqref{eq:app:netmat} with its estimate. A second IV method is performed with the estimate of $\mathbb{E}(\mathbf{\tilde G}_r\mathbf{y}_r|\mathbf{X}_r, \mathbf{\tilde G}_r)$ as an instrument to estimate $\boldsymbol{\theta}_0$.

While the optimal IV approach has been thoroughly considered in the literature, it may not entirely resolve the issue of weak instruments. Specifically, if the instrument $\mathbf{\tilde G}_r^2 \mathbf{X}_r$ that is used in the first IV approach is weak, the estimation of $\boldsymbol{\theta}_0$ can be biased, leading to a biased estimate for $\mathbb{E}(\mathbf{\tilde G}_r\mathbf{y}_r|\mathbf{X}_r, \mathbf{\tilde G}_r)$.
To circumvent this problem, we propose a new approach that consists of expanding the set of instruments for $\mathbf{\tilde G}_r\mathbf{y}_r$. As $(\mathbf{I}_{n_r} - \theta_{0,1}\mathbf{\tilde G}_r)^{-1} = \sum_{p = 0}^{\infty}\theta_{0,1}^p\mathbf{\tilde G}_r^p$, it can be shown from Equation \eqref{eq:app:EGy} that: 
$$\textstyle\mathbb{E}(\mathbf{\tilde G}_r\mathbf{y}_r|\mathbf{X}_r, \mathbf{\tilde G}_r) = 
\alpha_{0,r}\mathbf{\tilde G}_r(\mathbf{I}_{n_r} - \theta_{0,1}\mathbf{\tilde G}_r)^{-1}\mathbf{1}_{n_r} + \mathbf{\tilde G}_r\mathbf{X}_r\boldsymbol{\theta}_{0,2} +  \sum_{p = 0}^{\infty}\theta_{0,1}^p\mathbf{\tilde G}_r^{2+p}\mathbf{X}_r(\theta_{0,1}\boldsymbol{\theta}_{0,2} + \boldsymbol{\theta}_{0,3}).$$
This suggests the use of $\mathbf{\tilde G}_r^{2+p}\mathbf{X}_r$, for $p=0,~1,~2, ~\dots, k_{\max}$ as instruments, where $k_{\max}$ can be as large as possible for the matrix of instruments to be full rank. This set of instruments can be interpreted as averages of $\boldsymbol{x}_{r,i}$ among close- and long-distance friends. We control for 25 student characteristics in $\boldsymbol{x}_{r,i}$ (see below) and set $k_{\max} = 9$. This leads to 250 excluded instruments for $\mathbf{\tilde G}_r\mathbf{y}_r$. Despite this large number of instruments, our inference method can be used to construct CIs for the parameters of the model. Moreover, following Theorem \ref{theo:debias}, we correct for the finite sample bias of the resulting IV estimator.

\subsection{Add Health Data}
We use data from the National Longitudinal Study of Adolescent to Adult Health (Add Health) survey. The purpose of this survey was to investigate how various social contexts (families, friends, peers, schools, neighborhoods, and communities) affect adolescents' health and risk behaviors. The survey provides nationally representative and detailed information on adolescents in grades 7--12 from 144 schools during the 1994--1995 school year in the United States (US).  Approximately 90,000 students were asked to complete a brief questionnaire covering demographics, family backgrounds, academic performance, and health-related behaviors, as well as friendship links (best friends within the same school, up to 5 females and up to 5 males). From this main sample, an in-home sample (core sample) of about 20,000 students was randomly selected. These students participated in a more extensive questionnaire featuring detailed questions. This subsample has been followed in subsequent waves of the survey.

We use the Wave II dataset, which encompasses most of the variables that are relevant to this study. This wave targets the subsample of 20,000 students tracked over time. However, only 16 schools, comprising approximately 3,000 students, were completely surveyed in Wave II. To avoid the issue of sampled networks, we focus on the sample that was derived from these 16 schools, where we can observe the entire list of nominated best friends, and all nominated friends are also surveyed. In addition to the Wave II dataset, we gather information on each student's race and their mother's background from Wave I.

Our dependent variable is the weekly fast-food consumption frequency, measured by the reported frequency (in days) of fast-food restaurant visits in the past week. The final sample consists of 2,735 students, with an equal distribution between boys and girls. We control for 25 observable characteristics in $\mathbf{X}_r$ such as students' gender, grade, race, weekly allowance, and parents' education and occupation. On average, students report consuming fast food 2.35 days per week. The average age of students at the time of Wave II data collection is 16.62 years. Additional details on the data summary can be found in Table \ref{tab:data} in OA \ref{sm:app:data}.

\vspace{-0.25cm}
\subsection{Estimation and Inference}
Figure \ref{fig:app} displays estimates of peer effects using our approach and alternative methods, including the OLS estimator, the classical IV (CIV) estimator, the optimal IV (OIV) estimator, our IV estimator with many instruments (IV-MI), and the corresponding debiased IV estimator with many instruments (DIV-MI).\footnote{See full results, including coefficients of control variables in OA \ref{sm:app:results}.} The OLS approach overlooks the endogeneity issue, whereas the CIV method uses  $\mathbf{\tilde G}_r^2 \mathbf{X}_r$ as an instrument. The OIV estimator employs the estimate of $\mathbb{E}(\mathbf{\tilde G}_r\mathbf{y}_r|\mathbf{X}_r, \mathbf{\tilde G}_r)$ as an instrument, replacing unknown parameters in Equation \eqref{eq:app:EGy} with their CIV estimates. 

The OLS estimate indicates that the peer effect parameter is significant. The estimate decreases from 0.192 to 0.150 when we control for school-fixed effects. In contrast, the CIV estimator has a large variance, indicating that the coefficient is not statistically significant. This imprecision is a consequence of the weakness of the instruments \citep[][]{mikusheva2022inference}, leading to a biased estimator toward the OLS one. For instance, although it is known that the model suffers from an endogeneity problem, the Hausman-Wu endogeneity test (not reported here) indicates that the OLS and CIV estimators are not significantly different. 

As discussed earlier, this issue also invalidates the OIV approach since biased CIV estimates are used to estimate the optimal instrument. Notably, we observe that the 95\% confidence interval of the OIV is even larger than that of CIV. While the estimator becomes more precise when we control for school-fixed effects, the results still indicate that peer effects are not significant.

These findings align with the results of \cite{fortin2015peer}. Their OIV estimate of the peer effect parameter is 0.110 with a standard error approximated at 0.395, indicating non-significance. Additionally, they implemented a quasi-maximum likelihood (QML) method, estimating peer effects at 0.129, but the coefficient is significant only at the 10\% level.

After expanding the pool of instruments, the IV estimator estimate with many instruments reveals significant peer effects. The estimate decreases from 0.276 to 0.208 after accounting for school-fixed effects. Moreover, the results highlight evidence of finite sample bias, as the confidence intervals are not centered on the estimates. This bias is a consequence of the numerous instrumental variables in the first stage. After correcting for this bias, the estimates slightly increase to 0.300 and 0.218, respectively. The increase is smaller when controlling for school-fixed effects.

In summary, our results highlight the presence of peer effects in adolescent fast-food consumption habits. These results have two important implications. First, key players in the network can play a crucial role as channels for influencing adolescent habits. The more key players are influenced, the greater the potential spread of this influence in the network \citep[see][]{ballester2006s, zenou2016key}. Second, the social multiplier becomes crucial in determining the impact of a policy on adolescent fast-food consumption frequency. The social multiplier coefficient, given by $1/(1 - \theta_{0,1})$, is estimated at 1.279 for the model with fixed effects, considering the finite sample bias. This implies that the effect of a tax increase on fast-food consumption frequency, in an environment where adolescents do not interact with each other, must be multiplied by 1.279 when they interact \citep[see][]{agarwal2021thy}.

\begin{figure}[!htbp]
    \centering
    \includegraphics[scale = 0.8]{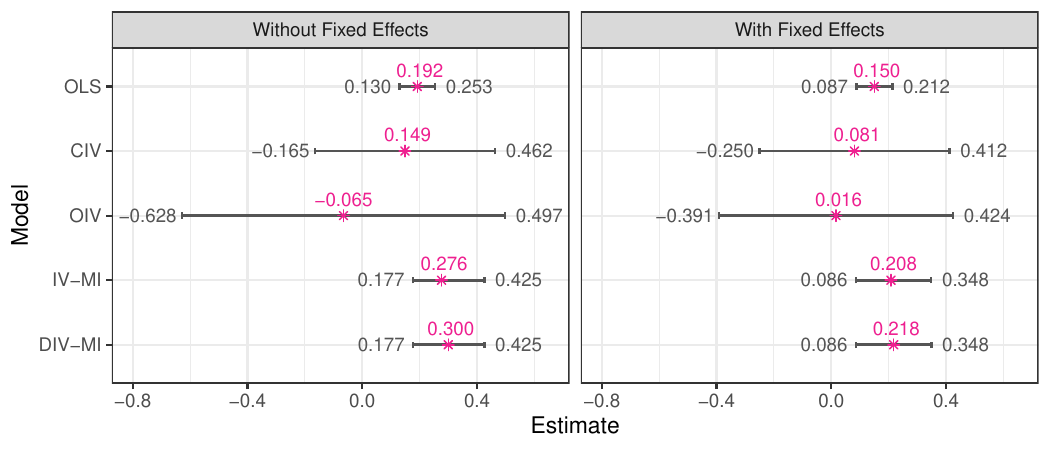}
    \caption{Peer Effect Estimates and Confidence Intervals}
    \label{fig:app}

    \vspace{-0.2cm}
    \justify{\footnotesize This figure presents peer effect estimates and confidence intervals. The asterisk symbols "$\ast$" denote the peer effect estimates, whereas the lines are the range of the 95\% CIs of the peer effects.}
\end{figure} 

\section{Conclusion}\label{sec:conc}
This paper proposes a simulation-based approach to estimate the asymptotic CDF of two-stage estimators. We consider a broad class of estimators in the first stage and an extremum estimator in the second stage. 

A key challenge in inference using two-stage estimation methods is that the asymptotic distribution of the plug-in estimator is influenced by first-stage sampling error. We tackle this issue by disentangling the first-stage sampling error from the model error in the second stage. Notably, we consider the possibility that the limiting distribution of $\sqrt{n}(\boldsymbol{\hat\theta}_n - \boldsymbol{\theta}_0)$ may not be normal and not centered at zero. This flexibility enables bias reduction for the plug-in estimator, particularly in cases where first-stage sampling error induces significant bias in the second stage. We introduce a debiased plug-in estimator and demonstrate that its limiting distribution has a zero mean. We conduct an extensive simulation study confirming the finite sample performance of our debiased estimator.

We leverage the proposed approach to study peer effects on adolescents' fast-food consumption habits. We employ an instrumental variable (IV) approach and address the issue of biased estimates that can arise when instruments are weak. By incorporating the characteristics of both close and distant friends, we expand the IV set to include a larger number of weak instruments. We correct the finite sample bias of the IV estimator and provide valid inference.

This paper contributes to the emerging literature on inference methods when standard regularity conditions are violated. The proposed approach is straightforward to implement and does not impose restrictions on the class of first-stage estimators. Our method is also suitable for complex models. Unlike resampling methods, it avoids the need for multiple computations of the estimators, except when resampling is required to obtain the asymptotic distribution in the first stage.

\newpage
\appendix
\renewcommand\thefigure{\thesection.\arabic{figure}}  
\renewcommand\thetable{\thesection.\arabic{table}}  
\setcounter{figure}{0} 
\setcounter{table}{0} 
\section{Appendix -- Proofs}\label{Append:proof}

\subsection{Proof of Theorem \ref{theo:dist}}\label{append:dist}
Let $F_{u,n}(\mathbf{t}) = \mathbb{P}(\boldsymbol{u}_{n}(\mathbf{y}_n,~\mathbf{\hat{B}}_n) \preceq \boldsymbol{t})$ for all $\boldsymbol{t} \in \mathbb{R}^{K_{\theta}}$. We first state and show the following lemma.
\begin{lemma}[Unconditional Asymptotic Normality]\label{lem:CLT}
    Under Assumption \ref{ass:CLT}, the unconditional distribution of $\boldsymbol{u}_{n}(\mathbf{y}_n,~\mathbf{\hat{B}}_n)$ is asymptotically normal, in the sense that $\lim F_{u,n}(\mathbf{t}) = \Phi(\boldsymbol{t})$ for each $\boldsymbol{t}$.
\end{lemma}
\begin{proof}
    We have $\displaystyle F_{u,n}(\mathbf{t}) = \mathbb{E}[\mathbb{P}(\boldsymbol{u}_{n}(\mathbf{y}_n,~\mathbf{\hat{B}}_n) \preceq \boldsymbol{t}|\mathbf{\hat{B}}_n)]$, where the expectation is taken w.r.t $\mathbf{\hat{B}}_n$. Thus, $\displaystyle \lim F_{u,n}(\mathbf{t}) = \lim \mathbb{E}\{\mathbb{P}(\boldsymbol{u}_{n}(\mathbf{y}_n,~\mathbf{\hat{B}}_n) \preceq \boldsymbol{t}|\mathbf{\hat{B}}_n)\}$. Given that $\lvert \mathbb{P}(\boldsymbol{u}_{n}(\mathbf{y}_n,~\mathbf{\hat{B}}_n) \preceq \boldsymbol{t}|\mathbf{\hat{B}}_n)\rvert \leq 1$, we can interchange the expectation and the limit in the previous equation  \citep[see Lebesgue–Vitali convergence theorem in][Theorem 4.5.4]{bogachev2007measure}. It follows that $\displaystyle \lim F_{u,n}(\mathbf{t}) = \mathbb{E}\{\plim\mathbb{P}(\boldsymbol{u}_{n}(\mathbf{y}_n,~\mathbf{\hat{B}}_n) \preceq \boldsymbol{t}|\mathbf{\hat{B}}_n)\}$. By Assumption \ref{ass:CLT}, $\plim\mathbb{P}(\boldsymbol{u}_{n}(\mathbf{y}_n,~\mathbf{\hat{B}}_n) \preceq \boldsymbol{t}|\mathbf{\hat{B}}_n) = \Phi(\boldsymbol{t})$. Hence, $\lim F_{u,n}(\mathbf{t}) = \Phi(\boldsymbol{t})$.
\end{proof}
We now show Theorem \ref{theo:dist}. By substituting $\boldsymbol{\dot q}_{n}(\mathbf{y}_n,~\mathbf{\hat{B}}_n) = \mathbf{V}_n^{1/2}\boldsymbol{u}_{n}(\mathbf{y}_n,~\mathbf{\hat{B}}_n) + \mathcal{E}_n$ in Equation \eqref{eq:sqrtntheta} we obtain
$\Delta_n = \mathbf{A}_n^{-1}\mathbf{V}_n^{1/2}\boldsymbol{u}_{n}(\mathbf{y}_n,~\mathbf{\hat{B}}_n) + \mathbf{A}_n^{-1}\mathcal{E}_n$. By Lemma \ref{lem:CLT}, $\boldsymbol{u}_{n}(\mathbf{y}_n, \mathbf{\hat{B}}_n)$ converges in distribution to a standard normal distribution. 
Thus, by the Slutsky theorem, $\mathbf{A}_n^{-1}\mathbf{V}_n^{1/2}\boldsymbol{u}_{n}(\mathbf{y}_n,~\mathbf{\hat{B}}_n)$ converges in distribution to $\mathbf{A}_0^{-1}\mathbf{V}_0^{1/2}\boldsymbol{\zeta}$, where $\boldsymbol{\zeta} \sim N(0, \boldsymbol{I}_{K_{\theta}})$.
Importantly, this convergence is regardless of the distribution $\mathbf{\hat{B}}_n$, and thus, regardless of the distribution of $\mathcal{E}_n$, since the latter is random only because of $\mathbf{\hat{B}}_n$.  Therefore, the conditional distribution of $
\Delta_n = \mathbf{A}_n^{-1}\mathbf{V}_n^{1/2}\boldsymbol{u}_{n}(\mathbf{y}_n,~\mathbf{\hat{B}}_n) + \mathbf{A}_n^{-1}\mathcal{E}_n$, given $\mathcal{E}_n$, and the conditional distribution of $\boldsymbol{\psi}_n = \mathbf{A}_0^{-1}\mathbf{V}_0^{1/2}\boldsymbol{\zeta} + \mathbf{A}_0^{-1}\mathcal{E}_n$, given $\mathcal{E}_n$, have the same limit. Put differently, $\plim\mathbb{P}(\Delta_n \preceq \boldsymbol{t} |\mathcal{E}_n) = \plim\mathbb{P}(\boldsymbol{\psi}_n \preceq \boldsymbol{t}|\mathcal{E}_n)$, which implies that $\mathbb{E}\big(\plim\mathbb{P}(\Delta_n \preceq \boldsymbol{t}  |\mathcal{E}_n)\big) = \mathbb{E}\big(\plim\mathbb{P}(\boldsymbol{\psi}_n \preceq \boldsymbol{t}  |\mathcal{E}_n)\big)$. As is the case in the proof of Lemma \ref{lem:CLT}, we can interchange the expectation and the limit because $\lvert \mathbb{P}(\Delta_n \preceq \boldsymbol{t}  |\mathcal{E}_n)\rvert \leq 1$ and $\lvert \mathbb{P}(\boldsymbol{\psi}_n \preceq \boldsymbol{t}  |\mathcal{E}_n)\rvert \leq 1$. Thus, $\lim\mathbb{E}\big(\mathbb{P}(\Delta_n \preceq \boldsymbol{t}  |\mathcal{E}_n)\big) = \lim\mathbb{E}\big(\mathbb{P}(\boldsymbol{\psi}_n \preceq \boldsymbol{t}  |\mathcal{E}_n)\big)$. As a result, $\lim\mathbb{P}(\Delta_n \preceq \boldsymbol{t}) = \lim\mathbb{P}(\boldsymbol{\psi}_n \preceq \boldsymbol{t})$. This completes the proof of Theorem \ref{theo:dist}.

\subsection{Proof of Theorem \ref{theo:variance}}\label{append:variance}
Equation \eqref{eq:sqrtntheta} is given by \begin{equation*}
        \textstyle\Delta_n=\sqrt{n}(\boldsymbol{\hat{\theta}}_n - \boldsymbol{\theta}_0)  = \mathbf{A}_n^{-1}\boldsymbol{\dot q}_{n}(\mathbf{y}_n,~\mathbf{\hat{B}}_n)\,.
\end{equation*} Let $\boldsymbol{\Sigma}_n = \mathbb{V}(\boldsymbol{\dot q}_{n}(\mathbf{y}_n,~\mathbf{\hat{B}}_n))$ and $\boldsymbol{\Sigma}_0 = \lim\boldsymbol{\Sigma}_n$. The existence of $\boldsymbol{\Sigma}_n$ and $\boldsymbol{\Sigma}_0$ is guaranteed by Assumption \ref{ass:IF}. As $\mathbb{E}(\lVert \boldsymbol{\dot q}_{n}(\mathbf{y}_n,~\mathbf{\hat{B}}_n) \rVert^{\nu})<\infty$ for some $\nu > 2$, it follows that the limiting distribution of $\boldsymbol{\dot q}_{n}(\mathbf{y}_n,~\mathbf{\hat{B}}_n)$ has variance $\boldsymbol{\Sigma}_0$ \citep[see][Theorem 4.5.2]{chung2001course}. Moreover, Assumption \ref{ass:Anosing} implies that $\plim \mathbf{A}_n = \mathbf{A}_0$. Therefore, $\mathbb{V}(\textstyle\Delta_0) = \mathbf{A}_0^{-1}\boldsymbol{\Sigma}_0\mathbf{A}_0^{-1}$. By the law of iterated variances, we have: 
\begin{equation}
    \textstyle\boldsymbol{\Sigma}_n = \mathbb{E}(\mathbf{V}_n) +  \mathbb{V}(\mathcal{E}_n). \label{eq:Sigma}
\end{equation}
As $\mathbf{V}_n$ is a semi-positive definite matrix with bounded expectation (Assumptions  \ref{ass:IF}), it is also uniformly integrable. Therefore, $\mathbb{E}(\mathbf{V}_n)$ converges to $\mathbf{V}_0$ \citep[see Lebesgue–Vitali theorem in][Theorem 4.5.4]{bogachev2007measure}. Moreover, since $\mathcal{E}_n$ converges in distribution to $\mathcal{E}_0$ and $\mathbb{E}(\lVert\mathcal{E}_n\rVert^{\nu})<\infty $  for some $\nu>2$, it follows that $\mathbb{V}(\mathcal{E}_n)$ converges to $\mathbb{V}(\mathcal{E}_0)$ \citep[see][Theorem 4.5.2]{chung2001course}. As a result,
$
   \textstyle  \boldsymbol{\Sigma}_0 = \mathbf{V}_0 +  \mathbb{V}(\mathcal{E}_0)$ and $\mathbb{V}(\textstyle\Delta_0) = \mathbf{A}_0^{-1}(\mathbf{V}_0 +  \mathbb{V}(\mathcal{E}_0))\mathbf{A}_0^{-1}$.

\subsection{Proof of Corollary \ref{cor:normal}}\label{append:cor:normal}
By Condition (\ref{ass:IF:dist}) of Assumption \ref{ass:IF}, $\{\mathbf{V}_n, ~\mathcal{E}_n\}$ converges in distribution to $\{\mathbf{V}_0, ~\mathcal{E}_0\}$. As $\{\mathbf{V}_n, ~\mathcal{E}_n\}$ is asymptotically independent of $\boldsymbol{\zeta}$, Theorem \ref{theo:dist} implies that $\sqrt{n}(\boldsymbol{\hat\theta}_n - \boldsymbol{\theta}_0)$ converges in distribution to $\boldsymbol{\psi}_0 :=  \mathbf{A}_0^{-1}\mathbf{V}_0^{1/2}\boldsymbol{\zeta} + \mathbf{A}_0^{-1}\mathcal{E}_0$. As $\mathbf{A}_0^{-1}\mathbf{V}_0^{1/2}\boldsymbol{\zeta}$ is normally distributed, it follows that $\boldsymbol{\psi}_0$ is also normally distributed if $\mathbf{A}_0^{-1}\mathcal{E}_0$  is normally distributed. As a result, $\Delta_n$ is normally distributed with mean $\mathbb{E}(\boldsymbol{\psi}_0) = \mathbf{A}_0^{-1}\mathbb{E}(\mathcal{E}_0)$ and variance $\mathbb{V}(\boldsymbol{\psi}_0) = \mathbf{A}_0^{-1}(\mathbf{V}_0 + \mathbb{V}(\mathcal{E}_0))\mathbf{A}_0^{-1}$.

\subsection{Proof of Theorem \ref{theo:debias}} \label{append:debias}
We have $\sqrt{n}(\boldsymbol{\theta}_{n,\kappa}^{\ast} - \boldsymbol{\theta}_0) = \Delta_n - \mathbf{\hat A}_n^{-1}\boldsymbol{\hat \Omega}_{n}^{\kappa}$. As $\Delta_n = O_p(1)$, $\plim \mathbf{\hat A}_n = \mathbf{A}_0$, and $\plim_{\kappa, n} \boldsymbol{\hat \Omega}_{n}^{\kappa} = \mathbb{E}(\mathcal{E}_0)$, it follows that $\sqrt{n}(\boldsymbol{\theta}_{n,\kappa}^{\ast} - \boldsymbol{\theta}_0) = O_p(1)$, that is, $\boldsymbol{\theta}_{n,\kappa}^{\ast}$ is a $\sqrt{n}$-consistent estimator of $\boldsymbol{\theta}_0$. This completes the proof of Statement (\ref{theo:debias:consistency}).

$\Delta_n$ has the same asymptotic distribution as $\boldsymbol{\psi}_n =  \mathbf{A}_0^{-1}\mathbf{V}_n^{1/2}\boldsymbol{\zeta} + \mathbf{A}_0^{-1}\mathcal{E}_n$ and $\mathbf{\hat A}_n^{-1}\boldsymbol{\hat \Omega}_{n}^{\kappa}$ converges in probability to the constant $ \mathbf{A}_0^{-1}\mathbb{E}(\mathcal{E}_0)$, as $\kappa$ and $n$ grow to infinity. Thus, by the Slutsky theorem,  $\sqrt{n}(\boldsymbol{\theta}_{n,\kappa}^{\ast} - \boldsymbol{\theta}_0) = \Delta_n - \mathbf{\hat A}_n^{-1}\boldsymbol{\hat \Omega}_{n}^{\kappa}$ has the same asymptotic distribution as $\boldsymbol{\psi}_n - \mathbf{A}_0^{-1} \mathbb{E}(\mathcal{E}_0)=  \mathbf{A}_0^{-1}\mathbf{V}_n^{1/2}\boldsymbol{\zeta} + \mathbf{A}_0^{-1}(\mathcal{E}_n - \mathbb{E}(\mathcal{E}_0))$. As a result, $\boldsymbol{\psi}_n^{\ast}$ and $\sqrt{n}(\boldsymbol{\theta}_{n,\kappa}^{\ast} - \boldsymbol{\theta}_0)$ share the same distribution as $\kappa$ and $n$ approach infinity.  This completes the proof of Statement (\ref{theo:debias:distribution}).

Statement (\ref{theo:debias:moments}) is a direct implication of Statement (\ref{theo:debias:distribution}). The variable $\boldsymbol{\psi}_n^{\ast}$ converges in distribution to $\boldsymbol{\psi}_0^{\ast} :=  \mathbf{A}_0^{-1}\mathbf{V}_0^{1/2}\boldsymbol{\zeta} + \mathbf{A}_0^{-1}(\mathcal{E}_0 - \mathbb{E}(\mathcal{E}_{0}))$, with $\mathbb{E}(\boldsymbol{\psi}_0^{\ast}) = 0$ and $\mathbb{V}(\boldsymbol{\psi}_0^{\ast}) = \mathbf{A}_0^{-1}(\mathbf{V}_0 + \mathbb{V}(\mathcal{E}_{0}))\mathbf{A}_0^{-1}$. As $\boldsymbol{\psi}_n^{\ast}$ and $\sqrt{n}(\boldsymbol{\theta}_{n,\kappa}^{\ast} - \boldsymbol{\theta}_0)$ have the same distribution, the result follows.

{\linespread{0.8}
\fontsize{9}{10}\selectfont
\bibliography{References}
\bibliographystyle{ecta}}

\newpage
\setcounter{page}{1}
\renewcommand{\theequation}{\thesection.\arabic{equation}}
\setcounter{equation}{0}
\setcounter{footnote}{0} 
\setcounter{section}{1} 
\setcounter{table}{0} 

\begin{mytitlepage}
\title{Supplemental Appendix for "Inference for Two-Stage Extremum Estimators"}
\maketitle

\begin{abstract}
{\linespread{1.2}\selectfont
\noindent This supplemental appendix includes additional results and technical details omitted from the main text.  Section \ref{sm:proof} discusses primitive conditions of some high-level assumptions introduced in the paper. Section \ref{sm:simu} provides technical details regarding our simulation study, including a description of how we approximate the asymptotic distributions. Section \ref{sm:app} presents detailed results of our empirical application with network data. }
\end{abstract}
\end{mytitlepage}

\bigskip
\section{Online Appendix--Proofs}\label{sm:proof}
\subsection{Consistency of Plug-in Estimators}\label{sm:cons}
We impose lower-level assumptions that result in a consistent plug-in estimator. For notational ease, we omit $\mathbf{y}_n$ and $\mathbf{X}_n$ in $Q_n(\boldsymbol{\theta}, ~\mathbf{y}_n, ~\mathbf{X}_n, ~\mathbf{B})$ and we simply write $Q_n(\boldsymbol{\theta}, ~\mathbf{B})$. We also write $q_i(\boldsymbol{\theta}, ~\boldsymbol{\beta}_{i})$ instead of $q(\boldsymbol{\theta},~ y_{i}, ~\boldsymbol{x}_{i}, ~\boldsymbol{\beta}_{i})$. For any $\boldsymbol{\theta}\in\boldsymbol{\Theta}$ and $a$, $\tau ~ > 0$, we define $$\textstyle\mathcal{H}_i(a, \tau) =\sup_{\lVert \boldsymbol{\beta}_{i} - \boldsymbol{\beta}_{0,i} \rVert < \tau} \dfrac{\lvert q_i(\boldsymbol{\theta},~ \boldsymbol{\beta}_{i}) - q_i(\boldsymbol{\theta},~ \boldsymbol{\beta}_{0,i})\rvert}{\lVert \boldsymbol{\beta}_{i} - \boldsymbol{\beta}_{0,i} \rVert^a},$$

\noindent where $\lVert . \rVert$ is the $\ell^2$-norm.
\begin{assumption}[Primitive Conditions for Assumption \ref{ass:consistent}] \label{ass:converge}\hfill\\
\begin{inparaenum}[(i)]
\item $\boldsymbol{\Theta}$ is a compact subset and $\boldsymbol{\theta}_0$ is an interior point of $\boldsymbol{\Theta}$. \label{ass:compact}\\
\item The function $Q_n(\boldsymbol{\theta}, ~\mathbf{B}_{0})$  converges uniformly in probability (across $\boldsymbol{\theta}\in\boldsymbol{\Theta}$) to a nonstochastic function $Q_0(\boldsymbol{\theta})$ that is maximized only at $\boldsymbol{\theta}_0$.\label{ass:converge:Q} \\
\item For any $\boldsymbol{\theta}\in\boldsymbol{\Theta}$, there exists constants $a(\boldsymbol{\theta})$, $\tau(\boldsymbol{\theta}) ~ > 0$, such that, $\max_i \mathcal{H}_i(a(\boldsymbol{\theta}), \tau(\boldsymbol{\theta})) = O_p(1)$.\label{ass:converge:holder}\\
\item There exists a sequence of neighborhoods $\mathcal{O}(\boldsymbol{\beta}_{0,1}), ~\dots, ~\mathcal{O}(\boldsymbol{\beta}_{0,n})$, such that $\partial_{\boldsymbol{\theta}}q_i(\boldsymbol{\theta}, ~\boldsymbol{\beta}_{i})$ is $O_p(1)$, uniformly across $\boldsymbol{\theta} \in \boldsymbol{\Theta}$, $\boldsymbol{\beta}_{i} \in \mathcal{O}(\boldsymbol{\beta}_{0,i})$ and $i$, that is,  $\max_{\boldsymbol{\theta} \in \boldsymbol{\Theta}, ~ \boldsymbol{\beta}_{i} \in \mathcal{O}(\boldsymbol{\beta}_{0,i}), ~i\leq n} \textstyle \lVert \partial_{\boldsymbol{\theta}}q_i(\boldsymbol{\theta},~ \boldsymbol{\beta}_{i})\rVert = O_p(1)$.\label{ass:converge:bounded}\end{inparaenum}
\end{assumption}

\noindent  
The compactness restriction in Condition (\ref{ass:compact}) allows for the plug-in estimator to converge in its support. Condition (\ref{ass:converge:Q}) is a classical identification condition also required for a standard M-estimator. Importantly, this condition does not involve the estimator $\mathbf{\hat{B}}_n$. It is an identification condition that is set at the true $\mathbf{B}_0$. Condition (\ref{ass:converge:holder}) implies that $\lvert q_i(\boldsymbol{\theta}, ~\boldsymbol{\beta}_{i}) - q_i(\boldsymbol{\theta}, ~\boldsymbol{\beta}_{0,i})\rvert \leq  \lVert \boldsymbol{\beta}_{i} - \boldsymbol{\beta}_{0,i} \rVert^{a(\boldsymbol{\theta})}O_p(1)$ for all $i$ and $\lVert \boldsymbol{\beta}_{i} - \boldsymbol{\beta}_{0,i} \rVert < \tau(\boldsymbol{\theta})$. A similar assumption is also imposed by \citeoa{cattaneo2019twooa} and requires $q_i(\boldsymbol{\theta}, ~\boldsymbol{\beta}_{i})$ to be smooth in $\boldsymbol{\beta}_{i}$. We use this condition and the uniform convergence of $\boldsymbol{\hat{\beta}}_{n,i}$ (Assumption \ref{ass:converge:beta}) to show that $Q_n(\boldsymbol{\theta}, ~\mathbf{\hat{B}}_{n}) - Q_n(\boldsymbol{\theta}, ~\mathbf{B}_{0})$ converges in probability to $0$ for each $\boldsymbol{\theta}$. Condition (\ref{ass:converge:bounded}) allows us to generalize this point-wise convergence to a uniform convergence.


\begin{proposition}\label{prop:consistent}
Under Assumptions \ref{ass:converge:beta}, \ref{ass:Qfonction}, and \ref{ass:converge}, the estimator $\boldsymbol{\hat{\theta}}_n$ converges in probability to $\boldsymbol{\theta}_0$.
\end{proposition} 

\begin{proof}
The proof is performed in two steps.

\medskip
\noindent \textbf{Step 1}: We show that $Q_n(\boldsymbol{\theta}, ~\mathbf{\hat{B}}_{n})$  converges uniformly in probability to  $Q_0(\boldsymbol{\theta})$.\\
For any $\boldsymbol{\theta}$, we have  $\lvert Q_n(\boldsymbol{\theta}, ~\mathbf{\hat{B}}_{n}) - Q_0(\boldsymbol{\theta})\rvert 
\leq \lvert Q_n(\boldsymbol{\theta}, ~\mathbf{\hat{B}}_{n}) - Q_n(\boldsymbol{\theta}, ~\mathbf{B}_{0})\rvert + \lvert Q_n(\boldsymbol{\theta}, ~\mathbf{B}_{0}) - Q_0(\boldsymbol{\theta}) \rvert$. Since $Q_n(\boldsymbol{\theta}, ~\mathbf{B}_{0}) - Q_0(\boldsymbol{\theta})$ converges uniformly in probability in $\boldsymbol{\theta}$ to $0$ (Condition (\ref{ass:converge:Q}) of Assumption \ref{ass:converge}), it is sufficient to show that $Q_n(\boldsymbol{\theta}, ~\mathbf{\hat{B}}_{n}) - Q_n(\boldsymbol{\theta}, ~\mathbf{B}_{0})$ also converges uniformly in probability in $\boldsymbol{\theta}$ to $0$. 

By Assumption \ref{ass:converge:beta}, for $n$ large enough, $\displaystyle\max_i  \textstyle\lVert \boldsymbol{\hat\beta}_{n,i} - \boldsymbol{\beta}_{0,i} \rVert < \tau(\boldsymbol{\theta})$ with probability approaching one. Thus, by Condition (\ref{ass:converge:holder}) of Assumption \ref{ass:converge}, 
$$\lvert q_i(\boldsymbol{\theta}, ~\boldsymbol{\hat\beta}_{n,i}) - q_i(\boldsymbol{\theta}, ~\boldsymbol{\beta}_{0,i})\rvert \leq \lVert \boldsymbol{\hat\beta}_{n,i} - \boldsymbol{\beta}_{0,i} \rVert^{a(\boldsymbol{\theta})}\max_i \mathcal{H}_i(a(\boldsymbol{\theta}), \tau(\boldsymbol{\theta}))\quad \text{for all $i$}$$ with probability approaching one. As  $\lvert Q_n(\boldsymbol{\theta}, ~\mathbf{\hat{B}}_{n}) - Q_n(\boldsymbol{\theta}, ~\mathbf{B}_{0}) \rvert \leq \max_i\lvert q_i(\boldsymbol{\theta}, ~\boldsymbol{\hat\beta}_{n,i}) - q_i(\boldsymbol{\theta}, ~\boldsymbol{\beta}_{0,i})\rvert$, this implies that 
$$\lvert Q_n(\boldsymbol{\theta}, ~\mathbf{\hat{B}}_{n}) - Q_n(\boldsymbol{\theta}, ~\mathbf{B}_{0}) \rvert \leq \max_i\lVert \boldsymbol{\beta}_{i} - \boldsymbol{\beta}_{0,i} \rVert^{a(\boldsymbol{\theta})} \max_i \mathcal{H}_i(a(\boldsymbol{\theta}), \tau(\boldsymbol{\theta}))$$ 
with probability approaching one. As a result $Q_n(\boldsymbol{\theta}, ~\mathbf{\hat{B}}_{n}) - Q_n(\boldsymbol{\theta}, ~\mathbf{B}_{0})$ converges in probability to zero because $\max_i\lVert \boldsymbol{\beta}_{i} - \boldsymbol{\beta}_{0,i} \rVert^{a(\boldsymbol{\theta})} = o_p(1)$ and  $\max_i \mathcal{H}_i(a(\boldsymbol{\theta}), \tau(\boldsymbol{\theta})) = O_p(1)$.

To show that the convergence is uniform, we apply the mean value theorem to $Q_n(\boldsymbol{\theta}, ~\mathbf{\hat{B}}_{n}) - Q_n(\boldsymbol{\theta}, ~\mathbf{B}_{0})$ with respect to $\boldsymbol{\theta}$. For any $\boldsymbol{\tilde\theta} \in \boldsymbol{\Theta}$, we have 
$$Q_n(\boldsymbol{\theta}, ~\mathbf{\hat{B}}_{n}) - Q_n(\boldsymbol{\theta}, ~\mathbf{B}_{0}) - \big(Q_n(\boldsymbol{\tilde\theta}, ~\mathbf{\hat{B}}_{n}) - Q_n(\boldsymbol{\tilde\theta}, ~\mathbf{B}_{0})\big)  = (\boldsymbol{\theta}-\boldsymbol{\tilde\theta})^{\prime}\boldsymbol{\hat{Q}}_n,$$
where $\boldsymbol{\hat{Q}}_n = \frac{1}{n}\sum_{i = 1}^n \partial_{\boldsymbol{\theta}}\big(q_i(\boldsymbol{\theta}^+_n, ~\boldsymbol{\hat\beta}_{n,i}) - q_i(\boldsymbol{\theta}^+_n, ~\boldsymbol{\beta}_{0,i})\big)$, for some $\boldsymbol{\theta}^+_n$ that lies between $\boldsymbol{\theta}$ and $\boldsymbol{\tilde \theta}$. Thus,
$$\lvert Q_n(\boldsymbol{\theta}, ~\mathbf{\hat{B}}_{n}) - Q_n(\boldsymbol{\theta}, ~\mathbf{B}_{0}) - \big(Q_n(\boldsymbol{\tilde\theta}, ~\mathbf{\hat{B}}_{n}) - Q_n(\boldsymbol{\tilde\theta}, ~\mathbf{B}_{0})\big)\rvert \leq \lVert\boldsymbol{\hat{Q}}_n\rVert \lVert \boldsymbol{\theta}-\boldsymbol{\tilde\theta}\rVert.$$
As $\boldsymbol{\hat{Q}}_n = O_p(1)$ (Condition (\ref{ass:converge:bounded}) of Assumption \ref{ass:converge}) and $\boldsymbol{\Theta}$ is compact, it follows from Lemma 2.9 of \citeoa{newey1994largeoa} that ${Q_n(\boldsymbol{\theta}, ~\mathbf{\hat{B}}_{n}) - Q_n(\boldsymbol{\theta}, ~\mathbf{B}_{0})}$ converges uniformly in probability to $0$. 

\medskip
\noindent \textbf{Step 2:} We establish the consistency of the estimator $\boldsymbol{\hat{\theta}}_n$. 

\noindent Let $\mathcal{O}(\boldsymbol{\theta}_0)^c$ be the complement of $\mathcal{O}(\boldsymbol{\theta}_0)$ in $\boldsymbol{\Theta}$. Note that $\mathcal{O}(\boldsymbol{\theta}_0)^c$ is nonempty and compact and $\max_{\boldsymbol{\theta} \in \mathcal{O}(\boldsymbol{\theta}_0)^c}Q_0(\boldsymbol{\theta})$ exists. 
Let $\delta = Q_0(\boldsymbol{\theta}_0) - \max_{\boldsymbol{\theta} \in \mathcal{O}(\boldsymbol{\theta}_0)^c}Q_0(\boldsymbol{\theta})$ and $J_n = \big\{|Q_n(\boldsymbol{\theta}, ~\mathbf{\hat{B}}_{n})-Q_0(\boldsymbol{\theta})|<\delta/2, \,\text{ for all } \boldsymbol{\theta}\big\}$. We know that  $Q_0(\boldsymbol{\theta})$ is uniquely maximized at $\boldsymbol{\theta}_0$ and that $\boldsymbol{\theta}_0\notin \mathcal{O}(\boldsymbol{\theta}_0)^c$. Thus $\delta >0$. Moreover, since $Q_n(\boldsymbol{\theta}, ~\mathbf{\hat{B}}_{n})$ converges uniformly in probability to $Q_0(\boldsymbol{\theta})$, we have ${\lim\mathbb{P}(J_n) = 1}$.
\begin{equation}
    J_n \implies \big\{Q_0(\boldsymbol{\hat \theta}_n)> Q_n(\boldsymbol{\hat \theta}_n,  \mathbf{\hat{B}}_{n})-\delta/2\big\} \cap \big\{Q_n(\boldsymbol{\theta}_0, ~\mathbf{\hat{B}}_{n})> Q_0(\boldsymbol{\theta}_0)-\delta/2\big\}\label{eq:Jn1}
\end{equation}
As $\boldsymbol{\hat{\theta}}_n = \argmax_{\boldsymbol{\theta}} Q_n(\boldsymbol{\theta}, ~\mathbf{\hat{B}}_{n})$, we also have $Q_n(\boldsymbol{\hat{\theta}}_n, ~\mathbf{\hat{B}}_{n}) \geq Q_n(\boldsymbol{\theta}_0, ~\mathbf{\hat{B}}_{n})$. Thus, $\big\{Q_0(\boldsymbol{\hat \theta}_n)> Q_n(\boldsymbol{\hat \theta}_n,  \mathbf{\hat{B}}_{n})-\delta/2\big\}$ implies $\big\{Q_0(\boldsymbol{\hat \theta}_n)> Q_n(\boldsymbol{ \theta}_0,  \mathbf{\hat{B}}_{n})-\delta/2\big\}$. It follows from \eqref{eq:Jn1} that 
\begingroup\allowdisplaybreaks\begin{align}
    J_n &\implies \big\{Q_0(\boldsymbol{\hat \theta}_n)> Q_n(\boldsymbol{\theta}_0,  \mathbf{\hat{B}}_{n})-\delta/2\big\} \cap \big\{Q_n(\boldsymbol{\theta}_0, ~\mathbf{\hat{B}}_{n})> Q_0(\boldsymbol{\theta}_0)-\delta/2\big\},\nonumber\\
    J_n &\implies \big\{Q_0(\boldsymbol{\hat{\theta}}_n)> Q_0(\boldsymbol{\theta}_0)-\delta\big\}.\label{eq:Jn2}
\end{align}\endgroup
\noindent As $\delta=Q_0(\boldsymbol{\theta}_0) - \max_{\boldsymbol{\theta} \in \mathcal{O}(\boldsymbol{\theta}_0)^c}Q_0(\boldsymbol{\theta})$, it turns out from \eqref{eq:Jn2} that
\begingroup\allowdisplaybreaks\begin{align}
    J_n &\implies \big\{Q_0(\boldsymbol{\hat{\theta}}_n)> \max_{\boldsymbol{\theta} \in \mathcal{O}(\boldsymbol{\theta}_0)^c}Q_0(\boldsymbol{\theta})\big\},\nonumber\\
    J_n &\implies \boldsymbol{\hat{\theta}}_n \in \mathcal{O}(\boldsymbol{\theta}_0).\label{eq:Jn3}
\end{align}\endgroup
As ${\lim\mathbb{P}(J_n) = 1}$, then \eqref{eq:Jn3} implies that ${\lim\mathbb{P}(\boldsymbol{\hat{\theta}}_n \in \mathcal{O}(\boldsymbol{\theta}_0)) = 1}$. This is true for any open subset $\mathcal{O}(\boldsymbol{\theta}_0)$ that contains $\boldsymbol{\theta}_0$. As a result, $\boldsymbol{\hat{\theta}}_n$ converges in probability to $\boldsymbol{\theta}_0$. \end{proof}
     
\subsection{Primitive Conditions for Assumption \ref{ass:Anosing}} \label{sm:hessian}
For notational ease, let $\ddot q_i(\boldsymbol{\theta}, ~\boldsymbol{\beta}_{i}) = \partial_{\boldsymbol{\theta}}\partial_{\boldsymbol{\theta}^\prime}q(\boldsymbol{\theta},~ y_{i}, ~\boldsymbol{x}_{i}, ~\boldsymbol{\beta}_{i})$. For any $\boldsymbol{\theta}\in\boldsymbol{\Theta}$ and $a$, $\tau ~ > 0$, we define
$$\textstyle\ddot{\mathcal{H}}_i(a, \tau) =\sup_{(\lVert \boldsymbol{\theta} - \boldsymbol{\theta}_0 \rVert + \lVert \boldsymbol{\beta}_{i} - \boldsymbol{\beta}_{0,i} \rVert)^{a} < \tau} \dfrac{\lVert \ddot q_i(\boldsymbol{\theta}, ~\boldsymbol{\beta}_{i}) - \ddot q_i(\boldsymbol{\theta}_0, ~\boldsymbol{\beta}_{0,i})\rVert}{(\lVert \boldsymbol{\theta} - \boldsymbol{\theta}_0 \rVert + \lVert \boldsymbol{\beta}_{i} - \boldsymbol{\beta}_{0,i} \rVert)^a}\,.$$ We impose lower-level conditions that imply Assumption \ref{ass:Anosing}.
\begin{assumption}[Primitive Conditions for Assumption \ref{ass:Anosing}]\label{sm:pm:ass:Anosing}
    \begin{inparaenum}[(i)]\hfill\\
        \item The matrix $\frac{1}{n}\sum_{i = 1}^n \ddot q_i(\boldsymbol{\theta}_0, ~\boldsymbol{\beta}_{0,i})$ converges in probability to a finite nonsingular matrix $\mathbf{A}_0$ defined by $\lim \mathbb{E}\big(\frac{1}{n}\sum_{i = 1}^n \ddot q_i(\boldsymbol{\theta}_0, ~\boldsymbol{\beta}_{0,i})\big)$.\label{sm:pm:ass:Anosing:limite}\\
        \item There exists constants $a^\ast, ~\tau^\ast > 0$ such that $\max_i\ddot{\mathcal{H}}_i(a^\ast, \tau^\ast) = O_p(1)$.\label{sm:pm:ass:Anosing:holder}
    \end{inparaenum}
\end{assumption}
\noindent Condition (\ref{sm:pm:ass:Anosing:limite}) imposes that $\frac{1}{n}\sum_{i}^n\ddot q_i(\boldsymbol{\theta}_0, ~\boldsymbol{\beta}_{0,i})$ converges in to $\lim \mathbb{E}\big(\frac{1}{n}\sum_{i = 1}^n \ddot q_i(\boldsymbol{\theta}_0, ~\boldsymbol{\beta}_{0,i})\big)$. This condition is classical as in the case of a single-step estimator. It does not involve any estimator and can be implied by the weak law of large numbers (WLLN). For $\ddot q_i(\boldsymbol{\theta}_0, ~\boldsymbol{\beta}_{0,i})$'s dependent across $i$, WLLN for dependent processes can be used. Condition (\ref{sm:pm:ass:Anosing:holder}) is similar to Condition (\ref{ass:converge:holder}) of Assumption \ref{ass:converge}. It requires $\ddot q_i(\boldsymbol{\theta}, ~\boldsymbol{\beta}_{i})$ to be smooth in both $\boldsymbol{\beta}_i$ and $\boldsymbol{\theta}$, uniformly in $i$. 

\begin{proposition}
Under Assumptions \ref{ass:converge:beta}--\ref{ass:consistent} and \ref{ass:converge}--\ref{sm:pm:ass:Anosing}, the Hessian of the objective function evaluated at any consistent estimator $\boldsymbol{\theta}_n^+$, given by $\frac{1}{n}\sum_{i = 1}^n \ddot q_i(\boldsymbol{\theta}_n^+, ~\boldsymbol{\hat{\beta}}_{n,i})$, converges in probability to a finite nonsingular matrix $\mathbf{A}_0 = \lim \mathbb{E}\big(\frac{1}{n}\sum_{i = 1}^n \ddot q_i(\boldsymbol{\theta}_0, ~\boldsymbol{\beta}_{0,i})\big)$.
\end{proposition} 
\begin{proof}
    By Assumption \ref{ass:converge:beta}, for $n$ large enough, $\lVert \boldsymbol{\theta}_n^+ - \boldsymbol{\theta}_0 \rVert+ \displaystyle\max_i  \textstyle\lVert \boldsymbol{\hat\beta}_{n,i} - \boldsymbol{\beta}_{0,i} \rVert < \tau^{\ast}$ with probability approaching one. Thus, $\lVert \ddot q_i(\boldsymbol{\theta}_n^+, ~\boldsymbol{\hat\beta}_{n,i}) - \ddot q_i(\boldsymbol{\theta}_0, ~\boldsymbol{\beta}_{0,i})\rVert \leq (\lVert \boldsymbol{\theta}_n^+ - \boldsymbol{\theta}_0 \rVert+ \lVert \boldsymbol{\hat\beta}_{n,i} - \boldsymbol{\beta}_{0,i} \rVert)^{a^{\ast}}\max_i\ddot{\mathcal{H}}_i(a^\ast, \tau^\ast)$, for all $i$, with probability approaching one. \\
    As  $\lVert \frac{1}{n}\sum_{i = 1}^n \ddot q_i(\boldsymbol{\theta}_n^+, ~\boldsymbol{\hat{\beta}}_{n,i}) - \frac{1}{n}\sum_{i = 1}^n \ddot q_i(\boldsymbol{\theta}_0, ~\boldsymbol{\beta}_{0,i}) \rVert \leq \max_i \lVert \ddot q_i(\boldsymbol{\theta}_n^+, ~\boldsymbol{\hat{\beta}}_{n,i})  - \ddot q_i(\boldsymbol{\theta}_0, ~\boldsymbol{\beta}_{0,i})\rVert$, thus $\lVert \frac{1}{n}\sum_{i = 1}^n \ddot q_i(\boldsymbol{\theta}_n^+, ~\boldsymbol{\hat{\beta}}_{n,i}) - \frac{1}{n}\sum_{i = 1}^n \ddot q_i(\boldsymbol{\theta}_0, ~\boldsymbol{\beta}_{0,i}) \rVert \leq (\lVert \boldsymbol{\theta}_n^+ - \boldsymbol{\theta}_0 \rVert+ \max_i\lVert \boldsymbol{\hat\beta}_{n,i} - \boldsymbol{\beta}_{0,i} \rVert)^{a^{\ast}}\max_i\ddot{\mathcal{H}}_i(a^\ast, \tau^\ast)$ with probability approaching one, where $\lVert \boldsymbol{\theta}_n^+ - \boldsymbol{\theta}_0 \rVert + \max_i\lVert \boldsymbol{\beta}_{i} - \boldsymbol{\beta}_{0,i} \rVert^{a(\boldsymbol{\theta})} = o_p(1)$ and  $\max_i\ddot{\mathcal{H}}_i(a^\ast, \tau^\ast) = O_p(1)$. As a result, $\frac{1}{n}\sum_{i = 1}^n \ddot q_i(\boldsymbol{\theta}_n^+, ~\boldsymbol{\hat{\beta}}_{n,i}) - \frac{1}{n}\sum_{i = 1}^n \ddot q_i(\boldsymbol{\theta}_0, ~\boldsymbol{\beta}_{0,i}) = o_p(1)$. Given that $\frac{1}{n}\sum_{i = 1}^n \ddot q_i(\boldsymbol{\theta}_0, ~\boldsymbol{\beta}_{0,i})$ converges in probability to a finite nonsingular matrix $\mathbf{A}_0$, the result follows.
\end{proof}

\subsection{Uniform Convergence of the Distribution of the Plug-in Estimator}\label{append:dist:p}
We establish a more general result than Theorem \ref{theo:dist} in terms of uniform convergence. To achieve this result, we consider a stronger version of Assumption \ref{ass:CLT} using uniform convergence.
\begin{assumption}[Conditional Asymptotic Normality]\label{ass:CLT:p}
The conditional distribution of the standardized influence function $\boldsymbol{u}_{n}(\mathbf{y}_n,~\mathbf{\hat{B}}_n)$, given $\mathbf{\hat{B}}_n$, uniformly  converges in distribution to $N(0, ~\boldsymbol{I}_{K_{\theta}})$, almost surely; that is, $\sup_{\boldsymbol{t}\in\mathbb{R}^{K_{\theta}}}\lvert \mathbb{P}(\boldsymbol{u}_{n}(\mathbf{y}_n,~\mathbf{\hat{B}}_n) \preceq \boldsymbol{t}|\mathbf{\hat{B}}_n) - \Phi(\boldsymbol{t})\rvert = o_p(1)$.
\end{assumption}
\noindent Assumption \ref{ass:CLT:p} can be obtained by applying a strong version of the CLT to $\boldsymbol{u}_{n}(\mathbf{y}_n,~\mathbf{\hat{B}}_n)$, conditional on $\mathbf{\hat{B}}_n$. This variant of the CLT is often referred to as the uniform CLT \citepoa[see][]{jirak2016berryoa, raivc2019multivariateoa}. As in Lemma \ref{lem:CLT}, Assumption \ref{ass:CLT:p} also implies the uniform convergence of the unconditional distribution.

\begin{lemma}[Unconditional Asymptotic Normality]\label{lem:CLT:p}
    Under Assumption \ref{ass:CLT:p}, $\boldsymbol{u}_{n}(\mathbf{y}_n,~\mathbf{\hat{B}}_n)$ uniformly converges to a standard normal distribution, in the sense that $\sup_{\boldsymbol{t}\in\mathbb{R}^{K_{\theta}}}\lvert \mathbb{P}(\boldsymbol{u}_{n}(\mathbf{y}_n,~\mathbf{\hat{B}}_n) \preceq \boldsymbol{t}) - \Phi(\boldsymbol{t})\rvert = o(1)$.
\end{lemma}

\begin{proof}
    Let $m_{\boldsymbol{t}}(\mathbf{\hat{B}}_n) = \mathbb{P}(\boldsymbol{u}_{n}(\mathbf{y}_n,~\mathbf{\hat{B}}_n) \preceq \boldsymbol{t}|\mathbf{\hat{B}}_n) - \Phi(\boldsymbol{t})$. We have $\mathbb{E}(m_{\boldsymbol{t}}(\mathbf{\hat{B}}_n)) = \mathbb{P}(\boldsymbol{u}_{n}(\mathbf{y}_n,~\mathbf{\hat{B}}_n) \preceq \boldsymbol{t}) - \Phi(\boldsymbol{t})$, where the expectation is taken with respect to $\mathbf{\hat{B}}_n$. As $\sup_{\boldsymbol{t}\in\mathbb{R}^{K_{\theta}}}\lvert \mathbb{E}(m_{\boldsymbol{t}}(\mathbf{\hat{B}}_n)) \rvert \leq \mathbb{E}(\sup_{\boldsymbol{t}\in\mathbb{R}^{K_{\theta}}}\lvert m_{\boldsymbol{t}}(\mathbf{\hat{B}}_n)\rvert)$, then $\sup_{\boldsymbol{t}\in\mathbb{R}^{K_{\theta}}}\lvert \mathbb{P}(\boldsymbol{u}_{n}(\mathbf{y}_n,~\mathbf{\hat{B}}_n) \preceq \boldsymbol{t}) - \Phi(\boldsymbol{t}) \rvert \leq \mathbb{E}(\sup_{\boldsymbol{t}\in\mathbb{R}^{K_{\theta}}}\lvert m_{\boldsymbol{t}}(\mathbf{\hat{B}}_n)\rvert)$. Since $\sup_{\boldsymbol{t}\in\mathbb{R}^{K_{\theta}}}\lvert m_{\boldsymbol{t}}(\mathbf{\hat{B}}_n)\rvert$ is bounded by one and is $o_{p}(1)$, then  $\lim\mathbb{E}(\sup_{\boldsymbol{t}\in\mathbb{R}^{K_{\theta}}}\lvert m_{\boldsymbol{t}}(\mathbf{\hat{B}}_n)\rvert) = 0$. This completes the proof.
\end{proof}

The following theorem establishes the uniform convergence of the distribution $\Delta_n$. 
\begin{theorem}[Asymptotic Distribution]
    Assumptions  \ref{ass:converge:beta}--\ref{ass:Anosing}, and \ref{ass:CLT:p} hold. Let
    $\boldsymbol{\chi}_n =  \boldsymbol{\zeta} + \mathbf{V}_n^{-1/2}\mathcal{E}_n$, where $\boldsymbol{\zeta} \sim N(0, ~\boldsymbol{I}_{K_{\theta}})$. Let $G(\boldsymbol{t}) = \lim \mathbb{P}(\boldsymbol{\chi}_n\preceq \boldsymbol{t})$ for $\boldsymbol{t}\in\mathbb{R}^{K_{\theta}}$ be the limiting distribution function of $\boldsymbol{\chi}_n$. We have $ \sup_{\boldsymbol{t}\in\mathbb{R}^{K_{\theta}}}\lvert \mathbb{P}(\sqrt{n}\mathbf{V}_n^{-1/2}\mathbf{A}_0(\boldsymbol{\hat{\theta}}_n - \boldsymbol{\theta}_0) \preceq \boldsymbol{t}) -  G(\boldsymbol{t})\rvert = o(1)$. \label{theo:dist:p}
\end{theorem}
\begin{proof}
    As $\{\boldsymbol{\chi}_n\preceq \boldsymbol{t}\}=\{\boldsymbol{\zeta} \preceq \boldsymbol{t} - \mathbf{V}_n^{-1/2}\mathcal{E}_n\}$, we have   
    $G(\boldsymbol{t}) = \lim \mathbb{E}\big\{\mathbb{P}(\boldsymbol{\zeta} \preceq \boldsymbol{t} - \mathbf{V}_n^{-1/2}\mathcal{E}_n|\mathcal{E}_n)\big\}$. Thus, 
    \begin{equation}G(\boldsymbol{t}) = \lim\mathbb{E}\{\Phi(\boldsymbol{t} - \mathbf{V}_n^{-1/2}\mathcal{E}_n)\},\label{append:dist:eqGt}\end{equation} 
    because  $\plim \mathbf{V}_n$ is nonstochastic and $\boldsymbol{\zeta} \sim N(0, ~\boldsymbol{I}_{K_{\theta}})$. 
    
    Since $\Delta_n = \mathbf{A}_n^{-1}\mathbf{V}_n^{1/2}\boldsymbol{u}_{n}(\mathbf{y}_n,~\mathbf{\hat{B}}_n) + \mathbf{A}_n^{-1}\mathcal{E}_n$, then  $\{\mathbf{V}_n^{-1/2}\mathbf{A}_n\Delta_n \preceq \boldsymbol{t}\}=\{\boldsymbol{u}_{n}(\mathbf{y}_n,~\mathbf{\hat{B}}_n) \preceq \boldsymbol{t} - \mathbf{V}_n^{-1/2}\mathcal{E}_n\}$. This translates to $\lim\mathbb{P}(\mathbf{V}_n^{-1/2}\mathbf{A}_0\Delta_n \preceq \boldsymbol{t}) = \lim\mathbb{P}\{\boldsymbol{u}_{n}(\mathbf{y}_n,~\mathbf{\hat{B}}_n) \preceq \boldsymbol{t} - \mathbf{V}_n^{-1/2}\mathcal{E}_n\}$ because $\plim \mathbf{A}_n = \mathbf{A}_0$ is nonstochastic. As $\mathbb{P}\{\boldsymbol{u}_{n}(\mathbf{y}_n,~\mathbf{\hat{B}}_n) \preceq \boldsymbol{t} - \mathbf{V}_n^{-1/2}\mathcal{E}_n\} = \mathbb{E}\{\mathbb{P}(\boldsymbol{u}_{n}(\mathbf{y}_n,~\mathbf{\hat{B}}_n) \preceq \boldsymbol{t} - \mathbf{V}_n^{-1/2}\mathcal{E}_n|\mathcal{E}_n)\}$, we thus have 
    \begin{equation}\lim\mathbb{P}(\mathbf{V}_n^{-1/2}\mathbf{A}_n\Delta_n \preceq \boldsymbol{t}) = \lim\mathbb{E}\big\{\mathbb{P}\{\boldsymbol{u}_{n}(\mathbf{y}_n,~\mathbf{\hat{B}}_n) \preceq \boldsymbol{t} - \mathbf{V}_n^{-1/2}\mathcal{E}_n|\mathcal{E}_n\}\big\}.\label{append:dist:eqPV}\end{equation} 
    From \eqref{append:dist:eqGt} and \eqref{append:dist:eqPV}, we have  $\lim\lvert \mathbb{P}(\mathbf{V}_n^{-1/2}\mathbf{A}_0\Delta_n \preceq \boldsymbol{t}) -  G(\boldsymbol{t})\rvert =\textstyle \lim\lvert\mathbb{E}\big\{\mathbb{P}\{\boldsymbol{u}_{n}(\mathbf{y}_n,~\mathbf{\hat{B}}_n) \preceq \boldsymbol{t} - \mathbf{V}_n^{-1/2}\mathcal{E}_n| \mathcal{E}_n\} - \Phi(\boldsymbol{t} - \mathbf{V}_n^{-1/2}\mathcal{E}_n)\big\}\rvert$. In addition,  for any $\boldsymbol{t}$ and $\mathcal{E}_n$, we have \\
    $\mathbb{P}\{\boldsymbol{u}_{n}(\mathbf{y}_n,~\mathbf{\hat{B}}_n) \preceq \boldsymbol{t} - \mathbf{V}_n^{-1/2}\mathcal{E}_n| \mathcal{E}_n\} - \Phi(\boldsymbol{t} - \mathbf{V}_n^{-1/2}\mathcal{E}_n) \leq \sup_{\boldsymbol{t}\in\mathbb{R}^{K_{\theta}}}\lvert\mathbb{P}\{\boldsymbol{u}_{n}(\mathbf{y}_n,~\mathbf{\hat{B}}_n) \preceq \boldsymbol{t}|\mathcal{E}_n\} - \Phi(\boldsymbol{t})\rvert$. Thus,  $\lim\lvert \mathbb{P}(\mathbf{V}_n^{-1/2}\mathbf{A}_0\Delta_n \preceq \boldsymbol{t}) -  G(\boldsymbol{t})\rvert \textstyle\leq  \lim\mathbb{E}\big\{\sup_{\boldsymbol{t}\in\mathbb{R}^{K_{\theta}}}\lvert\mathbb{P}\{\boldsymbol{u}_{n}(\mathbf{y}_n,~\mathbf{\hat{B}}_n) \preceq \boldsymbol{t}|\mathcal{E}_n\} - \Phi(\boldsymbol{t})\rvert\big\}$. 
    
    By Lemma \ref{lem:CLT:p},
    $\sup_{\boldsymbol{t}\in\mathbb{R}^{K_{\theta}}}\lvert\mathbb{P}\{\boldsymbol{u}_{n}(\mathbf{y}_n,~\mathbf{\hat{B}}_n) \preceq \boldsymbol{t}|\mathcal{E}_n\} - \Phi(\boldsymbol{t})\rvert = o_p(1)$ because conditioning on $\mathcal{E}_n$ involves conditioning on $\mathbf{\hat{B}}_n$. As $\sup_{\boldsymbol{t}\in\mathbb{R}^{K_{\theta}}}\lvert\mathbb{P}\{\boldsymbol{u}_{n}(\mathbf{y}_n,~\mathbf{\hat{B}}_n) \preceq \boldsymbol{t}|\mathcal{E}_n\} - \Phi(\boldsymbol{t})\rvert$ is bounded, this implies that $\lim\mathbb{E}\{\sup_{\boldsymbol{t}\in\mathbb{R}^{K_{\theta}}}\lvert\mathbb{P}(\boldsymbol{u}_{n}(\mathbf{y}_n,~\mathbf{\hat{B}}_n) \preceq \boldsymbol{t}|\mathcal{E}_n) - \Phi(\boldsymbol{t})\rvert\} = 0$. Consequently, we have  $\lim\sup_{\boldsymbol{t}\in\mathbb{R}^{K_{\theta}}}\lvert \mathbb{P}(\mathbf{V}_n^{-1/2}\mathbf{A}_0\Delta_n \preceq \boldsymbol{t}) -  G(\boldsymbol{t})\rvert = 0$. This completes the proof.
\end{proof}

\section{Supplementary Materials on the Simulation Study}\label{sm:simu}
This section provides a detailed explanation of the method used to estimate the asymptotic variance and the asymptotic cumulative distribution function (CDF) of the plug-in estimators in the simulation study. Our replication code available at \url{https://github.com/ahoundetoungan/InferenceTSE} implements this method. We also present the estimates of the asymptotic CDF of the debiased estimators.

\subsection{Asymptotic Variance and Asymptotic CDF}
\paragraph{DGPs A and B}
DGP A is a treatment effect model with endogeneity. The model is defined as follows:
\begin{align*}
    y_i = \theta_0d_i + \varepsilon_i, \quad d_i = \mathbbm{1}\{z_i > 0.5(\varepsilon_i + 1.2)\}, \quad z_i \sim \text{Uniform}[0, ~1], \quad \varepsilon_i \sim \text{Uniform}[-1, 1],
\end{align*}

\noindent where $d_i$ is a treatment status indicator, $z_i$ is an instrument for the treatment and $\theta_0=1$. In the first stage, we perform two OLS regressions: a regression of $y_i$ on $\boldsymbol{z}_i = (1, z_i)^{\prime}$ and another regression of $d_i$ on $\boldsymbol{z}_i$. For DGP B, the vector of regressors in the first stage is $\boldsymbol{z}_i = (1, ~z_{1,i}, ~\dots, z_{k_n,i})^{\prime}$.

Let $\boldsymbol{\hat\gamma}_n^{(y)}$ and $\boldsymbol{\hat\gamma}_n^{(d)}$ be the respective OLS estimators and let $\boldsymbol{\hat\gamma}_n = (\boldsymbol{\hat\gamma}_n^{(y)\prime}, ~\boldsymbol{\hat\gamma}_n^{(d)\prime})^{\prime}$ be the joint first-stage (FS) estimator.
Let also $\hat \nu_i^{(y)}$ and $\hat \nu_i^{(d)}$ be the residuals of the regressions; that is, $\hat \nu_i^{(y)} = y_i - \boldsymbol{z}_i^{\prime}\boldsymbol{\hat\gamma}_n^{(y)}$ and $\hat \nu_i^{(d)} = d_i - \boldsymbol{z}_i^{\prime}\boldsymbol{\hat\gamma}_n^{(d)}$. We define $\boldsymbol{z}_i^{(\nu)} = (\hat \nu_i^{(y)}\boldsymbol{z}_i^{\prime}, ~\hat \nu_i^{(d)}\boldsymbol{z}_i^{\prime})^{\prime}$.  The estimator of the asymptotic distribution of $\boldsymbol{\hat\gamma}_n$ is a normal distribution with mean $\boldsymbol{\hat\gamma}_n$ and covariance matrix $\hat{\mathbb{V}}(\boldsymbol{\hat\gamma}_n) = \frac{1}{n}\mathbf{\hat{H}}_n^{-1}\mathbf{\hat{J}}_n\mathbf{\hat{H}}_n^{-1}$, where $$\textstyle\mathbf{\hat{H}}_n = \frac{1}{n}\sum_{i=1}^{n}\diag(\sum_{i = 1}^n \boldsymbol{z}_i\boldsymbol{z}_i^{\prime}, ~\sum_{i = 1}^n \boldsymbol{z}_i\boldsymbol{z}_i^{\prime})\quad \text{and} \quad \textstyle\mathbf{\hat{J}}_n = \frac{1}{n}\sum_{i = 1}^n \boldsymbol{z}_i^{(\nu)}\boldsymbol{z}_i^{(\nu)\prime}.$$
The notation $\diag$ stands for the bloc diagonal matrix operator.
In the second stage, the objective function to be maximized is
$\textstyle Q_n(\theta, ~\mathbf{y}_n, ~\mathbf{\hat{B}}_n) = -\frac{1}{n}\sum_{i = 1}^n (\boldsymbol{z}_i^\prime \boldsymbol{\hat\gamma}_{n}^{(y)} - \theta\boldsymbol{z}_i^\prime \boldsymbol{\hat\gamma}_{n}^{(d)})^2$, where $\mathbf{\hat{B}}_n = (\boldsymbol{\hat\beta}_{n,1}, ~\dots,~ \boldsymbol{\hat\beta}_{n,n})^{\prime}$ and $\boldsymbol{\hat\beta}_{n,i} = (\boldsymbol{z}_i^{\prime}\boldsymbol{\hat\gamma}_n^{(y)}, ~ \boldsymbol{z}_i^{\prime}\boldsymbol{\hat\gamma}_n^{(d)})^{\prime}$.  This implies that $\boldsymbol{\dot q}_{n}(\textbf{y}_n, ~\mathbf{\hat{B}}_{n}) = \frac{2}{\sqrt{n}}\sum_{i=1}^n\boldsymbol{z}_i^\prime \boldsymbol{\hat\gamma}_{n}^{(d)}(\boldsymbol{z}_i^\prime \boldsymbol{\hat\gamma}_{n}^{(y)} - \theta_0\boldsymbol{z}_i^\prime \boldsymbol{\hat\gamma}_{n}^{(d)})$. 
We define the following expressions: 
$$\textstyle \hat A_n = \frac{1}{n}\sum_{i = 1}^n (\boldsymbol{z}_i^\prime \boldsymbol{\hat\gamma}_{n}^{(d)})^2 \quad \text{and} \quad \hat{\mathcal{E}}_{n,s} = \frac{2}{\sqrt{n}}\sum_{i=1}^n\boldsymbol{z}_i^\prime \boldsymbol{\hat\gamma}_{n}^{(d,s)}(\boldsymbol{z}_i^\prime \boldsymbol{\hat\gamma}_{n}^{(y,s)} - \hat\theta_n\boldsymbol{z}_i^\prime \boldsymbol{\hat\gamma}_{n}^{(d,s)}),$$
where $(\boldsymbol{\hat\gamma}_n^{(y,1)\prime}, ~\boldsymbol{\hat\gamma}_n^{(d,1)\prime})^{\prime}, ~\dots, ~(\boldsymbol{\hat\gamma}_n^{(y,\kappa)\prime}, ~\boldsymbol{\hat\gamma}_n^{(d,\kappa)\prime})^{\prime} \overset{i.i.d}{\sim} N(\boldsymbol{\hat\gamma}_n, ~\hat{\mathbb{V}}(\boldsymbol{\hat\gamma}_n))$. Let $\hat\psi_{n,s} = \frac{\hat{\mathcal{E}}_{n,s}}{\hat A_n}$. The asymptotic CDF of $\sqrt{n}(\hat\theta_n - \theta_0)$ can be estimated by the CDF of the sample: $\{\hat\psi_{n,s}, ~ s = 1, ~\dots,~ \kappa\}$. The estimator of the asymptotic variance of $\hat{\theta}_n$ is 
$$\textstyle \hat{\mathbb{V}}(\hat{\theta}_n) =\dfrac{\sum_{s = 1}^{\kappa}(\hat{\mathcal{E}}_{n,s} - \hat{\mathbb{E}}(\mathcal{E}_n))^2}{n(\kappa - 1)\hat A_n^2},$$ where $\hat{\mathbb{E}}(\mathcal{E}_n) = \frac{1}{\kappa}\sum_{s = 1}^{\kappa}\hat{\mathcal{E}}_{n,s}$.
The debiased estimator is given by $$\theta_{n,\kappa}^{\ast} = \hat\theta_n - \hat{\mathbb{E}}(\mathcal{E}_n)/(\sqrt{n}\hat A_n).$$ Let $\hat\psi_{n,s}^{\ast} = \frac{\hat{\mathcal{E}}_{n,s}^{\ast} - \hat{\mathbb{E}}(\mathcal{E}_n^{\ast})}{\hat A_n^{\ast}}$, where $\hat A_n^{\ast}$, $\hat{\mathcal{E}}_{n,s}^{\ast}$, and $\hat{\mathbb{E}}(\mathcal{E}_n^{\ast})$ are defined as $\hat A_n$, $\hat{\mathcal{E}}_{n,s}$, and $\hat{\mathbb{E}}(\mathcal{E}_n)$, respectively, with the difference that they are computed using $\theta_{n,\kappa}^{\ast}$ and not $\hat\theta_n$. We can estimate the asymptotic CDF of $\sqrt{n}(\theta_{n,\kappa}^{\ast} - \theta_0)$ by the CDF of the sample: $\{\hat\psi_{n,s}^{\ast}, ~ s = 1, ~\dots,~ \kappa\}$.

\medskip
\paragraph{DGP C}
DGP C is a Poisson model with a latent covariate that is defined as:
    $$y_i \sim \text{Poisson}(\exp(\theta_{0,1} + \theta_{0,2}p_i)), \quad p_i = \sin^2(\pi z_i), \quad z_i \sim \text{Uniform}[0, ~10], \quad d_i \sim \text{Bernoulli}(p_i),$$

\noindent where $p_i$ is an unobserved probability and $\boldsymbol{\theta}_0 = (\theta_{0,1}, ~\theta_{0,2})^{\prime} = (0.5,~ 2)^{\prime}$. The practitioner observes the pairs $(y_i, ~z_i)$ for all $i$ but only observes $d_i$ for a representative subsample of size $n^{\ast} \leq n$. In the first stage, $p_i = \mathbb{E}(d_i|z_i)$ is estimated using a generalized additive model (GAM) of $d_i$ on $z_i$ in the subsample of size $n^{\ast}$ where $d_i$ is observed. The GAM involves approximating $p_i$ by piecewise polynomial functions of $z_i$ \citepoa[see][]{hastie2017generalizedoa}. We consider cubic polynomial functions on the intervals $[0, ~0.5], ~\dots,~ [9.5, ~10]$. This approach can thus be regarded as an OLS regression of $d_i$ on numerous explanatory variables called bases, which are computed from $z_i$. We can write $\hat p_i = h(z_i,~ \boldsymbol{\hat \gamma}_n)$, where $h$ is a piecewise cubic polynomial function and $\boldsymbol{\hat \gamma}_n$ is the OLS estimator. The regression results can be used to compute $\hat{p}_i$ for any $i$ in the full sample because we observe $z_i$ of all $i$. The estimator of the asymptotic distribution of $\boldsymbol{\hat \gamma}_n$ is a normal distribution with mean $\boldsymbol{\hat \gamma}_n$ and the standard OLS variance denoted $\hat{\mathbb{V}}(\hat{\theta}_n)$.

In the second stage, we perform a maximum likelihood (ML) estimation by assuming that $y_i$ follows a Poisson distribution with mean $\exp(\boldsymbol{\hat{\beta}}_{n,i}^\prime\boldsymbol{\theta})$, where $\boldsymbol{\hat{\beta}}_{n,i} = (1, ~\hat{p}_i)$. The objective function is thus given by $Q_n(\theta_0, ~\mathbf{y}_n, ~\mathbf{\hat{B}}_n) = \frac{1}{n}\sum_{i = 1}^n \big(y_i \boldsymbol{\hat{\beta}}_{n,i}^\prime\boldsymbol{\theta} - \exp(\boldsymbol{\hat{\beta}}_{n,i}^\prime\boldsymbol{\theta})\big)$ and $\boldsymbol{\dot q}_{n}(\mathbf{y}_n,~\mathbf{\hat{B}}_n) = \frac{1}{\sqrt{n}}\sum_{i = 1}^n \big(y_i  - \exp( \boldsymbol{\hat{\beta}}_{n,i}^\prime \boldsymbol{\theta}_0)\big) \boldsymbol{\hat{\beta}}_{n,i}$. Therefore, 
\begin{align}
    &\textstyle\mathbf{\hat A}_n = \frac{1}{n}\sum_{i = 1}^n\exp\big(\boldsymbol{\hat{\beta}}_{n,i}^{\prime}\boldsymbol{\hat \theta}_n\big)\boldsymbol{\hat{\beta}}_{n,i}\boldsymbol{\hat{\beta}}_{n,i}^{\prime}\,, \quad \hat{\mathbf{V}}_{n} = \frac{1}{n}\sum_{i = 1}^n\exp\big(\boldsymbol{\hat{\beta}}_{n,i}^\prime \boldsymbol{\hat \theta}_n\big) \boldsymbol{\hat{\beta}}_{n,i}\boldsymbol{\hat{\beta}}_{n,i}^{\prime}, \quad \text{and}\nonumber\\
    &\textstyle \hat{\mathcal{E}}_{n,s} = \frac{1}{\sqrt{n}}\sum_{i=1}^n\big(\exp( \boldsymbol{\hat{\beta}}_{n,i}^{(s)^\prime} \boldsymbol{\hat \theta}_n)  - \exp( \boldsymbol{\hat{\beta}}_{n,i}^{(s)\prime} \boldsymbol{\hat \theta}_n)\big) \boldsymbol{\hat{\beta}}_{n,i}^{(s)},\nonumber
\end{align}
where $\boldsymbol{\hat{\beta}}_{n,i}^{(s)} = h(z_i,~\boldsymbol{\hat\gamma}_n^{(s)})$, for $s = 1, ~\dots,~ \kappa$, and $\boldsymbol{\hat \gamma}_n^{(1)}~\dots,~\boldsymbol{\hat \gamma}_n^{(\kappa)}\overset{i.i.d}{\sim} N(\boldsymbol{\hat\gamma}_n, ~\hat{\mathbb{V}}(\boldsymbol{\hat\gamma}_n))$. 

Let $\boldsymbol{\hat\psi}_{n,s} =  \mathbf{\hat A}_n^{-1}\hat{\mathbf{V}}_{n}^{1/2}\boldsymbol{\zeta}_{s} + \mathbf{\hat A}_n^{-1}\hat{\mathcal{E}}_{n,s}$, where $\boldsymbol{\zeta}_1$, ~\dots,~ $\boldsymbol{\zeta}_{\kappa} \overset{i.i.d}{\sim} N(0, ~\boldsymbol{I}_{K_{\theta}})$.
The asymptotic CDF of $\sqrt{n}(\boldsymbol{\hat\theta}_n - \boldsymbol{\theta}_0)$ can be estimated by the CDF of the sample: $\{\boldsymbol{\hat\psi}_{n,s}, ~ s = 1, ~\dots,~ \kappa\}$.    The asymptotic variance of $\boldsymbol{\hat{\theta}}_n$ is estimated by $$\hat{\mathbb{V}}(\boldsymbol{\hat{\theta}}_n)= \displaystyle\frac{\mathbf{\hat A}_n^{-1}\boldsymbol{\hat \Sigma}_n^{\kappa}\mathbf{\hat A}_n^{-1}}{n},$$ where
$\boldsymbol{\hat \Sigma}_n^{\kappa} =\hat{\mathbf{V}}_{n} + \frac{1}{\kappa - 1} \sum_{s = 1}^{\kappa} (\hat{\mathcal{E}}_{n,s} - \boldsymbol{\hat \Omega}_{n}^{\kappa}) (\hat{\mathcal{E}}_{n,s} - \boldsymbol{\hat \Omega}_{n}^{\kappa})^{\prime}$ and $\boldsymbol{\hat \Omega}_{n}^{\kappa} =  \frac{1}{\kappa}\sum_{s = 1}^{\kappa} \hat{\mathcal{E}}_{n,s}$. The debiased estimator is given by $$\boldsymbol{\theta}_{n,\kappa}^{\ast} = \boldsymbol{\hat\theta}_n - \mathbf{\hat A}_n^{-1}\boldsymbol{\hat \Omega}_{n}^{\kappa}/\sqrt{n}.$$ Let $\boldsymbol{\hat\psi}_{n,s}^{\ast} = (\mathbf{\hat A}_n^{\ast})^{-1}(\mathbf{\hat V}_{n}^{\ast})^{1/2}\boldsymbol{\zeta}_{s} + (\mathbf{\hat A}_n^{\ast})^{-1}(\hat{\mathcal{E}}_{n,s}^{\ast} -\boldsymbol{\hat \Omega}_{n}^{\ast\kappa})$, where $\mathbf{\hat A}_n^{\ast}$, $\mathbf{\hat V}_{n}^{\ast}$, $\hat{\mathcal{E}}_{n,s}^{\ast}$, and $\boldsymbol{\hat \Omega}_{n}^{\ast\kappa}$ are respectively defined as $\mathbf{\hat A}_n$, $\mathbf{\hat V}_{n}$, $\hat{\mathcal{E}}_{n,s}$, and $\boldsymbol{\hat \Omega}_{n}^{\kappa}$, with the difference that they are computed using $\boldsymbol{\theta}_{n,\kappa}^{\ast}$ and not $\boldsymbol{\hat\theta}_n$. We can estimate the asymptotic CDF of $\sqrt{n}(\boldsymbol{\hat\theta}_n^{\ast} - \boldsymbol{\theta}_0)$ by the CDF of the sample: $\{\boldsymbol{\hat\psi}_{n,s}^{\ast}, ~s = 1, ~\dots,~ \kappa\}$.

\medskip
\paragraph{DGP D}
DGP D is a copula-based multivariate time-series model. We consider $k_n$ returns $y_{1,i}$, ~\dots,~ $y_{k_n, i}$, where $i$ is time and $k_n \geq 2$. Each $y_{p,i}$, for $p = 2, ~\dots,~ k_n$, follows an AR(1)-GARCH(1, 1) model, such that 
$$y_{p,i} = \phi_{p,0} + \phi_{p,1}y_{p,i-1}+  \sigma_{p,i}\varepsilon_{p,i}, \quad \sigma_{p,i}^2 = \beta_{p,0}+\beta_{p,1}\sigma_{p,i-1}^2\varepsilon_{p,i-1}^2+\beta_{p,2}\sigma_{p,i-1}^2,$$
where $\phi_{p,0} = 0$, $\phi_{p,i-1} = 0.4$, $\beta_{p,0} = 0.05$, $\beta_{p,1} = 0.05$, $\beta_{p,2} = 0.9$, and $\varepsilon_{p,i}$ follows a standardized Student distribution of degree-of-freedom $\nu_p = 6$.  We account for the correlation between the returns using the Clayton copula. The joint density function of $y_i = (y_{1,i}, ~\dots,~ y_{p,i})^{\prime}$ conditional on $\mathcal{F}^{i-1}$ is given by $c_i(G_{1,i}(\boldsymbol{\beta}_{0,1}), ~\dots,~ G_{k_n,i}(\boldsymbol{\beta}_{0,k_n}), ~\theta_0)$, where $\boldsymbol{\beta}_{0,p} = (\phi_{p,0}, ~\phi_{p,1}, ~\beta_{p,0}, ~\beta_{p,1}, ~\beta_{p,2}, ~\nu_p)^{\prime}$, $G_{p,i}(\boldsymbol{\beta}_{0,p})$ is the CDF of $y_{p,i}$ conditional on $\mathcal{F}^{i-1}$, and $c_i$ is the probability density function (PDF) of the $k_n$-dimensional Clayton copula of parameter $\theta_0 = 4$. A multi-stage estimation strategy can be used to estimate $\theta_0$. In the first $k_n$ stages, we separately estimate each $\boldsymbol{\beta}_{0,p}$  by applying an AR(1)-GARCH(1, 1) model to the sample $y_{p,1}, ~\dots,~ y_{p,n}$.
$$\textstyle\boldsymbol{\hat\beta}_{n,p} = \arg\max_{\boldsymbol{\beta}_{p}} \ell_p:= \frac{1}{n}\sum_{i=1}^{n}\underbrace{\log g_{p,i}(\boldsymbol{\beta}_{p})}_{\ell_{p,i}}\,,\,\,\text{ for }\, p=1,\,\dots,\,k_n\,,$$
where $g_{p,i}(\boldsymbol{\beta}_{0,p})$ is the PDF of $y_{p,i}$ conditional on $\mathcal{F}^{i-1}$. Let $\boldsymbol{\hat\beta}_{n}$ be the estimator $\boldsymbol{\beta}_{0}:=(\boldsymbol{\beta}_{0,1}^\prime, ~\dots,~ \boldsymbol{\beta}_{0,k_n}^\prime)^\prime$. The estimator of the asymptotic distribution of $\boldsymbol{\hat\beta}_{n}$  is a normal distribution with mean $\boldsymbol{\hat\beta}_{n}$ and variance given by $\hat{\mathbb{V}}(\boldsymbol{\hat\beta}_{n}) = \frac{1}{n}\mathbf{\hat{H}}_n^{-1}\mathbf{\hat{J}}_n\mathbf{\hat{H}}_n^{-1}$, where $$\textstyle\mathbf{\hat{H}}_n = \frac{1}{n}\sum_{i=1}^{n}\diag(\frac{\partial^2}{\partial\boldsymbol{\beta}^\prime\partial\boldsymbol{\beta}}\ell_{1,i},\, \dots, \,\frac{\partial^2}{\partial\boldsymbol{\beta}^\prime\partial\boldsymbol{\beta}}\ell_{k_n,i})\quad \text{and} \quad \textstyle\mathbf{\hat{J}}_n = \mathbb{V}_{\text{HAC}}(\frac{1}{\sqrt{n}}\sum_{i=1}^{n}(\partial_{\boldsymbol{\beta}^\prime}\ell_{1,i},\,\dots,\,\partial_{\boldsymbol{\beta}^\prime}\ell_{k_n,i})^\prime)\,.$$
The notation $\mathbb{V}_{\text{HAC}}$ is the heteroskedasticity and autocorrelation consistent (HAC) covariance matrix to account for the serial correlation \citepoa[see][]{andrews1991heteroskedasticityoa}. In the HAC approach, we use the quadratic spectral kernel and set the bandwidth to $\frac{3}{4}n^{1/3}$. The gradient and the Hessian of the likelihood $\ell_{p,i}$ do not have a closed form. Fortunately, they can be approximated numerically in most statistical software. 

In the last stage, we estimate $\theta_0$ by ML after replacing $\boldsymbol{\beta}_{0}$ in the density function of $y_i$ with $\boldsymbol{\hat{\beta}}_n$. Let $q_{n,i}(\theta, ~\boldsymbol{\hat\beta}_{n}) = \log\big(c_i(G_{1,i}(\boldsymbol{\hat \beta}_{n,1}),\, \dots,\, G_{k_n,i}(\boldsymbol{\hat \beta}_{n,k_n}), ~\theta)\big)$, where $\log c_i(u_1,~\dots,~u_{k_n},\theta) = \sum_{p=1}^{k_n-1}\log(p\theta+1) -(\theta+1) \sum_{p=1}^{k_n}\log u_p - (k_n+\tfrac{1}{\theta})\log(\sum_{p=1}^{k_n}u_p^{-\theta} - k_n + 1)$ .
The objective function is $Q_n(\theta, ~\boldsymbol{\hat\beta}_{n}) = \frac{1}{n}\sum_{i = 1}^n q_{n, i}(\theta, ~\boldsymbol{\hat\beta}_{n})$. To compute $\partial_{\theta}q_{n, i}(\theta_0, ~\boldsymbol{\hat\beta}_{n})$ and $\frac{\partial^2}{\partial\theta^2}q_{n, i}(\theta_0, ~\boldsymbol{\hat\beta}_{n})$, we need the first and second derivatives of $\log c_i(u_1,~\dots,~u_{k_n},\theta)$ that can be expressed as follows:
\begin{align*}
    &\partial_{ \theta}\log c_i(u_1,~\dots,~u_{k_n},\theta) = \sum_{p=1}^{k_n-1}\frac{p}{p\theta+1}-\sum_{p=1}^{k_n}\log u_p + \frac{\log\left(\sum_{p=1}^{k_n}u_p^{-\theta} - k_n + 1\right)}{\theta^2} \\
    & \quad \quad + \frac{(k_n+\tfrac{1}{\theta})\sum_{p=1}^{k_n}u_p^{-\theta}\log u_p}{\sum_{p=1}^{k_n}u_p^{-\theta} - k_n + 1}\\
    &\frac{\partial^2}{\partial \theta^2}\log c_i(u_1,~\dots,~u_{k_n},\theta) = \frac{(k_n+\tfrac{1}{\theta})\left(\sum_{p=1}^{k_n}u_p^{-\theta}\log u_p\right)^2}{\left(\sum_{p=1}^{k_n}u_p^{-\theta} - k_n + 1\right)^2} - \frac{(k_n+\tfrac{1}{\theta})\sum_{p=1}^{k_n}u_p^{-\theta}\left(\log u_p\right)^2}{\sum_{p=1}^{k_n}u_p^{-\theta} - k_n + 1} \\
    &\quad\quad -\sum_{p=1}^{k_n}\left(\frac{p}{p\theta+1}\right)^2 - \frac{2\sum_{p=1}^{k_n}u_p^{-\theta}\log u_p}{\theta^2\left(\sum_{p=1}^{k_n}u_p^{-\theta} - k_n + 1\right)} - \frac{2\log\left(\sum_{p=1}^{k_n}u_p^{-\theta} - k_n + 1\right)}{\theta^3}\,. 
\end{align*}
We define the following expressions: $\hat A_n = \frac{1}{n}\sum_{i = 1}^n \frac{\partial^2}{\partial\theta}q_{n, i}(\hat{\theta}_n, ~\boldsymbol{\hat\beta}_{n})$, \quad \text{and} \\$\quad \hat{\mathcal{E}}_{n,s} = \frac{1}{\sqrt{n}}\sum_{i = 1}^n \partial_{\theta}q_{n, i}(\hat{\theta}_n,\boldsymbol{\hat{\beta}}_{n}^{(s)})$,  where $\boldsymbol{\hat{\beta}}_{n}^{(1)}, ~\dots, ~\boldsymbol{\hat{\beta}}_{n}^{(\kappa)} \overset{i.i.d}{\sim}  N(\boldsymbol{\hat{\beta}}_{n}, ~\hat{\mathbb{V}}(\boldsymbol{\hat\beta}_{n}))$.

The asymptotic distribution of $\sqrt{n}(\hat\theta_n - \theta_0)$ is normal with a mean that can be estimated by $\hat{\mathbb{E}}(\mathcal{E}_n)/\hat A_n$, where $\hat{\mathbb{E}}(\mathcal{E}_n) = \frac{1}{\kappa}\sum_{s = 1}^{\kappa}\hat{\mathcal{E}}_{n,s}$. The asymptotic variance can be estimated by jointly taking the first- and second-stage estimators as an extremum estimator; i.e., from the asymptotic variance of $(\boldsymbol{\hat\beta}_{n}^{\prime}, ~ \hat\theta_n)^{\prime}$ \citepoa[see][Section 6.3]{newey1994largeoa}. 

The debiased estimator is given by $\theta_{n,\kappa}^{\ast} = \hat\theta_n - \hat{\mathbb{E}}(\mathcal{E}_n)/(\sqrt{n}\hat A_n)$. 
The asymptotic distribution of $\sqrt{n}(\boldsymbol{\theta}_{n,\kappa}^{\ast} - \theta_0)$ is normal with a zero mean and the same variance as $\sqrt{n}(\hat\theta_n - \theta_0)$.

\subsection{Estimates of the Asymptotic Distribution of the Debiased Estimators}\label{sm:simuR}
This section presents the estimates of the asymptotic CDF of $\Delta_{n,\kappa}^{\ast} := \sqrt{n}(\boldsymbol{\theta}_{n,\kappa}^{\ast} - \theta_0)$, where $\boldsymbol{\theta}_{n,\kappa}^{\ast}$ is the debiased estimator. In contrast to the case of the classical plug-in estimator, the true sample CDFs are asymptotically centered at zero because $\mathbb{E}(\Delta_{n,\kappa}^{\ast})$ converges to zero asymptotically. Overall, the results demonstrate that the estimator of the CDF of $\Delta_{n,\kappa}^{\ast}$, as outlined in Theorem \ref{theo:debias}, performs well.

\begin{figure}[!htbp]
    \begin{center}
        DGP A
    \end{center}
    \hspace{-1.9cm}\includegraphics[scale = 0.75]{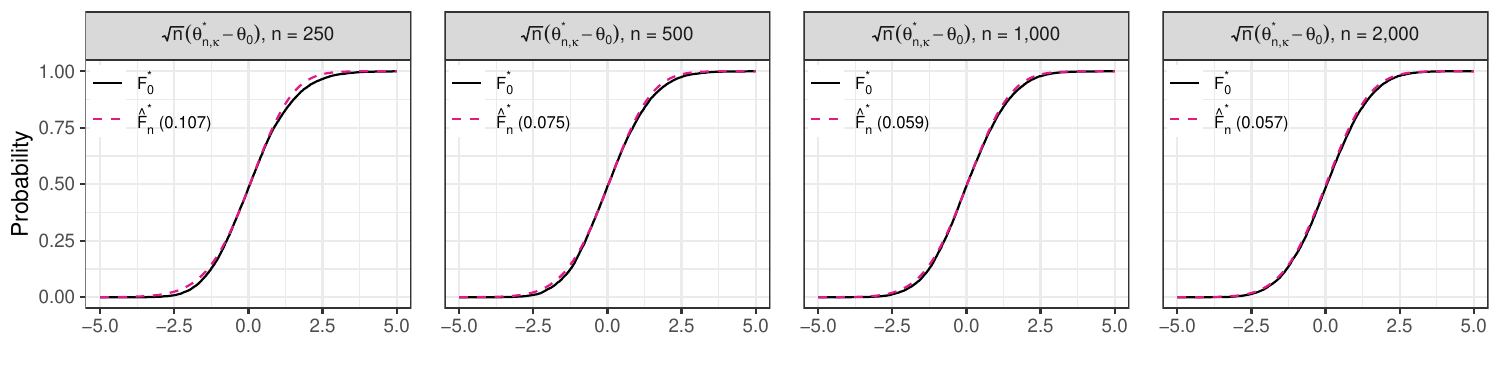}

    \vspace{-0.5cm}
    \begin{center}
        DGP B
    \end{center}
    \hspace{-1.9cm}\includegraphics[scale = 0.75]{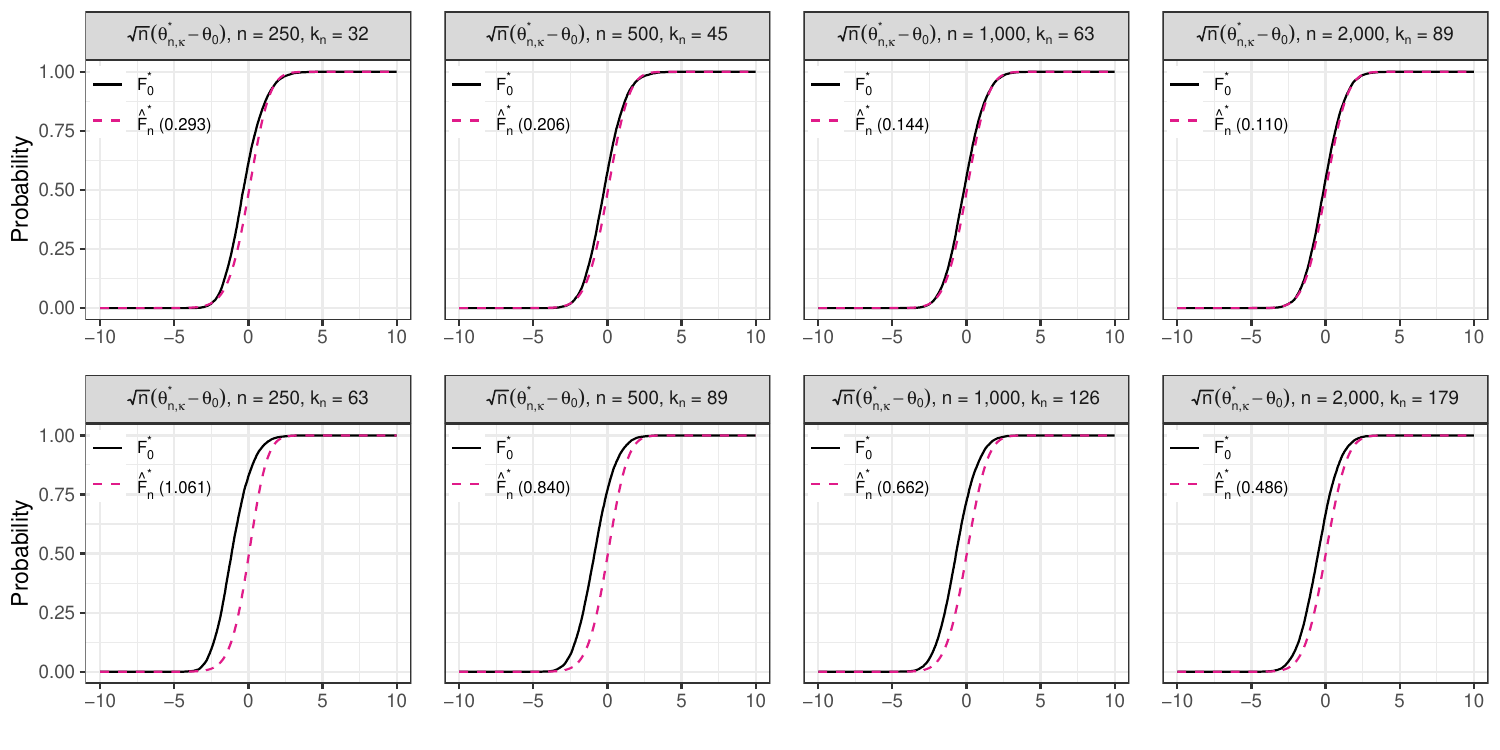}
    \caption{Monte Carlo Simulations: Estimates of asymptotic CDFs (DGPs A and B)}
    \label{fig:simuRAB}

    \vspace{-0.2cm}
    \justify{\footnotesize This figure displays average estimates of the asymptotic CDF of $\sqrt{n}(\boldsymbol{\theta}_{n,\kappa}^{\ast} - \theta_0)$ for DGPs A and B. $F_0^{\ast}$ represents the true sample CDF whereas $\hat F_n^{\ast}$ corresponds to the average estimate of the CDF using our simulation approach. The $L_1$-Wasserstein distance between each estimated CDF and $F_0^{\ast}$ is enclosed in parentheses.}
\end{figure}

\begin{figure}[!htbp]
    \vspace{-1.5cm}
    \begin{center}
        DGP C
    \end{center}
    \hspace{-1.9cm}\includegraphics[scale = 0.75]{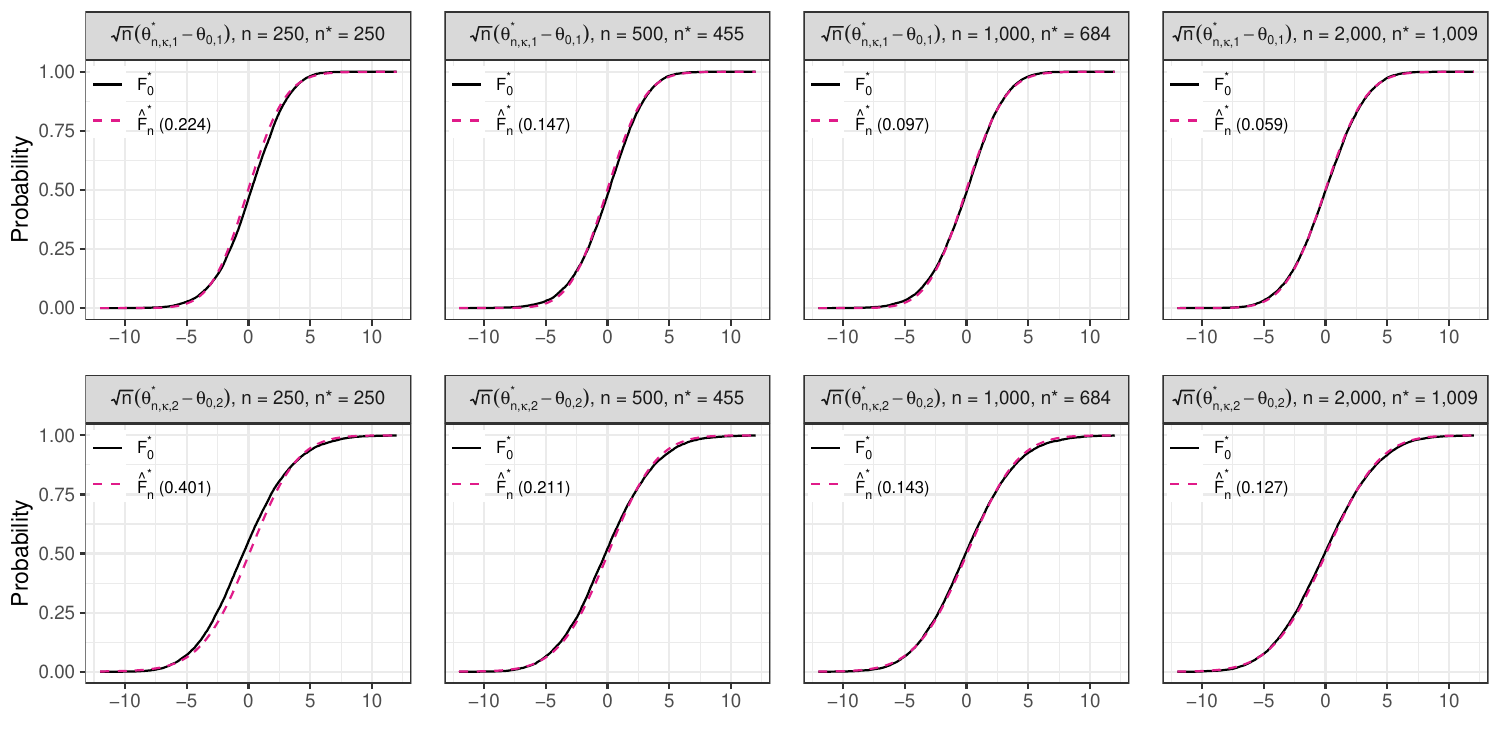}

    \vspace{-0.5cm}
    \begin{center}
      DGP D: Mean Correction
    \end{center}
    \hspace{-1.9cm}\includegraphics[scale = 0.75]{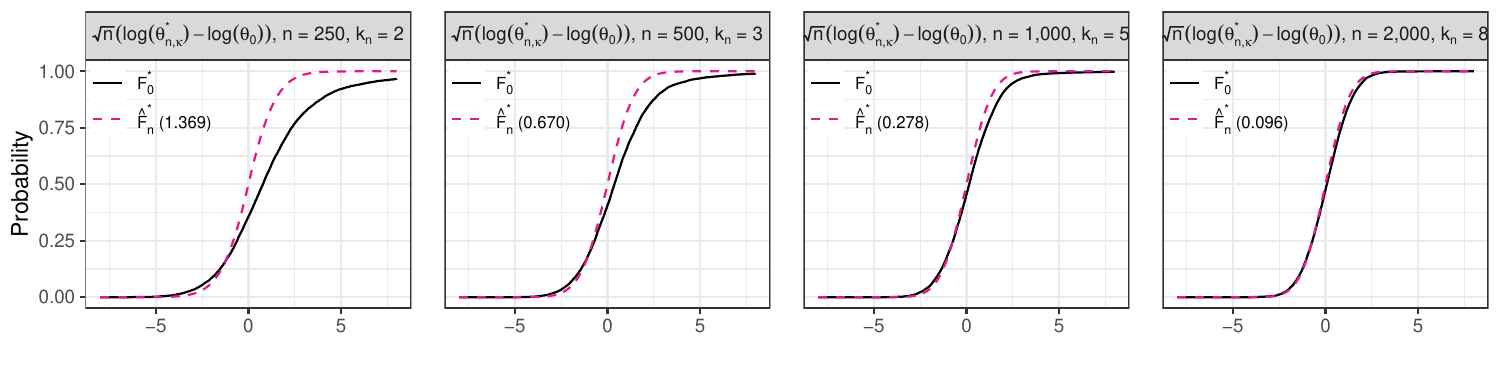}

    \vspace{-0.5cm}
    \begin{center}
      DGP D: Median Correction
    \end{center}
    \hspace{-1.9cm}\includegraphics[scale = 0.75]{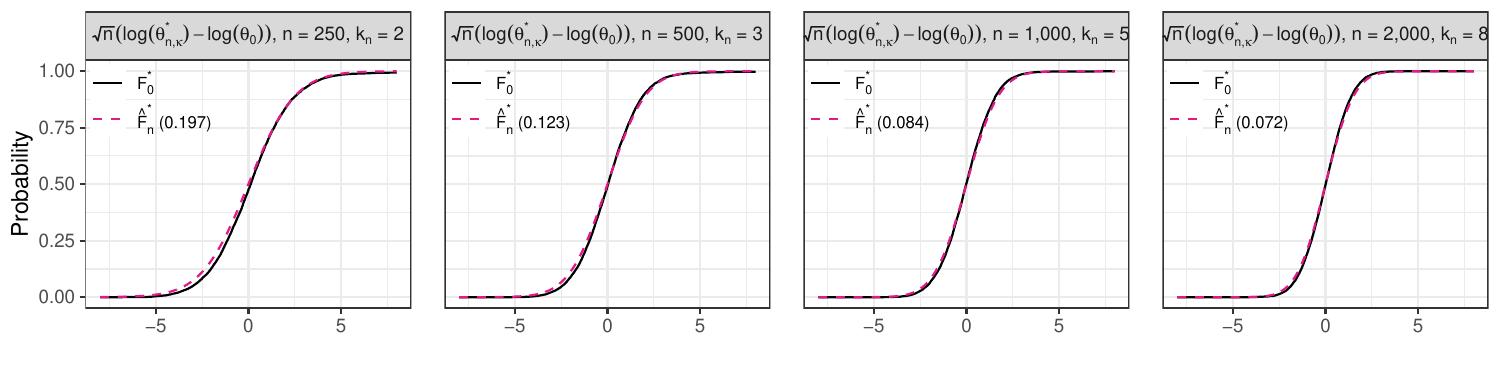}
    \caption{Monte Carlo Simulations: Estimates of asymptotic CDFs (DGPs C and D)}
    \label{fig:simuRCD}

    \vspace{-0.2cm}
    \justify{\footnotesize This figure displays average estimates of the asymptotic CDF of $\sqrt{n}(\boldsymbol{\theta}_{n,\kappa}^{\ast} - \theta_0)$ for DGPs A and B. $F_0^{\ast}$ represents the true sample CDF whereas $\hat F_n^{\ast}$ corresponds to the average estimate of the CDF using our simulation approach. The $L_1$-Wasserstein distance between each estimated CDF and $F_0^{\ast}$ is enclosed in parentheses.}
\end{figure}  

\newpage
\subsection{Confidence Interval (CI) Coverage Rates}\label{sm:crate}
\begin{table}[!ht]
\centering
\caption{Confidence Interval Coverage Rates}
\label{tab:crate}
\begin{adjustbox}{max width=\linewidth}
\begin{threeparttable}
\small
\begin{tabular}{ld{3}d{3}d{3}d{3}ld{3}d{3}d{3}d{3}ld{3}d{3}d{3}d{3}}
\toprule
       & \multicolumn{4}{c}{Two-side 95\% CI} &  & \multicolumn{4}{c}{Lower one-side 95\% CI} &  & \multicolumn{4}{c}{Upper one-side 95\% CI} \\
$n$    & \multicolumn{1}{c}{250}   & \multicolumn{1}{c}{500}    & \multicolumn{1}{c}{1,000}  & \multicolumn{1}{c}{2,000} &  & \multicolumn{1}{c}{250}   & \multicolumn{1}{c}{500}    & \multicolumn{1}{c}{1,000}  & \multicolumn{1}{c}{2,000}   &  & \multicolumn{1}{c}{250}   & \multicolumn{1}{c}{500}    & \multicolumn{1}{c}{1,000}  & \multicolumn{1}{c}{2,000}   \\
\midrule
       & \multicolumn{14}{c}{DGP A}                                                                                                           \\
Nor.   & 0.944  & 0.948  & 0.951    & 0.949   &  & 0.934    & 0.940   & 0.945     & 0.943     &  & 0.958    & 0.958   & 0.956     & 0.954     \\
Sim.   & 0.937  & 0.945  & 0.948    & 0.946   &  & 0.945    & 0.947   & 0.949     & 0.947     &  & 0.940    & 0.945   & 0.948     & 0.948     \\
D-Sim. & 0.939  & 0.946  & 0.948    & 0.946   &  & 0.946    & 0.947   & 0.949     & 0.947     &  & 0.941    & 0.946   & 0.948     & 0.948     \\[1.5ex]
       & \multicolumn{14}{c}{DGP B, $k_n = 2\sqrt{n}$}                                                                                        \\
Nor.   & 0.618  & 0.636  & 0.644    & 0.647   &  & 0.506    & 0.525   & 0.529     & 0.528     &  & 1.000    & 1.000   & 1.000     & 1.000     \\
Sim.   & 0.887  & 0.911  & 0.922    & 0.934   &  & 0.858    & 0.894   & 0.911     & 0.923     &  & 0.959    & 0.958   & 0.954     & 0.955     \\
D-Sim. & 0.921  & 0.932  & 0.940    & 0.947   &  & 0.883    & 0.911   & 0.922     & 0.933     &  & 0.975    & 0.969   & 0.965     & 0.962     \\[1.5ex]
       & \multicolumn{14}{c}{DGP B, $k_n = 4\sqrt{n}$}                                                                                        \\
Nor.   & 0.133  & 0.133  & 0.134    & 0.125   &  & 0.081    & 0.080   & 0.081     & 0.073     &  & 1.000    & 1.000   & 1.000     & 1.000     \\
Sim.   & 0.695  & 0.788  & 0.845    & 0.889   &  & 0.613    & 0.720   & 0.790     & 0.846     &  & 0.989    & 0.986   & 0.978     & 0.972     \\
D-Sim. & 0.757  & 0.831  & 0.879    & 0.913   &  & 0.669    & 0.765   & 0.826     & 0.872     &  & 0.994    & 0.992   & 0.987     & 0.982     \\[1.5ex]
       & \multicolumn{14}{c}{DGP C, $\theta_{0,1}$}                                                                                           \\
Nor.   & 0.857  & 0.885  & 0.899    & 0.902   &  & 0.997    & 0.991   & 0.991     & 0.993     &  & 0.781    & 0.823   & 0.843     & 0.846     \\
Sim.   & 0.919  & 0.928  & 0.935    & 0.940   &  & 0.955    & 0.948   & 0.945     & 0.950     &  & 0.907    & 0.923   & 0.934     & 0.934     \\
D-Sim. & 0.935  & 0.937  & 0.940    & 0.944   &  & 0.967    & 0.955   & 0.950     & 0.953     &  & 0.915    & 0.929   & 0.938     & 0.937     \\[1.5ex]
       & \multicolumn{14}{c}{DGP C, $\theta_{0,2}$}                                                                                           \\
Nor.   & 0.824  & 0.867  & 0.880    & 0.886   &  & 0.736    & 0.794   & 0.811     & 0.819     &  & 0.999    & 0.997   & 0.995     & 0.995     \\
Sim.   & 0.914  & 0.927  & 0.934    & 0.941   &  & 0.901    & 0.921   & 0.932     & 0.931     &  & 0.954    & 0.947   & 0.948     & 0.953     \\
D-Sim. & 0.928  & 0.935  & 0.940    & 0.946   &  & 0.908    & 0.925   & 0.935     & 0.935     &  & 0.966    & 0.955   & 0.952     & 0.954     \\[1.5ex]
       & \multicolumn{14}{c}{DGP D}                                                                                                           \\
Nor.   & 0.954  & 0.935  & 0.903    & 0.898   &  & 0.909    & 0.867   & 0.829     & 0.826     &  & 0.995    & 0.998   & 0.997     & 0.997     \\
Sim.   & 0.866  & 0.909  & 0.943    & 0.962   &  & 0.981    & 0.975   & 0.966     & 0.965     &  & 0.837    & 0.888   & 0.933     & 0.953     \\
D-Sim. & 0.712  & 0.805  & 0.877    & 0.923   &  & 0.944    & 0.940   & 0.935     & 0.945     &  & 0.702    & 0.798   & 0.876     & 0.922     \\[1.5ex]
       & \multicolumn{14}{c}{DGP D – Median correction for the debiased estimator}                                                            \\
D-Sim. & 0.944  & 0.957  & 0.959    & 0.964   &  & 0.970    & 0.965   & 0.957     & 0.962     &  & 0.935    & 0.951   & 0.960     & 0.960    \\\bottomrule
\end{tabular}
\begin{tablenotes}[para,flushleft]
\footnotesize
This table displays the 95\% confidence interval coverage rates obtained using various methods. "Nor." reports the coverage rates for normal CDFs with a zero mean, corresponding to the curves $\hat H_n$ in Figures \ref{fig:simuAB} and \ref{fig:simuCD}. "Sim." indicates the coverage rates using our proposed method, corresponding to the curves $\hat F_n$ in Figures \ref{fig:simuAB} and \ref{fig:simuCD}. "D-Sim" indicates the coverage rates using the debiased estimator, corresponding to the curves $\hat F_n^{\ast}$ in Figures \ref{fig:simuRAB} and \ref{fig:simuRCD}.
\end{tablenotes}
\end{threeparttable}
\end{adjustbox}
\end{table}
\renewcommand{\arraystretch}{1}

\section{Supplementary Materials on the Application}\label{sm:app}
\subsection{Data Summary}\label{sm:app:data}
Our dependent variable is the weekly fast-food consumption frequency, measured by the reported frequency (in days) of fast-food restaurant visits in the past week. We control for 25 observable characteristics in $\mathbf{X}_r$ such as students' gender, grade, race, weekly allowance, and parent's education and occupation. On average, students report consuming fast food 2.35 days per week.  (see Table \ref{tab:data}).

\renewcommand{\arraystretch}{1}
\begin{table}[!htbp]
\small
\centering
\caption{Data Summary}
\label{tab:data}
\begin{threeparttable}
\begin{tabular}{lld{4}d{4}d{4}d{4}}
\toprule
\multicolumn{2}{l}{Statistic}                & \multicolumn{1}{c}{Mean}   & \multicolumn{1}{c}{St. Dev.} & \multicolumn{1}{l}{Min} & \multicolumn{1}{l}{Max} \\ \midrule
\multicolumn{2}{l}{Fast food consumption}          & 2.353  & 1.762    & 0   & 7   \\
\multicolumn{2}{l}{Female}                         & 0.501  & 0.500    & 0   & 1   \\
\multicolumn{2}{l}{Age}                            & 16.628 & 1.554    & 12  & 21  \\
\multicolumn{2}{l}{Hispanic}                       & 0.200  & 0.400    & 0   & 1   \\
\multicolumn{2}{l}{Grade 7--8}                     & 0.100  & 0.300    & 0   & 1   \\
\multicolumn{2}{l}{Grade 9--10}                    & 0.230  & 0.421    & 0   & 1   \\
\multicolumn{2}{l}{Grade 11--12}                   & 0.533  & 0.499    & 0   & 1   \\
\multicolumn{2}{l}{Race (White)}                   &        &          &     &     \\
          & Black                                  & 0.142  & 0.349    & 0   & 1   \\
          & Asian                                  & 0.138  & 0.345    & 0   & 1   \\
          & Other                                  & 0.117  & 0.321    & 0   & 1   \\
\multicolumn{2}{l}{With parents}                   & 0.675  & 0.468    & 0   & 1   \\
\multicolumn{2}{l}{Allowance per week}             & 7.893  & 11.609   & 0   & 95  \\
\multicolumn{2}{l}{Mother Education (High school)} &        &          &     &     \\
          & $<$ High school                        & 0.146  & 0.353    & 0   & 1   \\
          & $>$ High school and not graduated        & 0.173  & 0.378    & 0   & 1   \\
          & $>$ High school and graduated          & 0.226  & 0.418    & 0   & 1   \\
          & Missing                                & 0.124  & 0.330    & 0   & 1   \\
\multicolumn{2}{l}{Father Education (High school)} &        &          &     &     \\
          & $<$ High school                        & 0.124  & 0.329    & 0   & 1   \\
          & $>$ High school and not graduated        & 0.137  & 0.344    & 0   & 1   \\
          & $>$ High school and graduated          & 0.202  & 0.402    & 0   & 1   \\
          & Missing                                & 0.284  & 0.451    & 0   & 1   \\
\multicolumn{2}{l}{Mother Job (None)}              &        &          &     &     \\
          & Professional                           & 0.157  & 0.364    & 0   & 1   \\
          & Other                                  & 0.623  & 0.485    & 0   & 1   \\
          & Missing                                & 0.088  & 0.283    & 0   & 1   \\
\multicolumn{2}{l}{Father Job (None)}              &        &          &     &     \\
          & Professional                           & 0.053  & 0.223    & 0   & 1   \\
          & Other                                  & 0.663  & 0.473    & 0   & 1   \\
          & Missing                                & 0.240  & 0.427    & 0   & 1  \\\bottomrule
\end{tabular}
\begin{tablenotes}[para,flushleft]
\footnotesize This table presents the mean, standard deviation (St. Dev.), minimum, and maximum of the variables used in the empirical application. For the categorical explanatory variables, the level in parentheses is set as the reference level. "With parents" is a dummy variable taking 1 if the student lives with both their mother and father. 
\end{tablenotes}
\end{threeparttable}
\end{table}
\renewcommand{\arraystretch}{1}

\subsection{Estimation and Inference}\label{sm:app:results}
The following table displays the full results of estimations of the peer effect model. 

\renewcommand{\arraystretch}{0.9}
\begin{table}[htbp]
\footnotesize
\centering
\caption{Estimation Results: OLS approach}
\label{tab:appOLS}
\begin{adjustbox}{max width=0.65\linewidth}
\begin{threeparttable}
\begin{tabular}{lld{4}d{4}d{4}d{4}}
\toprule
\multicolumn{2}{l}{Fixed Effects}                                           & \multicolumn{2}{c}{No} & \multicolumn{2}{c}{Yes}  \\
                      &                                                     & \multicolumn{1}{c}{Coef}       & \multicolumn{1}{c}{Sd. Err}     & \multicolumn{1}{c}{Coef}       & \multicolumn{1}{c}{Sd. Err}     \\\midrule
\multicolumn{2}{l}{\textbf{Peer effects:} $\theta_{0,1}$}                            & 0.192      & 0.031     & 0.150       & 0.032     \\\midrule
\multicolumn{2}{l}{\textbf{Individual characteristics:} $\boldsymbol{\theta}_{0,2}$} &            &           &             &           \\
\multicolumn{2}{l}{Female}                                                  & -0.158     & 0.075     & -0.161      & 0.074     \\
\multicolumn{2}{l}{Age}                                                     & 0.109      & 0.040     & 0.083       & 0.040     \\
\multicolumn{2}{l}{Hispanic}                                                & 0.289      & 0.130     & 0.103       & 0.147     \\
\multicolumn{2}{l}{Grade 7--8}                                              & -0.033     & 0.257     & 0.102       & 0.260     \\
\multicolumn{2}{l}{Grade 9--10}                                             & 0.020      & 0.168     & 0.020       & 0.169     \\
\multicolumn{2}{l}{Grade 11--12}                                            & 0.258      & 0.116     & 0.205       & 0.117     \\
\multicolumn{2}{l}{Race (White)}                                            &            &           &             &           \\
                      & Black                                               & 0.075      & 0.145     & -0.085      & 0.167     \\
                      & Asian                                               & 0.302      & 0.144     & 0.134       & 0.167     \\
                      & Other                                               & -0.118     & 0.138     & -0.145      & 0.139     \\
\multicolumn{2}{l}{With parents}                                            & -0.001     & 0.134     & 0.002       & 0.134     \\
\multicolumn{2}{l}{Allowance per week}                                      & 0.008      & 0.003     & 0.006       & 0.003     \\
\multicolumn{2}{l}{Mother Education (High school)}                          &            &           &             &           \\
                      & $<$ High school                                     & 0.101      & 0.119     & 0.084       & 0.119     \\
                      & $>$ High school and non graduated                     & 0.054      & 0.100     & 0.023       & 0.100     \\
                      & $>$ High school and graduated                       & 0.132      & 0.106     & 0.086       & 0.106     \\
                      & Missing                                             & -0.035     & 0.191     & -0.103      & 0.190     \\
\multicolumn{2}{l}{Father Education (High school)}                          &            &           &             &           \\
                      & $<$ High school                                     & -0.251     & 0.128     & -0.244      & 0.127     \\
                      & $>$ High school and non graduated                     & -0.100     & 0.112     & -0.112      & 0.112     \\
                      & $>$ High school and graduated                       & -0.005     & 0.109     & -0.038      & 0.110     \\
                      & Missing                                             & 0.016      & 0.186     & -0.024      & 0.186     \\
\multicolumn{2}{l}{Mother Job (None)}                                       &            &           &             &           \\
                      & Professional                                        & 0.015      & 0.133     & -0.007      & 0.133     \\
                      & Other                                               & 0.065      & 0.101     & 0.061       & 0.101     \\
                      & Missing                                             & 0.376      & 0.227     & 0.441       & 0.226     \\
\multicolumn{2}{l}{Father Job (None)}                                       &            &           &             &           \\
                      & Professional                                        & -0.247     & 0.224     & -0.281      & 0.223     \\
                      & Other                                               & -0.230     & 0.165     & -0.234      & 0.164     \\
                      & Missing                                             & -0.252     & 0.250     & -0.228      & 0.250     \\ \midrule
\multicolumn{2}{l}{\textbf{Contextual peer effects:} $\boldsymbol{\theta}_{0,3}$}    &            &           &             &           \\
\multicolumn{2}{l}{Female}                                                  & 0.044      & 0.120     & -0.001      & 0.121     \\
\multicolumn{2}{l}{Age}                                                     & -0.030     & 0.022     & -0.008      & 0.023     \\
\multicolumn{2}{l}{Hispanic}                                                & -0.091     & 0.195     & -0.189      & 0.203     \\
\multicolumn{2}{l}{Grade 7--8}                                              & -0.268     & 0.274     & -0.281      & 0.277     \\
\multicolumn{2}{l}{Grade 9--10}                                             & -0.189     & 0.209     & -0.107      & 0.209     \\
\multicolumn{2}{l}{Grade 11--12}                                            & -0.028     & 0.189     & 0.028       & 0.188     \\
\multicolumn{2}{l}{Race (White)}                                            &            &           &             &           \\
                      & Black                                               & 0.165      & 0.198     & 0.129       & 0.206     \\
                      & Asian                                               & -0.118     & 0.187     & -0.122      & 0.196     \\
                      & Other                                               & -0.511     & 0.228     & -0.470      & 0.228     \\
\multicolumn{2}{l}{With parents}                                            & -0.199     & 0.216     & -0.230      & 0.218     \\
\multicolumn{2}{l}{Allowance per week}                                      & 0.006      & 0.005     & 0.003       & 0.005     \\
\multicolumn{2}{l}{Mother Education (High school)}                          &            &           &             &           \\
                      & $<$ High school                                     & 0.417      & 0.196     & 0.359       & 0.196     \\
                      & $>$ High school and non graduated                     & -0.163     & 0.170     & -0.226      & 0.171     \\
                      & $>$ High school and graduated                       & -0.065     & 0.178     & -0.189      & 0.180     \\
                      & Missing                                             & -0.288     & 0.373     & -0.386      & 0.372     \\
\multicolumn{2}{l}{Father Education (High school)}                          &            &           &             &           \\
                      & $<$ High school                                     & -0.090     & 0.216     & -0.111      & 0.217     \\
                      & $>$ High school and non graduated                     & 0.140      & 0.176     & 0.084       & 0.179     \\
                      & $>$ High school and graduated                       & 0.093      & 0.177     & 0.052       & 0.182     \\
                      & Missing                                             & 0.289      & 0.324     & 0.278       & 0.324     \\
\multicolumn{2}{l}{Mother Job (None)}                                       &            &           &             &           \\
                      & Professional                                        & -0.217     & 0.223     & -0.246      & 0.224     \\
                      & Other                                               & -0.249     & 0.171     & -0.272      & 0.172     \\
                      & Missing                                             & -0.193     & 0.427     & -0.105      & 0.428     \\
\multicolumn{2}{l}{Father Job (None)}                                       &            &           &             &           \\
                      & Professional                                        & 0.484      & 0.370     & 0.111       & 0.376     \\
                      & Other                                               & 0.368      & 0.268     & 0.238       & 0.274     \\
                      & Missing                                             & 0.044      & 0.424     & -0.116      & 0.430     \\\bottomrule
\end{tabular}
\begin{tablenotes}[para,flushleft]
\footnotesize For the categorical variables, the level in parentheses is set as the reference level.
\end{tablenotes}
\end{threeparttable}
\end{adjustbox}
\end{table}

\begin{table}[htbp]
\footnotesize
\centering
\caption{Estimation Results: Classical and Optimal GMM approaches}
\label{tab:appGMM}
\begin{adjustbox}{max width=\linewidth}
\begin{threeparttable}
\begin{tabular}{lld{4}d{4}d{4}d{4}d{4}d{4}d{4}d{4}}
\toprule
 Model                    &                                                      & \multicolumn{2}{c}{CIV }   & \multicolumn{2}{c}{CIV}                       & \multicolumn{2}{c}{OIV}         & \multicolumn{2}{c}{OIV}                 \\
\multicolumn{2}{l}{Fixed Effects}                                           & \multicolumn{2}{c}{No} & \multicolumn{2}{c}{Yes} & \multicolumn{2}{c}{No} & \multicolumn{2}{c}{Yes} \\
                     &                                                      & \multicolumn{1}{c}{Coef}       & \multicolumn{1}{c}{Sd. Err}     & \multicolumn{1}{c}{Coef}       & \multicolumn{1}{c}{Sd. Err}     & \multicolumn{1}{c}{Coef}       & \multicolumn{1}{c}{Sd. Err}     & \multicolumn{1}{c}{Coef}       & \multicolumn{1}{c}{Sd. Err}     \\ \midrule
\multicolumn{2}{l}{\textbf{Peer effects:} $\theta_{0,1}$}                            & 0.149      & 0.160     & 0.081       & 0.169     & -0.065     & 0.287     & 0.016       & 0.208     \\\midrule
\multicolumn{2}{l}{\textbf{Individual characteristics:} $\boldsymbol{\theta}_{0,2}$} &            &           &             &           &            &           &             &           \\
\multicolumn{2}{l}{Female}                                                  & -0.156     & 0.075     & -0.158      & 0.075     & -0.145     & 0.077     & -0.154      & 0.075     \\
\multicolumn{2}{l}{Age}                                                     & 0.111      & 0.040     & 0.085       & 0.040     & 0.121      & 0.042     & 0.087       & 0.041     \\
\multicolumn{2}{l}{Hispanic}                                                & 0.288      & 0.130     & 0.095       & 0.148     & 0.284      & 0.132     & 0.088       & 0.149     \\
\multicolumn{2}{l}{Grade 7--8}                                              & -0.033     & 0.258     & 0.105       & 0.260     & -0.035     & 0.261     & 0.108       & 0.261     \\
\multicolumn{2}{l}{Grade 9--10}                                             & 0.019      & 0.168     & 0.018       & 0.169     & 0.017      & 0.171     & 0.016       & 0.169     \\
\multicolumn{2}{l}{Grade 11--12}                                            & 0.258      & 0.116     & 0.203       & 0.117     & 0.261      & 0.118     & 0.201       & 0.117     \\
\multicolumn{2}{l}{Race (White)}                                            &            &           &             &           &            &           &             &           \\
                     & Black                                                & 0.073      & 0.145     & -0.091      & 0.168     & 0.066      & 0.147     & -0.097      & 0.168     \\
                     & Asian                                                & 0.305      & 0.144     & 0.133       & 0.167     & 0.318      & 0.147     & 0.132       & 0.167     \\
                     & Other                                                & -0.123     & 0.139     & -0.153      & 0.141     & -0.144     & 0.142     & -0.160      & 0.142     \\
\multicolumn{2}{l}{With parents}                                            & -0.008     & 0.136     & -0.006      & 0.135     & -0.042     & 0.143     & -0.015      & 0.137     \\
\multicolumn{2}{l}{Allowance per week}                                      & 0.008      & 0.003     & 0.006       & 0.003     & 0.008      & 0.003     & 0.006       & 0.003     \\
\multicolumn{2}{l}{Mother Education (High school)}                          &            &           &             &           &            &           &             &           \\
                     & $<$ High school                                      & 0.102      & 0.120     & 0.085       & 0.119     & 0.110      & 0.121     & 0.087       & 0.119     \\
                     & $>$ High school and not graduated                    & 0.057      & 0.100     & 0.025       & 0.100     & 0.070      & 0.103     & 0.026       & 0.101     \\
                     & $>$ High school and graduated                        & 0.131      & 0.106     & 0.082       & 0.107     & 0.128      & 0.108     & 0.079       & 0.107     \\
                     & Missing                                              & -0.045     & 0.194     & -0.121      & 0.195     & -0.094     & 0.204     & -0.138      & 0.198     \\
\multicolumn{2}{l}{Father Education (High school)}                          &            &           &             &           &            &           &             &           \\
                     & $<$ High school                                      & -0.252     & 0.128     & -0.246      & 0.127     & -0.260     & 0.130     & -0.247      & 0.128     \\
                     & $>$ High school and not graduated                    & -0.100     & 0.112     & -0.113      & 0.112     & -0.101     & 0.113     & -0.114      & 0.112     \\
                     & $>$ High school and graduated                        & -0.006     & 0.109     & -0.040      & 0.111     & -0.011     & 0.111     & -0.042      & 0.111     \\
                     & Missing                                              & 0.019      & 0.187     & -0.020      & 0.186     & 0.037      & 0.190     & -0.017      & 0.186     \\
\multicolumn{2}{l}{Mother Job (None)}                                       &            &           &             &           &            &           &             &           \\
                     & Professional                                         & 0.011      & 0.134     & -0.014      & 0.134     & -0.009     & 0.137     & -0.021      & 0.135     \\
                     & Other                                                & 0.066      & 0.101     & 0.061       & 0.101     & 0.069      & 0.103     & 0.061       & 0.102     \\
                     & Missing                                              & 0.381      & 0.228     & 0.452       & 0.228     & 0.404      & 0.232     & 0.463       & 0.229     \\
\multicolumn{2}{l}{Father Job (None)}                                       &            &           &             &           &            &           &             &           \\
                     & Professional                                         & -0.242     & 0.224     & -0.278      & 0.224     & -0.219     & 0.228     & -0.274      & 0.224     \\
                     & Other                                                & -0.225     & 0.165     & -0.228      & 0.165     & -0.203     & 0.169     & -0.222      & 0.166     \\
                     & Missing                                              & -0.255     & 0.251     & -0.230      & 0.250     & -0.267     & 0.254     & -0.231      & 0.251     \\\midrule
\multicolumn{2}{l}{\textbf{Contextual peer effects:} $\boldsymbol{\theta}_{0,3}$}    &            &           &             &           &            &           &             &           \\
\multicolumn{2}{l}{Female}                                                  & 0.037      & 0.123     & -0.013      & 0.124     & 0.003      & 0.130     & -0.024      & 0.126     \\
\multicolumn{2}{l}{Age}                                                     & -0.025     & 0.029     & 0.001       & 0.031     & 0.000      & 0.040     & 0.009       & 0.034     \\
\multicolumn{2}{l}{Hispanic}                                                & -0.075     & 0.203     & -0.169      & 0.209     & 0.002      & 0.222     & -0.150      & 0.213     \\
\multicolumn{2}{l}{Grade 7--8}                                              & -0.264     & 0.274     & -0.273      & 0.278     & -0.246     & 0.278     & -0.266      & 0.279     \\
\multicolumn{2}{l}{Grade 9--10}                                             & -0.186     & 0.209     & -0.099      & 0.210     & -0.174     & 0.212     & -0.092      & 0.211     \\
\multicolumn{2}{l}{Grade 11--12}                                            & -0.011     & 0.199     & 0.056       & 0.201     & 0.073      & 0.222     & 0.083       & 0.207     \\
\multicolumn{2}{l}{Race (White)}                                            &            &           &             &           &            &           &             &           \\
                     & Black                                                & 0.172      & 0.199     & 0.141       & 0.208     & 0.205      & 0.205     & 0.153       & 0.209     \\
                     & Asian                                                & -0.106     & 0.192     & -0.101      & 0.202     & -0.047     & 0.205     & -0.082      & 0.206     \\
                     & Other                                                & -0.519     & 0.230     & -0.483      & 0.230     & -0.560     & 0.237     & -0.495      & 0.232     \\
\multicolumn{2}{l}{With parents}                                            & -0.212     & 0.221     & -0.247      & 0.222     & -0.276     & 0.235     & -0.264      & 0.224     \\
\multicolumn{2}{l}{Allowance per week}                                      & 0.007      & 0.006     & 0.004       & 0.005     & 0.009      & 0.006     & 0.004       & 0.006     \\
\multicolumn{2}{l}{Mother Education (High school)}                          &            &           &             &           &            &           &             &           \\
                     & $<$ High school                                      & 0.420      & 0.197     & 0.359       & 0.196     & 0.438      & 0.200     & 0.359       & 0.197     \\
                     & $>$ High school and not graduated                    & -0.159     & 0.170     & -0.223      & 0.172     & -0.139     & 0.174     & -0.220      & 0.172     \\
                     & $>$ High school and graduated                        & -0.053     & 0.184     & -0.178      & 0.182     & 0.010      & 0.199     & -0.168      & 0.184     \\
                     & Missing                                              & -0.264     & 0.383     & -0.351      & 0.382     & -0.146     & 0.409     & -0.319      & 0.387     \\
\multicolumn{2}{l}{Father Education (High school)}                          &            &           &             &           &            &           &             &           \\
                     & $<$ High school                                      & -0.112     & 0.231     & -0.144      & 0.231     & -0.223     & 0.264     & -0.175      & 0.239     \\
                     & $>$ High school and not graduated                    & 0.124      & 0.186     & 0.054       & 0.193     & 0.043      & 0.209     & 0.025       & 0.201     \\
                     & $>$ High school and graduated                        & 0.082      & 0.181     & 0.035       & 0.186     & 0.029      & 0.192     & 0.020       & 0.189     \\
                     & Missing                                              & 0.282      & 0.325     & 0.265       & 0.326     & 0.251      & 0.331     & 0.254       & 0.327     \\
\multicolumn{2}{l}{Mother Job (None)}                                       &            &           &             &           &            &           &             &           \\
                     & Professional                                         & -0.225     & 0.225     & -0.262      & 0.227     & -0.265     & 0.232     & -0.276      & 0.229     \\
                     & Other                                                & -0.245     & 0.172     & -0.268      & 0.173     & -0.227     & 0.175     & -0.265      & 0.173     \\
                     & Missing                                              & -0.209     & 0.431     & -0.125      & 0.431     & -0.286     & 0.445     & -0.144      & 0.433     \\
\multicolumn{2}{l}{Father Job (None)}                                       &            &           &             &           &            &           &             &           \\
                     & Professional                                         & 0.501      & 0.375     & 0.119       & 0.377     & 0.587      & 0.392     & 0.125       & 0.378     \\
                     & Other                                                & 0.382      & 0.273     & 0.256       & 0.277     & 0.452      & 0.287     & 0.273       & 0.280     \\
                     & Missing                                              & 0.061      & 0.429     & -0.091      & 0.434     & 0.146      & 0.444     & -0.067      & 0.438     \\\bottomrule
\end{tabular}
\begin{tablenotes}[para,flushleft]
\footnotesize For the categorical variables, the level in parentheses is set as the reference level.
\end{tablenotes}
\end{threeparttable}
\end{adjustbox}
\end{table}

\begin{table}[htbp]
\footnotesize
\centering
\caption{Estimation Results: Many Instrument Approaches}
\label{tab:appGMMlarge}
\begin{adjustbox}{max width=\linewidth}
\begin{threeparttable}
\begin{tabular}{lld{4}d{4}d{4}d{4}d{4}d{4}d{4}d{4}}
\toprule
Model                     &                                                      & \multicolumn{2}{c}{IV-MI}   & \multicolumn{2}{c}{IV-MI}                       & \multicolumn{2}{c}{DIV-MI}         & \multicolumn{2}{c}{DIV-MI}                 \\
\multicolumn{2}{l}{Fixed Effects}                                           & \multicolumn{2}{c}{No}    & \multicolumn{2}{c}{Yes}   & \multicolumn{2}{c}{No}     & \multicolumn{2}{c}{Yes} \\
                     &                                                      & \multicolumn{1}{c}{Coef}       & \multicolumn{1}{c}{Sd. Err}     & \multicolumn{1}{c}{Coef}       & \multicolumn{1}{c}{Sd. Err}     & \multicolumn{1}{c}{Coef}       & \multicolumn{1}{c}{Sd. Err}     & \multicolumn{1}{c}{Coef}       & \multicolumn{1}{c}{Sd. Err}     \\ \midrule
\multicolumn{2}{l}{\textbf{Peer effects:} $\theta_{0,1}$}           & 0.276       & 0.063       & 0.208       & 0.067       & 0.300       & 0.063        & 0.218       & 0.067        \\\midrule
\multicolumn{2}{l}{\textbf{Individual characteristics:} $\boldsymbol{\theta}_{0,2}$} &             &             &             &             &             &              &             &              \\
\multicolumn{2}{l}{Female}                                                  & -0.162      & 0.072       & -0.164      & 0.070       & -0.163      & 0.072        & -0.165      & 0.070        \\
\multicolumn{2}{l}{Age}                                                     & 0.105       & 0.038       & 0.082       & 0.038       & 0.104       & 0.038        & 0.082       & 0.038        \\
\multicolumn{2}{l}{Hispanic}                                                & 0.290       & 0.126       & 0.109       & 0.135       & 0.290       & 0.126        & 0.110       & 0.135        \\
\multicolumn{2}{l}{Grade 7--8}                                              & -0.032      & 0.252       & 0.100       & 0.241       & -0.030      & 0.252        & 0.102       & 0.241        \\
\multicolumn{2}{l}{Grade 9--10}                                             & 0.021       & 0.157       & 0.021       & 0.154       & 0.020       & 0.157        & 0.023       & 0.154        \\
\multicolumn{2}{l}{Grade 11--12}                                            & 0.257       & 0.118       & 0.207       & 0.119       & 0.255       & 0.118        & 0.209       & 0.119        \\
\multicolumn{2}{l}{Race (White)}                                            &             &             &             &             &             &              &             &              \\
                     & Black                                                & 0.078       & 0.147       & -0.080      & 0.158       & 0.080       & 0.147        & -0.078      & 0.158        \\
                     & Asian                                                & 0.297       & 0.134       & 0.135       & 0.151       & 0.296       & 0.134        & 0.136       & 0.151        \\
                     & Other                                                & -0.110      & 0.129       & -0.138      & 0.128       & -0.108      & 0.129        & -0.135      & 0.128        \\
\multicolumn{2}{l}{With parents}                                            & 0.012       & 0.128       & 0.010       & 0.127       & 0.016       & 0.129        & 0.011       & 0.128        \\
\multicolumn{2}{l}{Allowance per week}                                      & 0.008       & 0.003       & 0.006       & 0.003       & 0.008       & 0.003        & 0.006       & 0.003        \\
\multicolumn{2}{l}{Mother Education (High school)}                          &             &             &             &             &             &              &             &              \\
                     & $<$ High school                                      & 0.098       & 0.113       & 0.082       & 0.114       & 0.097       & 0.114        & 0.082       & 0.114        \\
                     & $>$ High school and not graduated                    & 0.048       & 0.093       & 0.022       & 0.093       & 0.047       & 0.094        & 0.021       & 0.093        \\
                     & $>$ High school and graduated                        & 0.133       & 0.095       & 0.089       & 0.094       & 0.134       & 0.095        & 0.089       & 0.094        \\
                     & Missing                                              & -0.015      & 0.179       & -0.088      & 0.178       & -0.008      & 0.179        & -0.085      & 0.178        \\
\multicolumn{2}{l}{Father Education (High school)}                          &             &             &             &             &             &              &             &              \\
                     & $<$ High school                                      & -0.248      & 0.114       & -0.243      & 0.112       & -0.247      & 0.114        & -0.244      & 0.112        \\
                     & $>$ High school and not graduated                    & -0.100      & 0.102       & -0.111      & 0.102       & -0.100      & 0.102        & -0.111      & 0.102        \\
                     & $>$ High school and graduated                        & -0.003      & 0.103       & -0.036      & 0.105       & -0.004      & 0.103        & -0.035      & 0.105        \\
                     & Missing                                              & 0.009       & 0.180       & -0.027      & 0.179       & 0.007       & 0.180        & -0.025      & 0.179        \\
\multicolumn{2}{l}{Mother Job (None)}                                       &             &             &             &             &             &              &             &              \\
                     & Professional                                         & 0.023       & 0.122       & -0.001      & 0.124       & 0.024       & 0.122        & 0.002       & 0.124        \\
                     & Other                                                & 0.064       & 0.093       & 0.060       & 0.095       & 0.063       & 0.093        & 0.062       & 0.095        \\
                     & Missing                                              & 0.366       & 0.216       & 0.432       & 0.213       & 0.362       & 0.216        & 0.429       & 0.213        \\
\multicolumn{2}{l}{Father Job (None)}                                       &             &             &             &             &             &              &             &              \\
                     & Professional                                         & -0.256      & 0.192       & -0.284      & 0.196       & -0.261      & 0.192        & -0.286      & 0.196        \\
                     & Other                                                & -0.238      & 0.136       & -0.239      & 0.139       & -0.241      & 0.137        & -0.239      & 0.139        \\
                     & Missing                                              & -0.248      & 0.230       & -0.227      & 0.231       & -0.246      & 0.230        & -0.226      & 0.231        \\\midrule
\multicolumn{2}{l}{\textbf{Contextual peer effects:} $\boldsymbol{\theta}_{0,3}$}    &             &             &             &             &             &              &             &              \\
\multicolumn{2}{l}{Female}                                                  & 0.058       & 0.111       & 0.010       & 0.112       & 0.062       & 0.111        & 0.013       & 0.112        \\
\multicolumn{2}{l}{Age}                                                     & -0.040      & 0.020       & -0.015      & 0.021       & -0.043      & 0.020        & -0.016      & 0.021        \\
\multicolumn{2}{l}{Hispanic}                                                & -0.121      & 0.167       & -0.206      & 0.180       & -0.129      & 0.167        & -0.210      & 0.180        \\
\multicolumn{2}{l}{Grade 7--8}                                              & -0.275      & 0.241       & -0.287      & 0.241       & -0.279      & 0.241        & -0.286      & 0.242        \\
\multicolumn{2}{l}{Grade 9--10}                                             & -0.194      & 0.179       & -0.114      & 0.180       & -0.193      & 0.179        & -0.112      & 0.180        \\
\multicolumn{2}{l}{Grade 11--12}                                            & -0.061      & 0.171       & 0.003       & 0.173       & -0.067      & 0.172        & 0.003       & 0.173        \\
\multicolumn{2}{l}{Race (White)}                                            &             &             &             &             &             &              &             &              \\
                     & Black                                                & 0.152       & 0.190       & 0.119       & 0.197       & 0.147       & 0.190        & 0.113       & 0.197        \\
                     & Asian                                                & -0.141      & 0.173       & -0.139      & 0.181       & -0.149      & 0.173        & -0.145      & 0.181        \\
                     & Other                                                & -0.495      & 0.177       & -0.459      & 0.180       & -0.489      & 0.177        & -0.456      & 0.180        \\
\multicolumn{2}{l}{With parents}                                            & -0.174      & 0.194       & -0.215      & 0.198       & -0.165      & 0.194        & -0.215      & 0.198        \\
\multicolumn{2}{l}{Allowance per week}                                      & 0.005       & 0.005       & 0.003       & 0.005       & 0.005       & 0.005        & 0.003       & 0.005        \\
\multicolumn{2}{l}{Mother Education (High school)}                          &             &             &             &             &             &              &             &              \\
                     & $<$ High school                                      & 0.410       & 0.180       & 0.359       & 0.177       & 0.409       & 0.180        & 0.361       & 0.177        \\
                     & $>$ High school and not graduated                    & -0.171      & 0.147       & -0.229      & 0.146       & -0.173      & 0.147        & -0.228      & 0.146        \\
                     & $>$ High school and graduated                        & -0.090      & 0.161       & -0.199      & 0.166       & -0.096      & 0.162        & -0.199      & 0.167        \\
                     & Missing                                              & -0.335      & 0.351       & -0.415      & 0.354       & -0.345      & 0.352        & -0.419      & 0.354        \\
\multicolumn{2}{l}{Father Education (High school)}                          &             &             &             &             &             &              &             &              \\
                     & $<$ High school                                      & -0.046      & 0.182       & -0.083      & 0.180       & -0.031      & 0.183        & -0.082      & 0.180        \\
                     & $>$ High school and not graduated                    & 0.172       & 0.155       & 0.110       & 0.158       & 0.184       & 0.156        & 0.115       & 0.158        \\
                     & $>$ High school and graduated                        & 0.113       & 0.157       & 0.065       & 0.163       & 0.122       & 0.158        & 0.065       & 0.163        \\
                     & Missing                                              & 0.301       & 0.241       & 0.288       & 0.243       & 0.308       & 0.241        & 0.289       & 0.242        \\
\multicolumn{2}{l}{Mother Job (None)}                                       &             &             &             &             &             &              &             &              \\
                     & Professional                                         & -0.202      & 0.200       & -0.233      & 0.201       & -0.196      & 0.200        & -0.231      & 0.201        \\
                     & Other                                                & -0.256      & 0.150       & -0.276      & 0.153       & -0.254      & 0.150        & -0.275      & 0.153        \\
                     & Missing                                              & -0.163      & 0.385       & -0.088      & 0.384       & -0.150      & 0.386        & -0.082      & 0.384        \\
\multicolumn{2}{l}{Father Job (None)}                                       &             &             &             &             &             &              &             &              \\
                     & Professional                                         & 0.450       & 0.289       & 0.105       & 0.298       & 0.438       & 0.289        & 0.106       & 0.298        \\
                     & Other                                                & 0.341       & 0.187       & 0.223       & 0.199       & 0.332       & 0.187        & 0.219       & 0.198        \\
                     & Missing                                              & 0.010       & 0.322       & -0.137      & 0.332       & 0.000       & 0.322        & -0.142      & 0.331        \\\bottomrule
\end{tabular}
\begin{tablenotes}[para,flushleft]
\footnotesize For the categorical variables, the level in parentheses is set as the reference level.
\end{tablenotes}
\end{threeparttable}
\end{adjustbox}
\end{table}
\renewcommand{\arraystretch}{1}

\clearpage
{\linespread{0.8}
\fontsize{9}{10}\selectfont
\bibliographyoa{Referencesoa}
\bibliographystyleoa{ecta}}
\end{document}